\documentclass[12pt]{article}
\usepackage[utf8]{inputenc}
\usepackage{amsmath}
\usepackage{amsthm}
\usepackage{amssymb}
\usepackage{csquotes}
\usepackage{verbatim}
\usepackage{tikz}
\usetikzlibrary{decorations.pathreplacing}
\usepackage{mathtools}
\usepackage[noend]{algpseudocode}
\usepackage{todonotes}
\usepackage{subfigure}
\usepackage{hyperref}
\usepackage{cleveref}
\usepackage{comment}
\usepackage{algorithm}
\usepackage{xspace}

\usepackage{fullpage}

\usetikzlibrary{trees}

\newtheorem{lemma}{Lemma}
\newtheorem{theorem}{Theorem}
\theoremstyle{definition}
\newtheorem{definition}{Definition}

\newtheorem{observation}{Observation}

\def\N{{\bf N}}

\DeclareUnicodeCharacter{2714}{\checkmark}

\begin{document}
	\title{The Slotted Online One-Sided Crossing Minimization Problem on $2$-Regular Graphs}
	\author{
		Elisabet Burjons
		\and
		Janosch Fuchs
		\and
		Henri Lotze
	}
	%
	\maketitle 
	\begin{abstract}
        In the area of graph drawing, the One-Sided Crossing
        Minimization Problem (OSCM) is defined on a bipartite graph
        with both vertex sets aligned parallel to each other and all edges
        being drawn as straight lines.  The task is to find a
        permutation of one of the node sets such that the total number
        of all edge-edge intersections, called \emph{crossings}, is
        minimized.  Usually, the degree of the nodes of one set is
        limited by some constant $k$, with the problem then
        abbreviated to OSCM-$k$.

        In this work, we study an online variant of this problem, in
        which one of the node sets is already given. The other node
        set and the incident edges are revealed iteratively and each node has to be inserted
        into placeholders, which we call \emph{slots}. The
        goal is again to minimize the number of crossings in the final
        graph. 
        Minimizing crossings in an online way is related to the more empirical field of dynamic graph drawing. 
        Note the \emph{slotted} OSCM problem makes instances harder to solve for an online algorithm 
        but in the offline case it is equivalent to the version without slots.

        We show that the online slotted OSCM-$k$ is not competitive
        for any $k\geq 2$ and subsequently limit the graph class to
        that of $2$-regular graphs, for which we show a lower bound of
        $4/3$ and an upper bound of $5$ on the competitive ratio.
	\end{abstract}

\section{Introduction}

Online algorithms were introduced by Sleator and Tarjan~\cite{DBLP:journals/cacm/SleatorT85}
to solve problems for which the instance is piecewise revealed to an algorithm,
which must make some irrevocable decision before the next element of the instance
is presented.
Online algorithms are classically analyzed using competitive analysis, where
the performance of an online algorithm is compared to that
of an optimal offline algorithm working on the same instance. 
The worst case ratio between any online algorithm and the optimal
offline solution is the competitive ratio of a problem.
For a deeper introduction to online algorithms and competitive analysis
we refer the reader to the reference books~\cite{DBLP:books/daglib/0097013, DBLP:series/txtcs/Komm16}.

In graph drawing problems, given a graph, one usually wants to embed
the graph into some space with limited dimensions. The most common 
and practical examples are on the Euclidean plane.
It is also usual to try to embed such graphs in a way that minimizes the 
number of edges that cross each other, i.e., their depictions overlap 
in a point that is not occupied by a vertex. If a graph can be embedded in
the Euclidean plane without any crossings, we say the graph is planar.
A survey on graph drawing and crossing minimization can be found in~\cite{DBLP:journals/comgeo/BattistaETT94, schaefer2012graph}. 

One common way to depict bipartite graphs is by arranging the 
vertices in each partition on a straight (horizontal) line, making the lines
for the two partition sides parallel. In this scenario, the edges
are drawn vertically from one side of the partition to the other 
as straight segments. 
The problem of minimizing the crossings in this scenario is reduced,
thus, to properly ordering the vertices in each partition.
However, in some practical applications it is enough to restrict ourselves 
to ordering one set of the partition (the free side), while the other set remains fixed (the fixed side).
It is also usual to restrict the degree of the vertices in the free side ~\cite{walter,DBLP:journals/ipl/LiS02}
This (one sided) problem is formally defined as follows. 
	
\begin{definition} Given a bipartite Graph $G = (S \dot\cup V, E)$.
  Let the nodes of $S$ and $V$ be aligned in some ordering on
  straight lines parallel to each other, where $S$ is on the top
  line and $V$ on the bottom line. Let the edges $E$ be drawn as
  straight lines only. Let the degree of the nodes of $S$ be bound
  by some $k \in \N$. The \emph{One-Sided Crossing Minimization
  Problem (OSCM-$k$)} is defined as the problem of finding a total
  ordering of the nodes of $S$ such that the number of resulting
  edge crossings in the graph is minimized.
\end{definition}
	
We will assume that the ordering of $V$ is part of the instance and
fixed, such that we can label and reference the nodes of $V$ with
ascending natural numbers, starting from the ``left''. If $|S| = |V|$, we sometimes speak of nodes ``above'' and ``below'' one another, by assuming that the nodes on both lines are drawn equidistantly.
	
\subsection{Related Work}
The OSCM problem has already been extensively studied in the past under different names, 
such as \emph{bipartite crossing number}~\cite{garey1983crossing,schaefer2012graph}, 
\emph{crossing problem}~\cite{DBLP:journals/algorithmica/EadesW94}, 
\emph{fixed-layer bipartite crossing minimization}~\cite{DBLP:journals/ipl/LiS02} and others. 
Eades and Wormald \cite{DBLP:journals/algorithmica/EadesW94} showed that the OSCM problem
is NP-complete for dense graphs, while Mu{\~n}oz et al.~\cite{walter} showed NP-completeness for sparse graphs. 
Mu{\~n}oz et al. also introduced the OSCM-$k$ and showed that the OSCM-$2$ can be solved optimally using the
barycenter heuristic. 

Li and Stallmann~\cite{DBLP:journals/ipl/LiS02} showed that the approximation ratio of the barycenter heuristic is 
in $\Omega(\sqrt{n})$ on general bipartite graphs and
also proved that OSCM-$k$ admits a tight $k-1$ approximation. 
Nagamochi presented a randomized approximation algorithm for general graphs ~\cite{DBLP:journals/dcg/Nagamochi05} and another approximation algorithm for bipartite graphs of large degree~\cite{DBLP:journals/tcs/Nagamochi05}. 

Further researching the complexity, Dujmovi{\'c} and Whitesides~\cite{DBLP:journals/algorithmica/DujmovicW04} first showed that OSCM is fixed parameter tractable, i.e., it can be solved in $f(k)n^{\mathcal{O}(1)}$,
where the parameter $k$ is the number of crossings. 
The currently best known FPT running time is $\mathcal{O}(3^{\sqrt{2k} + n})$ and was given by Kobayashi and Tamaki~\cite{DBLP:journals/algorithmica/KobayashiT15}.  

To the best of our knowledge, the field of online analysis on crossing minimization is hardly researched.  
A closely related problem arises in the field of graph drawing, called \emph{dynamic graph drawing}. 
Here, the task is to visually arrange a graph that is iteratively expanded over time. 
The visualization follows certain empirical criteria to make the data comprehensible, where crossing minimization is one of these criteria. 
For a survey regarding dynamic graph drawing, see~\cite{shannon2007considerations}. 
Dynamic graph drawing has many applications, for instance, 
Frishman and Tal~\cite{DBLP:journals/tvcg/FrishmanT08} present an algorithm to compute online layouts for a 
sequence of graphs and its application in discussion thread visualization and social network visualization. 
In another example, 
North and Woodhull~\cite{DBLP:conf/gd/NorthW01} focus on hierarchical graph drawing,
a more restricted graph class that needs to be visualized in a tree-like fashion, 
which overlaps with our topic regarding applications. 
While one of the most mentioned applications of the offline OSCM is wire crossing minimization in VLSI this is 
arguably less applicable when looking at an online version of the problem. 
However, the results of an online analysis can be helpful for the application fields of graph drawing, e.g., software visualization, decision support systems and interactive graph editors.

While dynamic graph drawing and online graph problems are similar in that parts of the graph are revealed in an iterative fashion and not previously known, a central difference is that in dynamic graph drawing the manipulation of previous decisions is usually allowed. 
This is not the case in the classical online model. 
Thus, while theory and practice are looking at similar problems, and are following the same goal of aesthetic graph drawings, the methods to achieve this goal are different.

\subsection{Our Contribution}

In this paper, we look at the online version of the OSCM-$k$ problem. 
Observe, that the online version of OSCM-$k$ can be defined in two different ways.
The first version is the online \emph{free} OSCM-$k$, where given a bipartite graph $(S\dot\cup V, E)$, 
an algorithm initially sees a fixed set of vertices $V$, and then, in each
step a request appears for a subset of
vertices $R_i\subseteq V$, which must be made adjacent to a vertex in $S$. 
Thus, after the arrival of the request $R_i$, one has to place a vertex $s_i \in S$ on the top line and adjacent to the vertices in $R_i$. 
In this version, one chooses the partial ordering of $s_i$ with respect to the other vertices already present in $S$.

The online free OSCM-$k$ problem is solvable with a competitive ratio of at most $k-1$, using the same barycenter algorithm as in the offline case~\cite{DBLP:journals/ipl/LiS02}. 

In this paper, we focus on a different version of this problem,
which we call the online \emph{slotted} OSCM-$k$, which is formally defined as follows.

\begin{definition}
  Given a vertex set $V$,
  a request sequence for online slotted OSCM-$k$ is a sequence $R_1,\hdots, R_n$ of subsets
  of $V$, each of size $k$. The set of vertices $S$ is initiated as $S=\{s_1,\hdots, s_n\}$. 
  Initially there are no edges between $S$ and $V$. Once a request $R_j\subseteq V$
  arrives, an online algorithm solving online slotted OSCM-$k$ chooses a vertex $s_i$ 
  without any edges, and places an edge between
  $s_i$ and every vertex in $R_j$. 
  The goal is to minimize total number of crossings.
\end{definition}

The slotted OSCM-$k$ is a model that follows the aesthetic paradigms of the area of dynamic graph drawing, where the so-called \emph{mental map} and human readability is sustained. 
The term \emph{mental map} describes the goal to make current visualization of the graph recognizable in later iterations of the graph. 
Compared to the free OSCM-$k$ no upper or lower bound on the competitive ratio is known. 

We call the vertices $s_i\in S$ {\em slots} moving forward. 
If a request $R_j\subseteq V$ is fulfilled by adding edges between every vertex
in $R_j$ and slot $s_i$ we say that request $R_j$ is {\em assigned} to slot $s_i$.
Moreover, we call a slot $s_i$ {\em unfulfilled} or {\em free} if no request has been satisfied using this slot,
thus the slot has no edges yet. Correspondingly, a {\em fulfilled} slot $s_i$ is a slot in $S$ with edges to 
a subset $R_j\subseteq V$.

Online slotted OSCM-$k$ has the advantage of knowing in advance the number of requests.
However, one has the distinct constraint that, once two consecutive slots are fulfilled, 
the algorithm will not be able to assign any request to a vertex between the fulfilled slots,
as such a vertex does not exist.

We prove, that online slotted OSCM-$k$ is not competitive for any $k\ge 2$
in general graph classes.
However,
if we focus on $2$-regular graphs, we prove that this problem has a constant competitive ratio.
In particular, we prove a lower bound of $4/3$ in this case, and then present an algorithm
with a competitive ratio of at most $5$ as an upper bound. 

\section{Lower Bounds on General Graphs}

We begin by looking at online slotted OSCM-$k$ on general graphs,
and show that for every non-trivial value of $k$, i.e., $k\ge 2$,
there is no algorithm with a constant competitive ratio.

\begin{theorem}\label{thm:oscmunrestricted}
    There is no online algorithm with a constant competitive ratio
    for online slotted OSCM-$k$, for any $k\ge 2$.
\end{theorem}

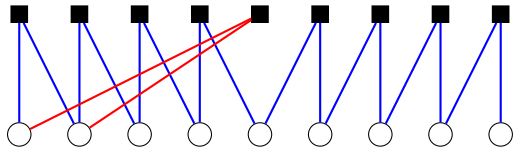
\begin{figure}[h]
	\centering
	\begin{tikzpicture}[node distance=1cm,
	slot/.style={draw,rectangle},
	vertex/.style={draw,circle},
	scale=0.8,every node/.style={scale=0.8}
	]
	\node[vertex] (x_1) at (0,0) {};
	\node[vertex] (x_2) at (1,0) {};
	\node[vertex] (x_3) at (2,0) {};
	\node[vertex] (x_4) at (3,0) {};
	\node[vertex] (x_5) at (4,0) {};
	\node[vertex] (x_6) at (5,0) {};
	\node[vertex] (x_7) at (6,0) {};
	\node[vertex] (x_8) at (7,0) {};
	\node[vertex] (x_9) at (8,0) {};

	\node[slot,fill=black] (s_1) at (0,2) {};
	\node[slot,fill=black] (s_2) at (1,2) {};
	\node[slot,fill=black] (s_3) at (2,2) {};
	\node[slot,fill=black] (s_4) at (3,2) {};
	\node[slot,fill=black] (s_5) at (4,2) {};
	\node[slot,fill=black] (s_6) at (5,2) {};
	\node[slot,fill=black] (s_7) at (6,2) {};
	\node[slot,fill=black] (s_8) at (7,2) {};
	\node[slot,fill=black] (s_9) at (8,2) {};
	
	\draw[blue,thick] (x_1) edge (s_1);
	\draw[blue,thick] (x_2) edge (s_1);
	\draw[blue,thick] (x_2) edge (s_2);
	\draw[blue,thick] (x_3) edge (s_2);
	\draw[blue,thick] (x_3) edge (s_3);
	\draw[blue,thick] (x_4) edge (s_3);
	\draw[blue,thick] (x_4) edge (s_4);
	\draw[blue,thick] (x_5) edge (s_4);
	\draw[blue,thick] (x_5) edge (s_6);
	\draw[blue,thick] (x_6) edge (s_6);
	\draw[blue,thick] (x_6) edge (s_7);
	\draw[blue,thick] (x_7) edge (s_7);
	\draw[blue,thick] (x_7) edge (s_8);
	\draw[blue,thick] (x_8) edge (s_8);
	\draw[blue,thick] (x_8) edge (s_9);
	\draw[blue,thick] (x_9) edge (s_9);
	
	\draw[red,thick] (x_1) edge (s_5);
	\draw[red,thick] (x_2) edge (s_5);
	\end{tikzpicture}
	\caption{\Cref{thm:oscmunrestricted}: An algorithm is presented the requests colored in blue first. Some slot has to be left open for which the request associated with the red edges is given.}
	\label{fig:oscmunrestricted}
\end{figure}
\begin{proof}
    Let us consider an algorithm $A$ solving  online slotted OSCM-$k$.
    Given the initial sets of vertices $V=\{v_1,\hdots, v_n\}$ and slots
    $S=\{s_1,\hdots,s_n\}$, $A$ is presented
    the following request sequence:
    $\{v_1,v_2\},\{v_2,v_3\},\ldots,\{v_{n-1},v_n\}$. 
    Assume without loss of generality that
    $A$ has assigned these requests to slots in $S$ without producing a single
    crossing.  Since we have $n$ requests to fill $n$ slots with, and
    $A$ has only one unfulfilled slot $s_i$ for some $i \in
    \{1,...,n\}$, the last request will be assigned to $s_i$. We assume,
    without loss of generality, that $i\le \lceil\frac{n}{2}\rceil$.
    The adversary now presents the request $\{v_{n-1},v_n\}$ as the
    last request of the input.  This results in at least $2\cdot
    2\cdot(\frac{n}{2} - 1)$ crossings as opposed to the optimal
    solution, which only results in a single crossing
    as depicted in \Cref{fig:oscmunrestricted}.  The
    competitive ratio is thus at least
    $\frac{2\cdot2\cdot(\frac{n}{2} - 1)}{1} = 2n-4$ and therefore not
    bounded by any fixed constant $c$.
\end{proof}

If we closely look at the proof, we see that the proof relies on
the adversary being able to freely choose the degrees of the vertices
in $V$. If we would require the degree of the vertices in $V$ to be 
defined in advance, the same strategy would not work. Thus, 
it makes sense to look at graph classes where the degree of
the vertices in the graph is fixed, in particular, 
regular graphs.

In what follows we focus, on online slotted OSCM-$2$ on
$2$-regular graphs, as this particular case is already hard to analyze,
and we prove that the competitive ratio is within the range between
$4/3$ and $5$. 

We conjecture that for any higher degree, online slotted OSCM-$k$ on $k$-regular graphs would also have a constant
competitive ratio, with the constant depending on $k$.
One can observe, that a higher vertex degree means that even optimal solutions
must have a lot of crossings. Thus, even when an online algorithm makes a sub-optimal
choice, the number of crossings of the optimal solution that it is compared to
should compensate for the mistakes.

\section{Lower Bound for \boldmath 2\unboldmath-Regular Graphs}

We begin by proving a lower bound for the competitive ratio of 
online slotted OSCM-$2$ on $2$-regular graphs.

It is important to note that an offline algorithm can find an optimal
solution in a greedy fashion, as we will see in Lemma~\ref{lem:node_exchange_uncritical}.
In the following lower bound, we prove that online algorithms cannot find an optimal solution,
greedily or otherwise. 
The difficulty is that a request cannot be assigned in between two consecutive fulfilled slots. 
Thus, an online algorithm has to fulfill a request by assigning it to a sub-optimal slot.
An example of such a situation is depicted in \Cref{fig:no_placement_optimal}. 
We can use this fact to construct a lower bound
for online slotted OSCM-$2$ on $2$-regular graphs as follows.

\begin{figure}[h]
    \centering
        \begin{tikzpicture}[node distance=1cm,
            slot/.style={draw,rectangle},
            vertex/.style={draw,circle},
    scale=0.8,every node/.style={scale=0.8}
            ]
            \node[vertex, label=below:{$x_1$}] (x_1) at (0,0) {};
            \node[vertex, label=below:{$x_2$}] (x_2) at (1,0) {};
            \node[vertex, label=below:{$y_1$}] (y_1) at (2,0) {};
            \node[vertex, label=below:{$y_2$}] (y_2) at (3,0) {};
            \node[slot,fill=black, label=above:$s_x$] (s_x) at (1,2) {};
            \node[slot,fill=black, label=above:$s_y$] (s_y) at (2,2) {};
            \node[slot] (dums) at (0,2) {};
            \node[slot] (dums2) at (3,2) {};
            
            \draw[red,thick] (x_1) edge (s_x);
            \draw[red,thick] (x_2) edge (s_x);
            \draw[blue,thick] (y_1) edge (s_y);
            \draw[blue,thick] (y_2) edge (s_y);
        \end{tikzpicture}
    \caption{In this graph, a new request $R_i = \{x_2,y_1\}$ appears. 
    This request cannot be fulfilled optimally. An assignment between $s_x$ and $s_y$ would
    be optimal, but there is no free slot between them.}
    \label{fig:no_placement_optimal}
\end{figure}
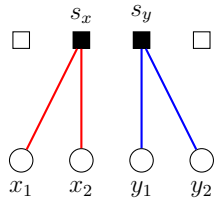

\begin{theorem}\label{thm:lower_bound}
    Every deterministic online algorithm, solving the slotted
    OSCM-$2$ on $2$-regular graphs, has a competitive ratio of at least
    $4/3 - \varepsilon$.
\end{theorem}

    \begin{figure}\begin{center}
        \begin{tikzpicture}[->,>=stealth,level distance=4em]
            \tikzstyle{level 1}=[sibling distance=8cm]
            \tikzstyle{level 2}=[sibling distance=4cm]
            \tikzstyle{level 3}=[sibling distance=2cm]
            \node {$\{v_3,v_4\}$}
                child { node {$\{v_1,v_2\},\{v_2,v_4\},\{v_1,v_3\}$}
                    child{ node{$c\ge 4/3$}
                      edge from parent node [right] {any placement}
                    }
                    edge from parent node [above left] {$s_3$ or lower}
                }
                child { node {$\{v_4,v_5\}, \{v_3,v_5\}, \{v_1,v_2\}, \{v_1,v_2\}$}
                    child{ node{$c\ge 4/3$}
                      edge from parent node [right] {any placement}
                    }
                    edge from parent node [above right] {$s_4$ or or higher}
                };
        \end{tikzpicture} 
        \caption{\Cref{thm:lower_bound}: All behaviors of any algorithm,
        when presented with request $\{v_3,v_4\}$, and the resulting competitive ratio of each decision branch, when confronted with this adversarial requests.}
        \label{fig:lower_bound}
    \end{center}\end{figure}
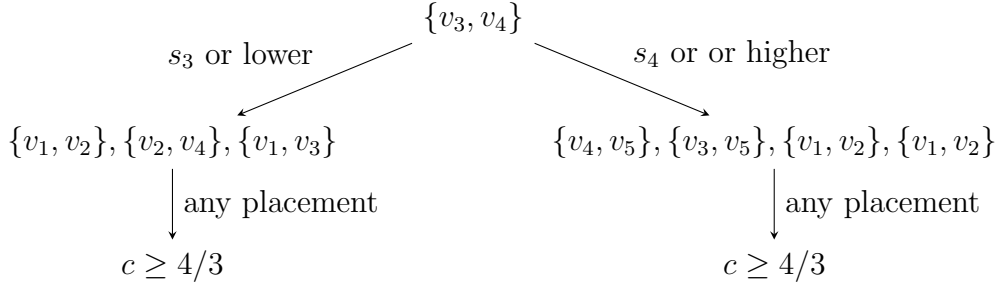

\begin{proof}
    For every node, every algorithm only has a finite amount of slots
    to insert it into. 
    Given an
    empty graph of size $n>6$, the adversary will repeat its strategy
    on the set of the five leftmost free nodes, filling up the graph
    from left to right, until $6$ or fewer free nodes are left in the
    graph. Given five free slots, the adversary will repeat the
    strategy depicted in \Cref{fig:lower_bound} which we will
    now describe in detail. For ease of notation, we will denote the
    five left-most free slots as $s_1,\hdots, s_5$ and the five 
    left-most edge-free vertices as $v_1,\hdots, v_5$.

    The adversary starts by presenting the request pair $\{v_3,v_4\}$.
    We consider two possibilities, either an algorithm places this request
    at $s_3$ or smaller, or $s_4$ or higher.

    \textbf{Case 1} (Algorithm assigns $\{v_3,v_4\}$ to $s_3$ or smaller): The adversary presents
    the request $\{v_1,v_2\}$. 
    We assume, that any reasonable algorithm places the second request to the left of the first one
    (on a smaller slot).
    If an algorithm would place the second request to the right it directly incurs $4$ crossings
    instead of none. 
    
    We branch on three possibilities depending on the free slots after the first two placements. 
    So, the free slots are either $\{s_1,s_4,s_5\}$, $\{s_2,s_4,s_5\}$, or $\{s_3,s_4,s_5\}$.
   
    The adversary presents the request $\{v_2,v_4\}$ and subsequently the request $\{v_1,v_3\}$.
    
    If the free slots are $\{s_3,s_4,s_5\}$ any assignment of $\{v_2,v_4\}$ and $\{v_1,v_3\}$ results in
    at least $7$ crossings.
    
    If the free slots are $\{s_2,s_4,s_5\}$ any assignment of $\{v_2,v_4\}$ and $\{v_1,v_3\}$ results in
    at least $4$ crossings.
  
    If the free slots are $\{s_1,s_4,s_5\}$ any assignment of $\{v_2,v_4\}$ and $\{v_1,v_3\}$ results in 
    at least $5$ crossings.
    
    However, an assignment of $\{v_3,v_4\}$ to $s_4$, $\{v_1,v_2\}$ to $s_1$, $\{v_2,v_4\}$ to $s_3$ and
    $\{v_1,v_3\}$ to $s_2$ results in only $3$ crossings, where the considered algorithms
    all have at least four crossings. Thus, the competitive ratio of any algorithm
    assigning request $\{v_3,v_4\}$ to $s_3$ or lower is at least $4/3$-competitive on these four requests.

    \textbf{Case 2} (Algorithm assigns $\{v_3,v_4\}$ to $s_4$ or larger): 
    The adversary presents the request $\{v_4,v_5\}$ followed by 
    $\{v_3,v_5\}$ and the two identical requests $\{v_1,v_2\}$. 
    
    The assignment of the last two requests to slots $s_1$ and $s_2$ is optimal 
    and generates one crossing.
    
    The optimal assignment places $\{v_3,v_4\}$ to $s_3$, $\{v_4,v_5\}$ to 
    $s_5$ and $\{v_3,v_5\}$ to $s_4$ resulting in two crossings, $3$ in total with the 
    last two requests.
    
    However, the algorithms we consider cannot assign $\{v_3,v_4\}$ to $s_3$, they assign it
    to $s_4$ or higher, these algorithms incur at least $3$ crossings for the first three requests,
    and $4$ crossings in total for the five requests. Thus, they have a competitive ratio 
    of at least $4/3$ on these five requests too,
    as expected.  
    
    Note that in Case 1, the adversary presents only
    four nodes in total, while in Case 2, five nodes
    are used. Independent of which case is used, the adversary can now
    use the five left-most free slots and edge-free vertices to repeat this tactic.
    Once $r \leq 6$ slots are left, the remaining slots are filled
    up as follows. 
    The adversary presents the following $r$ requests:
    $\{v_{n-r},v_{n-r+1}\},\ldots,\{v_{n-1},v_{n}\}$. 
    One slot is still free after presenting these requests, and the last request is $\{v_{n-r},v_{n}\}$.
    This results in $r-1 \leq 5$ additional, unavoidable crossings.

    From the case distinction above and the argument to fill up the
    rest of the graph, one can easily verify that the competitive
    ratio of every algorithm tends to $4/3$ for growing $n$.
\end{proof}

This lower bound proves that no online algorithm for online slotted OSCM-$2$ on $2$-regular
graphs can perform optimally on all instances. In the following, we introduce some notions that
are used to prove an upper bound for the competitive ratio in the same setting.

\section{Preliminaries and Notation}

In order to prove upper bounds for online slotted OSCM-$2$ on $2$-regular
graphs, we need to first extract some structural properties of this problem.
First, we introduce the notion of \emph{propagation arrows}, which helps
us to lower bound the total number of crossings of the remaining graph if we only
have a partial request sequence.
Then, we observe that finding an optimal placement, involves only looking
at the placement of every pair of requests relative to each other.


The number of crossings of an optimal assignment for a request sequence 
is the \emph{number of unavoidable crossings} of the request sequence. 
The difference between the number of crossings incurred by an algorithm $A$,
and the number of unavoidable crossings is consequently the \emph{number of avoidable
crossings} of $A$ on that request sequence.

Consider a $2$-regular instance for online slotted OSCM-$k$
with slots $S=\{s_1, \hdots, s_n\}$ and vertices $V=\{v_1,\hdots, v_n\}$,
and a request sequence $R_1,\hdots, R_n$. 
Let us assume that at some point after the $k$-th request has been fulfilled by algorithm $A$,
there are fulfilled slots, and the vertices in $V$ have degree 2, 1 or 0, depending on how many times these vertices have
appeared in requests. Because we know that the final graph will be $2$-regular, for those vertices in $V$ with degree less than two we are still expecting
a request that contains the vertex, and for any unfulfilled slot, there will be a request which will be fulfilled using this slot.

Intuitively, we use propagation arrows to greedily match unfulfilled vertices to available slots in a way that minimizes the number of crossings. 
We can see this through an illustration in Figure~\ref{fig:arrows}. For instance, in an
empty graph every vertex $v_i$ in $V$ will have two propagation arrows to the slot $s_i$,
but once some slots are occupied we take the leftmost vertex with degree less than two and
assign a propagation arrow to the left-most unfulfilled slot. We know that the instance is $2$-regular,
so for every pair of missing edges of vertices in $V$ there must be an empty slot.
We can define the propagation arrows formally as follows.

 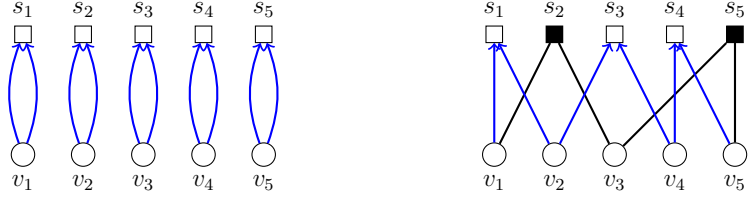
\begin{figure}
 \begin{center}
 \subfigure{
  \begin{tikzpicture}[node distance=1cm,
  slot/.style={draw,rectangle},
  vertex/.style={draw,circle},
  scale=0.8,every node/.style={scale=0.8}
  ]
      \foreach \t in {1,...,5}{
         \node[vertex, label=below:$v_\t$] (w_\t) at (\t,0) {};
         \node[slot, label=above:$s_\t$] (s_\t) at (\t,2) {};
         \draw[blue,thick,->] (w_\t) edge [bend left=20](s_\t);
         \draw[blue,thick,->] (w_\t) edge [bend right=20](s_\t);
      }
  \end{tikzpicture}
  }
  \hspace*{2cm}
  \subfigure{
  \begin{tikzpicture}[node distance=1cm,
  slot/.style={draw,rectangle},
  vertex/.style={draw,circle},
  scale=0.8,every node/.style={scale=0.8}
  ]
      \foreach \t in {1,...,5}{
         \node[vertex, label=below:$v_\t$] (w_\t) at (\t,0) {};
      }
      \node[slot, label=above:$s_1$] (s_1) at (1,2) {};
      \node[slot,fill=black, label=above:$s_2$] (s_2) at (2,2) {};
      \node[slot, label=above:$s_3$] (s_3) at (3,2) {};
      \node[slot, label=above:$s_4$] (s_4) at (4,2) {};
      \node[slot,fill=black, label=above:$s_5$] (s_5) at (5,2) {};
      \draw[thick] (s_2) edge (w_1)
		   (s_2) edge (w_3)
		   (s_5) edge (w_5)
		   (s_5) edge (w_3);
		   
      \draw[blue,thick,->] (w_1) edge (s_1);
      \draw[blue,thick,->] (w_2) edge (s_1);
      \draw[blue,thick,->] (w_2) edge (s_3);
      \draw[blue,thick,->] (w_4) edge (s_3);
      \draw[blue,thick,->] (w_4) edge (s_4);
      \draw[blue,thick,->] (w_5) edge (s_4);
  \end{tikzpicture}
  }
  \caption{Propagation arrows before the first instance and after part of the instance is fulfilled.}
  \label{fig:arrows}
  \end{center}
 \end{figure}

First,
we know that after $k$ requests for a $2$-regular graph, there are 
$n-k$ unfulfilled requests, which corresponds to $2(n-k)$ missing edges.
We will double count the missing edges with the following two lists.

The \emph{list of unfulfilled vertices} $L_V$ of an instance
after the $k$-th request, is an ordered list that contains every vertex
$v_i\in V$ from smallest to largest at most twice.
$L_V$ will contain no copies of a vertex $v_i\in V$ if it already has appeared
twice in the request sequence $R_1,\hdots, R_k$, i.e., if $v_i$ has degree $2$,
$L_V$ will contain $v_i\in V$ once if $v_i$ has appeared
only once in $R_1,\hdots, R_k$, i.e., if $v_i$ has degree $1$ in the partially fulfilled graph,
finally, $L_V$ contains a vertex $v_i$ twice if $v_i$ does not appear in $R_1,\hdots, R_k$, and 
thus has degree $0$ at that point.


We can, thus, analogously consider the \emph{list of unfulfilled slots} $L_S$ as an
ordered list that contains each unfulfilled slot twice, again from smallest to largest.
From the previous observation it should be clear that $|L_V|=|L_S|$.

\begin{definition}\label{def:prop_arrows}
Consider a $2$-regular instance for online slotted OSCM-$k$
with slots $S = \{s_1, \hdots, s_n\}$ and vertices $V = \{v_1,\hdots, v_n\}$,
and a request sequence $R_1,\hdots, R_n$. Let $A$ be an algorithm that has fulfilled
$k$ requests. Let us consider the corresponding $L_V$ and $L_S$ for this request.
There is a \emph{propagation arrow} from vertex $v$ to slot $s$ if both occupy
the same place in the ordered lists $L_V$ and $L_S$, i.e., if $v$ is the $i$-th
element of $L_V$ and $s$ is  $i$-th element of $L_S$ for some $i\in[2(n-k)]$.
\end{definition}

Observe that propagation arrows do not cross one another by construction.
So, if we count the crossings of a partial graph including the crossings between
graph edges and propagation arrows, we have a lower bound on the number of crossings
that the graph will have after the request sequence is completely fulfilled.

In the following, we want to observe, that an instance is optimally solved
if an only if, for every pair of requests, the relative order of their slot
assignments is optimal, i.e., if the placement of these two requests is such
that there are fewer crossings between them than otherwise.
This basically means, that a crossing is unavoidable, if and only if, the 
relative order of the two requests involved in this crossing is optimal, 
regardless of any other placement of any other request within the graph.
This provides us with a very powerful tool to analyze the performance
of online algorithms solving online slotted OSCM-$2$ on $2$-regular graphs. 

In order to prove the aforementioned statement, we first need the following lemma.

\begin{figure}
	\centering
	\subfigure{
		\begin{tikzpicture}[node distance=1cm,
		slot/.style={draw,rectangle},
		vertex/.style={draw,circle},
		scale=0.3,every node/.style={scale=0.4}
		]
		\node[vertex] (x_1) at (0,0) {};
		\node[vertex] (dum) at (1,0) {};
		\node[vertex] (x_2) at (2,0) {};
		\node[slot,fill=black] (s_x) at (1,2) {};
		\node[slot,fill=black] (s_y) at (0,2) {};
		\node[slot] (dums) at (2,2) {};

		\draw[red,thick] (x_1) edge (s_x);
		\draw[red,thick] (x_2) edge (s_x);
		\draw[blue,thick] (x_1) edge (s_y);
		\draw[blue,thick] (x_2) edge (s_y);
		
		\node[vertex] (x_1) at (3.5,0) {};
		\node[vertex] (dum) at (4.5,0) {};
		\node[vertex] (x_2) at (5.5,0) {};
		\node[slot,fill=black] (s_x) at (3.5,2) {};
		\node[slot,fill=black] (s_y) at (4.5,2) {};
		\node[slot] (dums) at (5.5,2) {};

		\draw[red,thick] (x_1) edge (s_x);
		\draw[red,thick] (x_2) edge (s_x);
		\draw[blue,thick] (x_1) edge (s_y);
		\draw[blue,thick] (x_2) edge (s_y);
		
		\node[scale=1.75] (a) at (2.75,-1) {(a) 1-1};
		\end{tikzpicture}
	}
	\hspace*{1.5cm}
	\subfigure{
		\begin{tikzpicture}[node distance=1cm,
		slot/.style={draw,rectangle},
		vertex/.style={draw,circle},
		scale=0.3,every node/.style={scale=0.4}
		]
		\node[vertex] (x_1) at (0,0) {};
		\node[vertex] (y_1) at (1,0) {};
		\node[vertex] (x_2) at (2,0) {};
		\node[vertex] (y_2) at (3,0) {};
		\node[slot,fill=black, label=above:] (s_x) at (1,2) {};
		\node[slot,fill=black, label=above:] (s_y) at (0,2) {};
		\node[slot] (dums) at (2,2) {};
		\node[slot] (dums2) at (3,2) {};
		
		\draw[red,thick] (x_1) edge (s_x);
		\draw[red,thick] (x_2) edge (s_x);
		\draw[blue,thick] (y_1) edge (s_y);
		\draw[blue,thick] (y_2) edge (s_y);
		
		\node[vertex] (x_1) at (4.5,0) {};
		\node[vertex] (y_1) at (5.5,0) {};
		\node[vertex] (x_2) at (6.5,0) {};
		\node[vertex] (y_2) at (7.5,0) {};
		\node[slot,fill=black] (s_x) at (4.5,2) {};
		\node[slot,fill=black] (s_y) at (5.5,2) {};
		\node[slot] (dums) at (6.5,2) {};
		\node[slot] (dums2) at (7.5,2) {};

		\draw[red,thick] (x_1) edge (s_x);
		\draw[red,thick] (x_2) edge (s_x);
		\draw[blue,thick] (y_1) edge (s_y);
		\draw[blue,thick] (y_2) edge (s_y);
		
		\node[scale=1.75] (a) at (3.75,-1) {(d) 3-1};
		\end{tikzpicture}
	}
	\hspace*{1.5cm}
	\subfigure{
		\begin{tikzpicture}[node distance=1cm,
		slot/.style={draw,rectangle},
		vertex/.style={draw,circle},
		scale=0.3,every node/.style={scale=0.4}
		]
		\node[vertex] (x_1) at (0,0) {};
		\node[vertex] (x_2) at (1,0) {};
		\node[vertex] (y_2) at (2,0) {};
		\node[slot,fill=black] (s_x) at (1,2) {};
		\node[slot,fill=black] (s_y) at (0,2) {};
		\node[slot] (dums) at (2,2) {};

		\draw[red,thick] (x_1) edge (s_x);
		\draw[red,thick] (x_2) edge (s_x);
		\draw[blue,thick] (x_1) edge (s_y);
		\draw[blue,thick] (y_2) edge (s_y);
		
		\node[vertex, ] (x_1) at (3.5,0) {};
		\node[vertex, ] (x_2) at (4.5,0) {};
		\node[vertex, ] (y_2) at (5.5,0) {};
		\node[slot,fill=black] (s_x) at (3.5,2) {};
		\node[slot,fill=black] (s_y) at (4.5,2) {};
		\node[slot] (dums) at (5.5,2) {};

		\draw[red,thick] (x_1) edge (s_x);
		\draw[red,thick] (x_2) edge (s_x);
		\draw[blue,thick] (x_1) edge (s_y);
		\draw[blue,thick] (y_2) edge (s_y);
		
		\node[scale=1.75] (a) at (2.75,-1) {(b) 2-1};
		\end{tikzpicture}
	}
	\\    
	\hspace*{0.1cm}
	\subfigure{
		\begin{tikzpicture}[node distance=1cm,
		slot/.style={draw,rectangle},
		vertex/.style={draw,circle},
		scale=0.3,every node/.style={scale=0.4}
		]
		\node[vertex] (x_1) at (0,0) {};
		\node[vertex] (y_1) at (2,0) {};
		\node[vertex] (x_2) at (1,0) {};
		\node[vertex] (y_2) at (3,0) {};
		\node[slot,fill=black] (s_x) at (1,2) {};
		\node[slot,fill=black] (s_y) at (0,2) {};
		\node[slot] (dums) at (2,2) {};
		\node[slot] (dums2) at (3,2) {};
		
		\draw[red,thick] (x_1) edge (s_x);
		\draw[red,thick] (x_2) edge (s_x);
		\draw[blue,thick] (y_1) edge (s_y);
		\draw[blue,thick] (y_2) edge (s_y);
		
		\node[vertex] (x_1) at (4.5,0) {};
		\node[vertex] (y_1) at (6.5,0) {};
		\node[vertex] (x_2) at (5.5,0) {};
		\node[vertex] (y_2) at (7.5,0) {};
		\node[slot,fill=black] (s_x) at (4.5,2) {};
		\node[slot,fill=black] (s_y) at (5.5,2) {};
		\node[slot] (dums) at (6.5,2) {};
		\node[slot] (dums2) at (7.5,2) {};

		\draw[red,thick] (x_1) edge (s_x);
		\draw[red,thick] (x_2) edge (s_x);
		\draw[blue,thick] (y_1) edge (s_y);
		\draw[blue,thick] (y_2) edge (s_y);
		
		\node[scale=1.75] (a) at (3.75,-1) {(e) 4-0};
		\end{tikzpicture}
	}
	\hspace*{1.5cm}
	\subfigure{
		\begin{tikzpicture}[node distance=1cm,
		slot/.style={draw,rectangle},
		vertex/.style={draw,circle},
		scale=0.3,every node/.style={scale=0.4}
		]
		\node[vertex] (x_1) at (0,0) {};
		\node[vertex] (x_2) at (1,0) {};
		\node[vertex] (y_2) at (2,0) {};
		\node[slot,fill=black] (s_x) at (1,2) {};
		\node[slot,fill=black] (s_y) at (0,2) {};
		\node[slot] (dums) at (2,2) {};

		\draw[red,thick] (x_1) edge (s_x);
		\draw[red,thick] (x_2) edge (s_x);
		\draw[blue,thick] (x_2) edge (s_y);
		\draw[blue,thick] (y_2) edge (s_y);
		
		\node[vertex] (x_1) at (3.5,0) {};
		\node[vertex] (x_2) at (4.5,0) {};
		\node[vertex] (y_2) at (5.5,0) {};
		\node[slot,fill=black] (s_x) at (3.5,2) {};
		\node[slot,fill=black] (s_y) at (4.5,2) {};
		\node[slot] (dums) at (5.5,2) {};

		\draw[red,thick] (x_1) edge (s_x);
		\draw[red,thick] (x_2) edge (s_x);
		\draw[blue,thick] (x_2) edge (s_y);
		\draw[blue,thick] (y_2) edge (s_y);
		
		\node[scale=1.75] (a) at (2.75,-1) {(c) 3-0};
		\end{tikzpicture}
	}
	\hspace*{1.5cm}
	\subfigure{
		\begin{tikzpicture}[node distance=1cm,
		slot/.style={draw,rectangle},
		vertex/.style={draw,circle},
		scale=0.3,every node/.style={scale=0.4}
		]
		\node[vertex] (x_1) at (0,0) {};
		\node[vertex] (y_1) at (1,0) {};
		\node[vertex] (y_2) at (2,0) {};
		\node[vertex] (x_2) at (3,0) {};
		\node[slot,fill=black] (s_x) at (1,2) {};
		\node[slot,fill=black] (s_y) at (0,2) {};
		\node[slot] (dums) at (2,2) {};
		\node[slot] (dums2) at (3,2) {};
		
		\draw[red,thick] (x_1) edge (s_x);
		\draw[red,thick] (x_2) edge (s_x);
		\draw[blue,thick] (y_1) edge (s_y);
		\draw[blue,thick] (y_2) edge (s_y);
		
		\node[vertex] (x_1) at (4.5,0) {};
		\node[vertex] (y_1) at (6.5,0) {};
		\node[vertex] (y_2) at (5.5,0) {};
		\node[vertex] (x_2) at (7.5,0) {};
		\node[slot,fill=black] (s_x) at (4.5,2) {};
		\node[slot,fill=black] (s_y) at (5.5,2) {};
		\node[slot] (dums) at (6.5,2) {};
		\node[slot] (dums2) at (7.5,2) {};

		\draw[red,thick] (x_1) edge (s_x);
		\draw[red,thick] (x_2) edge (s_x);
		\draw[blue,thick] (y_1) edge (s_y);
		\draw[blue,thick] (y_2) edge (s_y);
		
		\node[scale=1.75] (a) at (3.75,-1) {(f) 2-2};
		\end{tikzpicture}
	}
	\caption{Case distinction for step one of \Cref{lem:node_exchange_uncritical}. Each case is depicted before and after the untangling. The request $s_x$ is drawn in red and $s_y$ in blue.}
	\label{fig:node_exchange_uncritical}
\end{figure}
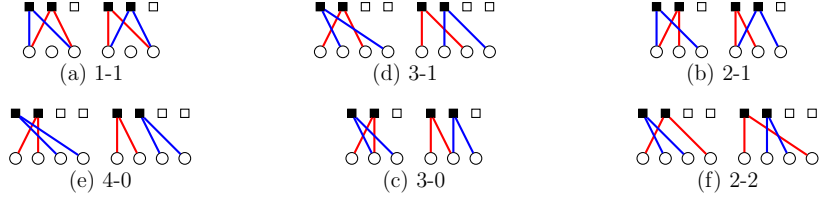

\begin{lemma}\label{lem:node_exchange_uncritical}
    Given two requests $R_x = \{x_1,x_2\}$ and $R_y = \{y_1, y_2\}$
    assigned to slots $s_x$ and $s_y$. Without loss of generality
    assume that $x_1 \leq y_1$  and $x_2\le y_2$.
    An assignment where $s_x < s_y$ generates fewer or equally many crossings in the final graph
    than an assignment where $s_y < s_x$ if every other assigned
    slot remains unchanged.
\end{lemma}
\begin{proof}
    We separate this proof into two steps. In the first step, we show
    that the number of crossings between the two requests is always
    the same or smaller if $s_x<s_y$. This can be done with an
    exhaustive case distinction
    and is depicted in \Cref{fig:node_exchange_uncritical}.
    
    Now, for the second step, we need to show that the number of
    crossings in an overall graph is still smaller or equal if $s_x<s_y$.
    We prove it by means of a contradiction.
    Let us consider a 2-regular graph $G$ for which we have two such requests
    $\{x_1,x_2\}$ and $\{y_1, y_2\}$
    and an assignment  where $s_y < s_x$,
    with a total number of crossings $c_G$.
    And let us consider the graph $G'$ which is the same as $G$ except that the placement of $\{x_1,x_2\}$ and $\{y_1, y_2\}$ is exchanged, making $s_x<s_y$, with a number of crossings $c_{G'}$. 
    Assume that $c_G<c_{G'}$. We already
    know that this is not due to the number of crossings between edges to $s_x$
    and $s_y$, as such a case would be covered by \Cref{fig:node_exchange_uncritical}.
    Without loss of generality assume thus that (one of) the extra crossing(s)
    in $G'$ is between some edge $(u,s_u)$ and one of the modified edges $(x_t,s_i)$ with $t\in\{1,2\}$.
    
    In order for this pair of edges to produce an extra crossing in $G'$
    compared to $G$ at all, we know that $u \notin [x_1,x_2]$, as otherwise this
    crossing is unavoidable and thus the same in $G$ and $G'$. We make a case
    distinction over the remaining cases, which we depict
    in~\Cref{fig:no_extra_cross_by_swap}.

    Assume now that $x_2 < u$. Then, $s_u < s_x$ in order for the edges to cross
    at all. This positioning produces two crossings with $R_x$ in $G$ and
    possibly some crossings with $R_y$. However, since $R_x$ is only assigned
    further to the right in $G'$, we get the exact same number of crossings
    between $(u,s_u)$ and the edges of $R_x$ and $R_y$ in $G'$.

    Assume finally that $u < x_2$. Then, $s_x < s_u$ in order for the edges to
    cross at all. This positioning produces two crossings with $R_x$ in $G$ and
    possibly some crossings with $R_y$. We do a case distinction whether $s_u <
    s_y$ or $s_y < s_u$.

    If $s_u < s_y$, then $R_y$ and $R_x$ simply ``change roles'' in $G'$
    compared to $G$ and the number of crossings remains the same. If $s_y <
    s_u$, then $(u,s_u)$ crosses the edges of $R_x$ and of $R_y$ completely in
    both $G$ and $G'$.
    
    Thus, by swapping the slot assignment in this way one cannot reduce the
    number of crossings in the overall graph.
\end{proof}

\Cref{lem:node_exchange_uncritical} plainly states that for each pair of requests,
the optimal ordering gives the left-most request a slot that is to the left of the slot 
assigned to the right-most request.
The notion of left and right requests only means here, that if the requests are not for 
identical pairs of vertices, the left request contains the left-most distinguished vertex.


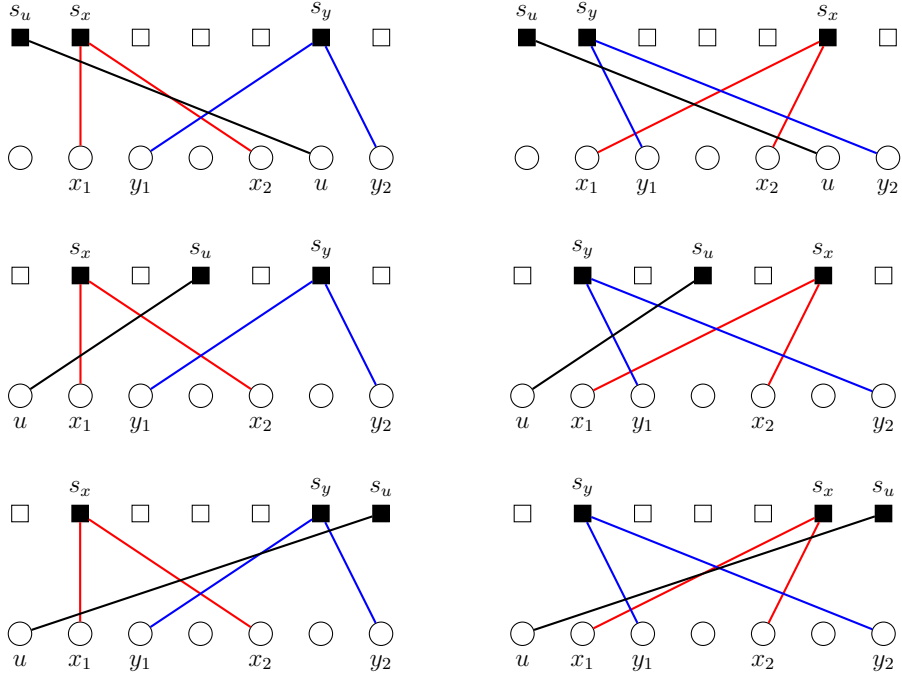
\begin{figure}
    \centering
    \subfigure{
    \begin{tikzpicture}[node distance=1cm,
        slot/.style={draw,rectangle},
        vertex/.style={draw,circle},
  scale=0.8,every node/.style={scale=0.8}
        ]
        \node[vertex, label=below:{}]      (a_1) at (0,0) {};
        \node[vertex, label=below:{$x_1$}] (x_1) at (1,0) {};
        \node[vertex, label=below:{$y_1$}]      (y_1) at (2,0) {};
        \node[vertex, label=below:{}]   (a_2)   at (3,0) {};
        \node[vertex, label=below:{$x_2$}] (x_2) at (4,0) {};
        \node[vertex, label=below:{$u$}]      (u) at (5,0) {};
        \node[vertex, label=below:{$y_2$}]      (y_2) at (6,0) {};
        \node[slot,fill=black, label=above:$s_x$] (s_x) at (1,2) {};
        \node[slot,fill=black, label=above:$s_y$] (s_y) at (5,2) {};
        \node[slot] (dums) at (2,2) {};
        \node[slot,fill=black, label=above:$s_u$] (s_u) at (0,2) {};
        \node[slot] (dums3) at (4,2) {};
        \node[slot] (dums4) at (3,2) {};
        \node[slot] (dums3) at (6,2) {};
        
        \draw[red,thick] (x_1) edge (s_x);
        \draw[red,thick] (x_2) edge (s_x);
        \draw[blue,thick] (y_1) edge (s_y);
        \draw[blue,thick] (y_2) edge (s_y);
        \draw[black,thick] (u) edge (s_u);
    \end{tikzpicture}
    }\hspace*{1cm}
    \subfigure{
    \begin{tikzpicture}[node distance=1cm,
        slot/.style={draw,rectangle},
        vertex/.style={draw,circle},
  scale=0.8,every node/.style={scale=0.8}
        ]
        \node[vertex, label=below:{}]      (a_1) at (0,0) {};
        \node[vertex, label=below:{$x_1$}] (x_1) at (1,0) {};
        \node[vertex, label=below:{$y_1$}]      (y_1) at (2,0) {};
        \node[vertex, label=below:{}]   (a_2)   at (3,0) {};
        \node[vertex, label=below:{$x_2$}] (x_2) at (4,0) {};
        \node[vertex, label=below:{$u$}]      (u) at (5,0) {};
        \node[vertex, label=below:{$y_2$}]      (y_2) at (6,0) {};
        \node[slot,fill=black, label=above:$s_y$] (s_y) at (1,2) {};
        \node[slot,fill=black, label=above:$s_x$] (s_x) at (5,2) {};
        \node[slot] (dums) at (2,2) {};
        \node[slot,fill=black, label=above:$s_u$] (s_u) at (0,2) {};
        \node[slot] (dums3) at (4,2) {};
        \node[slot] (dums4) at (3,2) {};
        \node[slot] (dums3) at (6,2) {};
        
        \draw[red,thick] (x_1) edge (s_x);
        \draw[red,thick] (x_2) edge (s_x);
        \draw[blue,thick] (y_1) edge (s_y);
        \draw[blue,thick] (y_2) edge (s_y);
        \draw[black,thick] (u) edge (s_u);
    \end{tikzpicture}
    }\\
    \subfigure{
    \begin{tikzpicture}[node distance=1cm,
        slot/.style={draw,rectangle},
        vertex/.style={draw,circle},
  scale=0.8,every node/.style={scale=0.8}
        ]
        \node[vertex, label=below:{$u$}]      (u) at (0,0) {};
        \node[vertex, label=below:{$x_1$}] (x_1) at (1,0) {};
        \node[vertex, label=below:{$y_1$}]      (y_1) at (2,0) {};
        \node[vertex, label=below:{}]   (a_2)   at (3,0) {};
        \node[vertex, label=below:{$x_2$}] (x_2) at (4,0) {};
        \node[vertex, label=below:{}]      (a_3) at (5,0) {};
        \node[vertex, label=below:{$y_2$}]      (y_2) at (6,0) {};
        \node[slot,fill=black, label=above:$s_x$] (s_x) at (1,2) {};
        \node[slot,fill=black, label=above:$s_y$] (s_y) at (5,2) {};
        \node[slot] (dums) at (2,2) {};
        \node[slot,fill=black, label=above:$s_u$] (s_u) at (3,2) {};
        \node[slot] (dums3) at (4,2) {};
        \node[slot] (dums4) at (0,2) {};
        \node[slot] (dums3) at (6,2) {};
        
        \draw[red,thick] (x_1) edge (s_x);
        \draw[red,thick] (x_2) edge (s_x);
        \draw[blue,thick] (y_1) edge (s_y);
        \draw[blue,thick] (y_2) edge (s_y);
        \draw[black,thick] (u) edge (s_u);
    \end{tikzpicture}
    }\hspace*{1cm}
    \subfigure{
    \begin{tikzpicture}[node distance=1cm,
        slot/.style={draw,rectangle},
        vertex/.style={draw,circle},
  scale=0.8,every node/.style={scale=0.8}
        ]
        \node[vertex, label=below:{$u$}]      (u) at (0,0) {};
        \node[vertex, label=below:{$x_1$}] (x_1) at (1,0) {};
        \node[vertex, label=below:{$y_1$}]      (y_1) at (2,0) {};
        \node[vertex, label=below:{}]   (a_2)   at (3,0) {};
        \node[vertex, label=below:{$x_2$}] (x_2) at (4,0) {};
        \node[vertex, label=below:{}]      (a_3) at (5,0) {};
        \node[vertex, label=below:{$y_2$}]      (y_2) at (6,0) {};
        \node[slot,fill=black, label=above:$s_y$] (s_y) at (1,2) {};
        \node[slot,fill=black, label=above:$s_x$] (s_x) at (5,2) {};
        \node[slot] (dums) at (2,2) {};
        \node[slot,fill=black, label=above:$s_u$] (s_u) at (3,2) {};
        \node[slot] (dums3) at (4,2) {};
        \node[slot] (dums4) at (0,2) {};
        \node[slot] (dums3) at (6,2) {};
        
        \draw[red,thick] (x_1) edge (s_x);
        \draw[red,thick] (x_2) edge (s_x);
        \draw[blue,thick] (y_1) edge (s_y);
        \draw[blue,thick] (y_2) edge (s_y);
        \draw[black,thick] (u) edge (s_u);
    \end{tikzpicture}
    }\\
    \subfigure{
    \begin{tikzpicture}[node distance=1cm,
        slot/.style={draw,rectangle},
        vertex/.style={draw,circle},
  scale=0.8,every node/.style={scale=0.8}
        ]
        \node[vertex, label=below:{$u$}]      (u) at (0,0) {};
        \node[vertex, label=below:{$x_1$}] (x_1) at (1,0) {};
        \node[vertex, label=below:{$y_1$}]      (y_1) at (2,0) {};
        \node[vertex, label=below:{}]   (a_2)   at (3,0) {};
        \node[vertex, label=below:{$x_2$}] (x_2) at (4,0) {};
        \node[vertex, label=below:{}]      (a_3) at (5,0) {};
        \node[vertex, label=below:{$y_2$}]      (y_2) at (6,0) {};
        \node[slot,fill=black, label=above:$s_x$] (s_x) at (1,2) {};
        \node[slot,fill=black, label=above:$s_y$] (s_y) at (5,2) {};
        \node[slot] (dums) at (2,2) {};
        \node[slot,fill=black, label=above:$s_u$] (s_u) at (6,2) {};
        \node[slot] (dums3) at (4,2) {};
        \node[slot] (dums4) at (0,2) {};
        \node[slot] (dums3) at (3,2) {};
        
        \draw[red,thick] (x_1) edge (s_x);
        \draw[red,thick] (x_2) edge (s_x);
        \draw[blue,thick] (y_1) edge (s_y);
        \draw[blue,thick] (y_2) edge (s_y);
        \draw[black,thick] (u) edge (s_u);
    \end{tikzpicture}
    }\hspace*{1cm}
    \subfigure{
    \begin{tikzpicture}[node distance=1cm,
        slot/.style={draw,rectangle},
        vertex/.style={draw,circle},
  scale=0.8,every node/.style={scale=0.8}
        ]
        \node[vertex, label=below:{$u$}]      (u) at (0,0) {};
        \node[vertex, label=below:{$x_1$}] (x_1) at (1,0) {};
        \node[vertex, label=below:{$y_1$}]      (y_1) at (2,0) {};
        \node[vertex, label=below:{}]   (a_2)   at (3,0) {};
        \node[vertex, label=below:{$x_2$}] (x_2) at (4,0) {};
        \node[vertex, label=below:{}]      (a_3) at (5,0) {};
        \node[vertex, label=below:{$y_2$}]      (y_2) at (6,0) {};
        \node[slot,fill=black, label=above:$s_y$] (s_y) at (1,2) {};
        \node[slot,fill=black, label=above:$s_x$] (s_x) at (5,2) {};
        \node[slot] (dums) at (2,2) {};
        \node[slot,fill=black, label=above:$s_u$] (s_u) at (6,2) {};
        \node[slot] (dums3) at (4,2) {};
        \node[slot] (dums4) at (0,2) {};
        \node[slot] (dums3) at (3,2) {};
        
        \draw[red,thick] (x_1) edge (s_x);
        \draw[red,thick] (x_2) edge (s_x);
        \draw[blue,thick] (y_1) edge (s_y);
        \draw[blue,thick] (y_2) edge (s_y);
        \draw[black,thick] (u) edge (s_u);
    \end{tikzpicture}
    }
    \caption{As shown in~\Cref{lem:node_exchange_uncritical}, there cannot be a
    crossing between $R_x$ and an edge $(u,s_u)$ that makes the ordering $s_y <
    s_x$ better than $s_x < s_y$.}
    \label{fig:no_extra_cross_by_swap}
\end{figure}

%

In order to find an upper bound on the competitive ratio, we only have to see
that any pair of requests is either placed optimally or otherwise bound the number of crossings
generated by that pair with the number of unavoidable crossings in the optimal solution.

\section{Upper Bound for \boldmath 2\unboldmath-Regular Graphs}

With these structural properties we are ready to present the algorithm that will provide us
with an upper bound of $5$ for the competitive ratio.

Neglecting to take the state of the graph into account when making
decisions regarding the insertion of requests seems to result in
relatively bad upper bounds. As an example, we take the simple
barycenter algorithm (\Cref{alg:barycenter}) proposed
in~\cite{walter}, which optimally solves the offline OSCM-2.
This algorithm, computes the average between the two requested vertices 
and assigns it to this particular point. In the case of the slotted version of
the problem, we have to adjust it to take the nearest free slot.

\begin{algorithm}[tb]
    \begin{algorithmic}[1]
        \State{\textit{free\_slots} = $\{1,\dots,n\}$;}
        \For{\textit{$\{x_1,x_2\}$} in \textit{input}}
            \State{$s := \lfloor\frac{x_1 + x_2}{2}\rfloor$;}
            \While{{\it$s$.isUsed()}}
                \State{$s := \{\, t \mid argmin_{\forall t\in S. \neg {\it t.isUsed()}}(t - s) \mid \,\}$ // take leftmost on tie}
            \EndWhile
            \State{Assign $\{x_1,x_2\}$ to $s$;}
        \EndFor
    \end{algorithmic}
    \caption{Barycenter algorithm from~\cite{walter}, adjusted for the slotted case.}
    \label{alg:barycenter}
\end{algorithm}

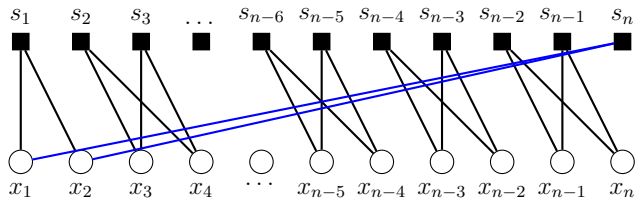
\begin{figure}
    \centering
    \begin{tikzpicture}[node distance=1cm,
        slot/.style={draw,rectangle},
        vertex/.style={draw,circle},
  scale=0.8,every node/.style={scale=0.8}
        ]
        \node[vertex, label=below:{$x_1$}] (x_1) at (0,0) {};
        \node[vertex, label=below:{$x_2$}] (x_2) at (1,0) {};
        \node[vertex, label=below:{$x_3$}] (x_3) at (2,0) {};
        \node[vertex, label=below:{$x_4$}] (x_4) at (3,0) {};
        \node[vertex, label=below:{$\dots$}] (x_dum) at (4,0) {};
        \node[vertex, label=below:{$x_{n-5}$}] (x_5) at (5,0) {};
        \node[vertex, label=below:{$x_{n-4}$}] (x_6) at (6,0) {};
        \node[vertex, label=below:{$x_{n-3}$}] (x_7) at (7,0) {};
        \node[vertex, label=below:{$x_{n-2}$}] (x_8) at (8,0) {};
        \node[vertex, label=below:{$x_{n-1}$}] (x_9) at (9,0) {};
        \node[vertex, label=below:{$x_n$}] (x_10) at (10,0) {};
        \node[slot,fill=black, label=above:$s_1$] (s_1) at (0,2) {};
        \node[slot,fill=black, label=above:$s_2$] (s_2) at (1,2) {};
        \node[slot,fill=black, label=above:$s_3$] (s_3) at (2,2) {};
        \node[slot,fill=black, label=above:$\dots$] (s_dum) at (3,2) {};
        \node[slot,fill=black, label=above:$s_{n-6}$] (s_4) at (4,2) {};
        \node[slot,fill=black, label=above:$s_{n-5}$] (s_5) at (5,2) {};
        \node[slot,fill=black, label=above:$s_{n-4}$] (s_6) at (6,2) {};
        \node[slot,fill=black, label=above:$s_{n-3}$] (s_7) at (7,2) {};
        \node[slot,fill=black, label=above:$s_{n-2}$] (s_8) at (8,2) {};
        \node[slot,fill=black, label=above:$s_{n-1}$] (s_9) at (9,2) {};
        \node[slot,fill=black, label=above:$s_{n}$] (s_10) at (10,2) {};
        
        \draw[black,thick] (x_1) edge (s_1);
        \draw[black,thick] (x_2) edge (s_1);
        \draw[black,thick] (x_3) edge (s_2);
        \draw[black,thick] (x_4) edge (s_2);
        \draw[black,thick] (x_3) edge (s_3);
        \draw[black,thick] (x_4) edge (s_3);
        \draw[black,thick] (x_5) edge (s_4);
        \draw[black,thick] (x_6) edge (s_4);
        \draw[black,thick] (x_5) edge (s_5);
        \draw[black,thick] (x_6) edge (s_5);
        \draw[black,thick] (x_7) edge (s_6);
        \draw[black,thick] (x_8) edge (s_6);
        \draw[black,thick] (x_7) edge (s_7);
        \draw[black,thick] (x_8) edge (s_7);
        \draw[black,thick] (x_9) edge (s_8);
        \draw[black,thick] (x_10) edge (s_8);
        \draw[black,thick] (x_9) edge (s_9);
        \draw[black,thick] (x_10) edge (s_9);

        \draw[blue,thick] (x_1) edge (s_10);
        \draw[blue,thick] (x_2) edge (s_10);
    \end{tikzpicture}
    \caption{The request sequence is $\{x_{n-1},x_n\},
    \{x_{n-1},x_n\},\allowbreak \{x_{n-3},x_{n-2}\},\allowbreak
    \{x_{n-3},x_{n-2}\}, \dots, \{x_1,x_2\}, \{x_1,x_2\}$. The last
    node crosses all others, resulting in roughly $4n$
    crossings compared to $\frac{n}{2}$ crossings in the optimal case.}
	\label{fig:barybad}
\end{figure}

Algorithm~\ref{alg:barycenter} is no better than 8-competitive, as the following simple example
illustrated in~\Cref{fig:barybad} shows. If we request the sequence $\{x_{n-1},x_n\},\allowbreak
\{x_{n-1},x_n\},\allowbreak \{x_{n-3},x_{n-2}\},\allowbreak \{x_{n-3},x_{n-2}\}, \dots, \{x_1,x_2\}, \{x_1,x_2\}$,
the two first requests are placed on slots $s_{n-1},s_{n-2}$ and then consecutively, the following
requests occupy slots to the left of those until the last request, which is assigned the only
available slot $s_n$. The last pair of edges crosses all others, resulting in roughly $4n$
crossings compared to $\frac{n}{2}$ crossings in the optimal case.

In order to achieve a good upper bound for the OSCM-2, we present
\Cref{alg:minallcrossings} that given a request, selects the slot
which minimizes the total number of crossings -- including crossings between edges and propagation
arrows -- among all available slots.
\begin{algorithm}[tb]
    \begin{algorithmic}[1]
        \State{\textit{free\_slots} = $\{1,\dots,n\}$;}
        \For{\textit{element} in \textit{input}}
            \State{\textit{least\_crossings} $:=\infty$;}
            \State{\textit{best\_slot} $:=0$;}
            \For{\textit{slot} in \textit{free\_slots}}
                \State{\textit{$G$.simulate\_node\_insertion}(\textit{slot}, \textit{element});}
                \State{\textit{new\_crossings}$=$\textit{$G$.edge\_edge\_crossings}() + \textit{$G$.edge\_prop\_crossings}();}
                \If{\textit{new\_crossings} $<$ \textit{least\_crossings}}
                    \State{\textit{least\_crossings} $=$ \textit{new\_crossings};}
                    \State{\textit{best\_slot} $=$ \textit{slot};}
                \EndIf
                \State{\textit{$G$.revert\_simulated\_insertion}(\textit{slot}, \textit{element});}
            \EndFor
            \State{\textit{$G$.insert\_node}(\textit{best\_slot}, \textit{element});}
            \State{\textit{free\_slots} $:=$ \textit{free\_slots} $\setminus$ \textit{best\_slot};}
        \EndFor
    \end{algorithmic}
    \caption{Chooses in each step  the insertion with the lowest number of additional edge-edge and edge-propagation arrow crossings.}
    \label{alg:minallcrossings}
\end{algorithm}

Note that analyzing an algorithm in this setting is not completely
trivial. Our approach is to show that the types of crossings between
two requests produced by our algorithm are good-natured. Specifically,
we look at pairs of requests for which the crossings can be completely avoided 
if they are appropriately ordered, i.e., 3-0 or 4-0 crossings as depicted in 
Figure~\ref{fig:node_exchange_uncritical} (c) and (e)
respectively.
This type of crossing, then, is either not
produced by Algorithm~\ref{alg:minallcrossings} or we can show that a number of unavoidable
crossings is necessary to produce this configuration. With this, we
can then upper bound the competitive ratio. Note that this relatively
rough estimate is most likely an overestimation of the actual
competitive ratio of the algorithm, but even such an estimate 
already requires a lot of structural analysis.

First, we present some lemmata outlining some relevant structural properties
of assignments made by \Cref{alg:minallcrossings}, then we consider each
type of critical crossing, 4-0 crossings and then 3-0 crossings and show 
that the competitive ratio is still bounded when these types of crossings appear.

\subsection{Structural Properties}
To start the analysis of \Cref{alg:minallcrossings} we first make a
few observations on the changes of the propagation arrows after a
request is fulfilled.
 
Consider a request $\{x_1,x_2\}$, which is assigned to slot $s_x$ by some algorithm. 
Before this request arrived, there were two propagation arrows from vertices $y_1$ and $y_2$ going to slot $s_x$ (note that it is possible that $y_1=y_2$). 
After the request is assigned to $s_x$ the propagation arrows pointing to $s_x$ have to be shifted, as slot is not available anymore. 
Simultaneously, one propagation arrow of each $x_1$ and $x_2$ disappears as the request is fulfilled. 
The rest of the propagation arrows have to reflect this movement out of $s_x$ and into the two empty positions
left by $x_1$ and $x_2$, and they do so in the following way. 
 
  \begin{figure}[h]
 \begin{center}
 \subfigure{
  \begin{tikzpicture}[node distance=1cm,
  slot/.style={draw,rectangle},
  vertex/.style={draw,circle},
  scale=0.8,every node/.style={scale=0.8}
  ]
      \node[vertex, label=below:$v_{i_1}$] (v_{i_1}) at (0,0) {};
      \node[vertex] (van1) at (1,0) {};
      \node[vertex, label=below:$v_{k_1}$] (v_{k_1}) at (2,0) {};
      \node[vertex, label=below:$v_{k_2}$] (v_{k_2}) at (3,0) {};
      \node[vertex] (van2) at (4,0) {};
      \node[vertex, label=below:$v_{i_2}$] (v_{i_2}) at (5,0) {};
      \node[slot,label=above:$s_j$] (s_j) at (2.5,2) {};
      \node[slot] (san1) at (1,2) {};
      \node[slot] (san2) at (4,2) {};
         \draw[blue,thick,->] (v_{i_1}) -- (san1);
         \draw[blue,thick,->] (van1) -- (san1);
         \draw[blue,thick,->] (v_{k_1}) -- (s_j);
         \draw[blue,thick,->] (v_{k_2}) -- (s_j);
         \draw[blue,thick,->] (van2) -- (san2);
         \draw[blue,thick,->] (v_{i_2}) -- (san2);
  \end{tikzpicture}
  }
  \hspace*{1cm}
  \subfigure{
  \begin{tikzpicture}[node distance=1cm,
  slot/.style={draw,rectangle},
  vertex/.style={draw,circle},
  scale=0.8,every node/.style={scale=0.8}
  ]
      \node[vertex, label=below:$v_{i_1}$] (v_{i_1}) at (0,0) {};
      \node[vertex] (van1) at (1,0) {};
      \node[vertex, label=below:$v_{k_1}$] (v_{k_1}) at (2,0) {};
      \node[vertex, label=below:$v_{k_2}$] (v_{k_2}) at (3,0) {};
      \node[vertex] (van2) at (4,0) {};
      \node[vertex, label=below:$v_{i_2}$] (v_{i_2}) at (5,0) {};
      \node[slot, fill=black,  label=above:$s_j$] (s_j) at (2.5,2) {};
      \node[slot] (san1) at (1,2) {};
      \node[slot] (san2) at (4,2) {};
         \draw[thick] (v_{i_1}) -- (s_j);
         \draw[blue,thick,->] (van1) -- (san1);
         \draw[blue,thick,->] (v_{k_1}) -- (san1);
         \draw[blue,thick,->] (v_{k_2}) -- (san2);
         \draw[blue,thick,->] (van2) -- (san2);
         \draw[thick] (v_{i_2}) -- (s_j);
  \end{tikzpicture}
  }
  \caption{Schematic diagram showing how propagation arrows shift after a placement.}
  \label{fig:shiftingarrows}
  \end{center}
 \end{figure}
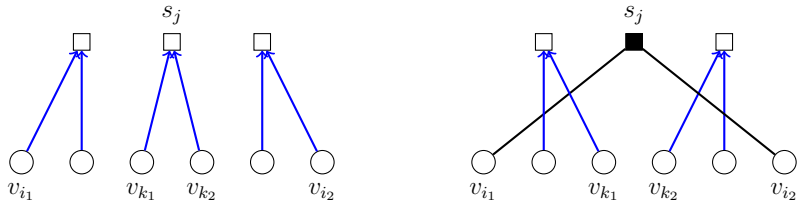
 
\begin{observation}\label{obs:shiftingarrows}
  Let $R=\{x_1,x_2\}$ be a request assigned to slot $s_x$.
  And let $y_1 \le y_2$ be the vertices (or vertex) whose propagation arrows point to $s_x$ before this request arrived. 
  Only propagation arrows connected to nodes between the leftmost vertex of $x_1$ and $y_1$ and the rightmost vertex of $x_2$ and $y_2$ will be shifted.
\end{observation}

Observe that there are no propagation arrows connected to nodes
between $y_1$ and $y_2$ as otherwise these would be connected
to slots other than $s_x$ and produce crossings between propagation
arrows, which is impossible by definition. 
The observation can be seen through Figure~\ref{fig:shiftingarrows}. 
 
 \begin{proof}
  Let $t$ be the amount of propagation arrows attached to nodes in the interval between the leftmost vertex of $x_1$ and $y_1$ and the rightmost vertex of $x_2$ and $y_2$ before $R$ is assigned to $s_x$.
  After the placement, the number of propagation arrows in $x_1$ and $x_2$ is reduced by 1. 
  The number of slots that require two arrows has been reduced by 1 in $s_x$. 
  If $t=2$, then $x_1=y_1$, $x_2=y_2$ and the interval has no remaining propagation arrows, after $R$ is placed. 
  Otherwise, each available slot has to be matched to each available vertex and no additional propagation arrows from outside the interval are required because the placement of $R$ removes two propagation arrows and one slot. 
 \end{proof}

\begin{figure}[t]
	\begin{center}
		\subfigure{
			\begin{tikzpicture}[node distance=1cm,
			slot/.style={draw,rectangle},
			vertex/.style={draw,circle},
			scale=0.8,every node/.style={scale=0.8}
			]
			\foreach \t/\label in {0/v_2,1/v_3,2/v_4}{
				\node[vertex, label=below:$\label$] (w_\t) at (\t,0) {};
			}
			\node[slot,fill=black, label=above:$s_1$] (s_1) at (1,2) {};
			\node[slot, label=above:$s_2$] (s_2) at (2,2) {};

			\draw[blue,thick,->] (w_0) edge [bend left=20](s_2);
			\draw[blue,thick,->] (w_0) edge [bend right=20](s_2);
			
			\draw[thick] (s_1) edge (w_1);
			\draw[thick] (s_1) edge (w_2);
			
			\node[] (a) at (1,-1.5) {(a)};
			\end{tikzpicture}
		}
		\subfigure{
		\begin{tikzpicture}[node distance=1cm,
		slot/.style={draw,rectangle},
		vertex/.style={draw,circle},
		scale=0.8,every node/.style={scale=0.8}
		]
		\foreach \t/\label in {-1/v_1,0/v_2,1/v_3,2/v_4}{
			\node[vertex, label=below:$\label$] (w_\t) at (\t,0) {};
		}
		\node[slot,fill=black, label=above:$s_1$] (s_1) at (1,2) {};
		\node[slot, label=above:$s_2$] (s_2) at (2,2) {};
		\node[slot, label=above:$s_3$] (s_3) at (3,2) {};

		\draw[blue,thick,->] (w_0) edge (s_2);
		\draw[blue,thick,->] (w_-1) edge (s_2);
		\draw[blue,thick,->] (w_0) edge (s_3);
		
		\draw[thick] (s_1) edge (w_1);
		\draw[thick] (s_1) edge (w_2);
		\node[] (a) at (0.5,-1.5) {(b)};
		\end{tikzpicture}
		}
		\subfigure{
			\begin{tikzpicture}[node distance=1cm,
			slot/.style={draw,rectangle},
			vertex/.style={draw,circle},
			scale=0.8,every node/.style={scale=0.8}
			]
			\foreach \t/\label in {-1/v_1,0/v_2,1/v_3,2/v_4}{
				\node[vertex, label=below:$\label$] (w_\t) at (\t,0) {};
			}
			\node[slot,fill=black, label=above:$s_1$] (s_1) at (1,2) {};
			\node[slot, label=above:$s_2$] (s_2) at (2,2) {};

			\draw[blue,thick,->] (w_0) edge (s_2);
			\draw[blue,thick,->] (w_-1) edge (s_2);
			
			\draw[thick] (s_1) edge (w_1);
			\draw[thick] (s_1) edge (w_2);
			\node[] (a) at (0.5,-1.5) {(c)};
			\end{tikzpicture}
		}
		\caption{Types of crossings avoided by \Cref{alg:minallcrossings}, mentioned in \Cref{lem:PropagationCrossEdge}. The propagation arrows are drawn in blue and the edges already present in the graph are drawn in black.}
		\label{fig:afterCrossing}
	\end{center}
\end{figure}
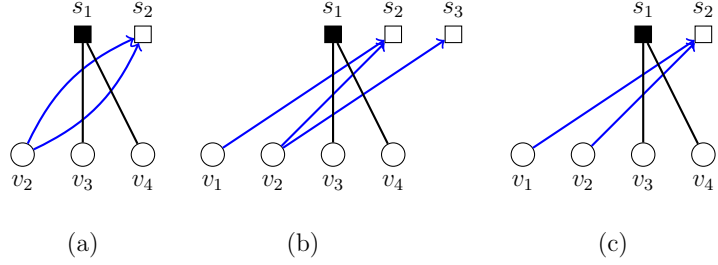

While \Cref{obs:shiftingarrows} is not specific to Algorithm~\ref{alg:minallcrossings}, we
can use it in the proofs to come. We continue with a lemma
that allows us to shorten a lot of case distinctions in the following proofs.

\begin{lemma}\label{lem:PropagationCrossEdge}
There is no instance during which two propagation arrows connected to a slot $s_2$ cross both edges adjacent to a fulfilled slot $s_1$ when using \Cref{alg:minallcrossings}.
\end{lemma}

Alternatively, the situations depicted in \Cref{fig:afterCrossing} will never occur if one uses \Cref{alg:minallcrossings}.
\begin{proof}
We prove the lemma by contradiction and assume that after \Cref{alg:minallcrossings} fulfills a request $R_x=\{x_1,x_2\}$ there are two propagation arrows crossing edges $(v_3,s_2)$ and $(v_4,s_2)$. 
    \Cref{fig:afterCrossing} shows three different situations how these crossings can occur: 
(a) Either both propagation arrows are connected to a single node $v_2 < v_3$ that cross the edges of $s_1$,
(b) there is a propagation arrow from two nodes $v_1 < v_2$ crossing the nodes of $s_1$ to a slot $s_2$ with $s_1 < s_2$ and another edge from $v_2$ to a slot $s_3$ with $s_1 < s_2 < s_3$ or
(c) there are the edges of (b) without the additional edge $(v_2,s_3)$.
Cases (a) and (b) are very similar, but in case (a) both propagation arrows from $v_2$ go to the same slot, whereas in case (b) they are split between $s_2$ and $s_3$.
In case (c) the vertex $v_2$ is already assigned an edge, thus it only has one remaining propagation arrow. We ignore this already present edge, as its precise nature makes no difference for the following argumentation.

We assume that the propagation arrows from $v_2$ (and possibly $v_1$) are the first ones that cross the edges of $s_1$ as described in the lemma after request $R_x$ has been fulfilled.
It is possible that there are vertices between $v_1$ and $v_2$ or between $v_2$ and $v_3$. However, if these vertices exist, they cannot have propagation arrows. Otherwise, $v_2$ (and possibly $v_1$) would not be responsible for the first two propagation arrows that cross $s_1$, but the propagation arrows of these other nodes.
We look at the first request $R_x$ whose assignment results in such a structure and how the graph looked like before serving $R_x$. 

Note that every slot has two propagation arrows pointing to it and after assigning a request to this slot, the propagation arrows pointing to that slot move to a neighboring free slot. 
Thus, there are four different configurations possible before the request $R_x$ is fulfilled, presented in \Cref{fig:beforeCrossing}: 
(a) Both arrows are connected to a single node $v_2 < v_3$ and a slot $s_0 < s_1$,
(b) the two arrows are from different nodes $v_2$ and $v_1$ and are connected to a slot $s_0 < s_1$,
(c) Both arrows are connected to $v_2$, one of them pointing to $s_0$ and one to $s_2$ with $s_0 < s_1 < s_2$
(d) the two arrows are from different nodes $v_1$ and $v_2$, one of them points to $s_0$ and one to $s_2$ with $s_0 < s_1 < s_2$. 

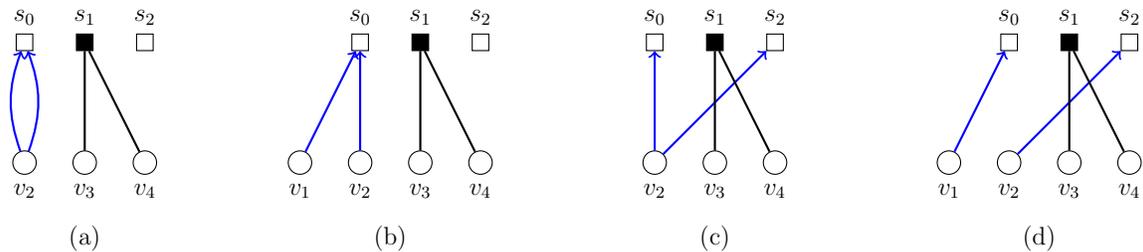
\begin{figure}
  \begin{center}
    \subfigure{
	    \begin{tikzpicture}[node distance=1cm,
	    slot/.style={draw,rectangle},
	    vertex/.style={draw,circle},
	    scale=0.8,every node/.style={scale=0.8}
	    ]
	    \foreach \t/\label in {0/v_2,1/v_3,2/v_4}{
		    \node[vertex, label=below:$\label$] (w_\t) at (\t,0) {};
	    }
	    \node[slot, label=above:$s_0$] (s_0) at (0,2) {};
	    \node[slot,fill=black, label=above:$s_1$] (s_1) at (1,2) {};
	    \node[slot, label=above:$s_2$] (s_2) at (2,2) {};

	    \draw[blue,thick,->] (w_0) edge [bend left=20](s_0);
	    \draw[blue,thick,->] (w_0) edge [bend right=20](s_0);
	    
	    \draw[thick] (s_1) edge (w_1);
	    \draw[thick] (s_1) edge (w_2);
	    
	    \node[] (a) at (1,-1.25) {(a)};
	    \end{tikzpicture}
    }
    \hspace*{1cm}
    \subfigure{
	    \begin{tikzpicture}[node distance=1cm,
	    slot/.style={draw,rectangle},
	    vertex/.style={draw,circle},
	    scale=0.8,every node/.style={scale=0.8}
	    ]
	    \foreach \t/\label in {-1/v_1,0/v_2,1/v_3,2/v_4}{
		    \node[vertex, label=below:$\label$] (w_\t) at (\t,0) {};
	    }
	    \node[slot, label=above:$s_0$] (s_0) at (0,2) {};
	    \node[slot,fill=black, label=above:$s_1$] (s_1) at (1,2) {};
	    \node[slot, label=above:$s_2$] (s_2) at (2,2) {};
	    
	    \draw[blue,thick,->] (w_0) edge (s_0);
	    \draw[blue,thick,->] (w_-1) edge (s_0);
	    
	    \draw[thick] (s_1) edge (w_1);
	    \draw[thick] (s_1) edge (w_2);
	    \node[] (a) at (0.5,-1.25) {(b)};
	    \end{tikzpicture}
    }
	\hspace*{1.25cm}
    \subfigure{
	    \begin{tikzpicture}[node distance=1cm,
	    slot/.style={draw,rectangle},
	    vertex/.style={draw,circle},
	    scale=0.8,every node/.style={scale=0.8}
	    ]
	    \foreach \t/\label in {0/v_2,1/v_3,2/v_4}{
		    \node[vertex, label=below:$\label$] (w_\t) at (\t,0) {};
	    }
	    \node[slot, label=above:$s_0$] (s_0) at (0,2) {};
	    \node[slot,fill=black, label=above:$s_1$] (s_1) at (1,2) {};
	    \node[slot, label=above:$s_2$] (s_2) at (2,2) {};

	    \draw[blue,thick,->] (w_0) edge [](s_0);
	    \draw[blue,thick,->] (w_0) edge [](s_2);
	    
	    \draw[thick] (s_1) edge (w_1);
	    \draw[thick] (s_1) edge (w_2);
	    
	    \node[] (a) at (1,-1.25) {(c)};
	    \end{tikzpicture}
    }
    \hspace*{1.25cm}
    \subfigure{
	    \begin{tikzpicture}[node distance=1cm,
	    slot/.style={draw,rectangle},
	    vertex/.style={draw,circle},
	    scale=0.8,every node/.style={scale=0.8}
	    ]
	    \foreach \t/\label in {-1/v_1,0/v_2,1/v_3,2/v_4}{
		    \node[vertex, label=below:$\label$] (w_\t) at (\t,0) {};
	    }
	    \node[slot, label=above:$s_0$] (s_0) at (0,2) {};
	    \node[slot,fill=black, label=above:$s_1$] (s_1) at (1,2) {};
	    \node[slot, label=above:$s_2$] (s_2) at (2,2) {};
	    
	    \draw[blue,thick,->] (w_-1) edge (s_0);
	    \draw[blue,thick,->] (w_0) edge (s_2);
	    
	    \draw[thick] (s_1) edge (w_1);
	    \draw[thick] (s_1) edge (w_2);
	    \node[] (a) at (0.5,-1.25) {(d)};
	    \end{tikzpicture}
    }
    \caption{Possible configuration before the request $R_x$ is added and the propagation arrows are shifted to $s_2$. The propagation arrows are drawn in blue and the edges already present in the graph are drawn in black.}
    \label{fig:beforeCrossing}
  \end{center}
\end{figure}

We know also by using Observation~\ref{obs:shiftingarrows}, that the assignment of $R_x$ will only shift the propagation arrows around $s_1$ if these arrows are part of the affected interval between the vertices of $R_x$ and the propagation arrows pointing to the slot assigned to $R_x$. 

Cases (a) and (b) have no previous arrows crossing with edges of $s_1$, thus, they require that two propagation arrows are shifted to the right hand side.
As we saw in Observation~\ref{obs:shiftingarrows}, this can only happen if $R_x$ is assigned to $s_0$ or to the left hand side of it and the vertices $x_1$ and $x_2$ are both to the right of $v_2$.
The propagation arrows of $x_1$ and $x_2$ point previously to a slot to the right of $s_0$, thus, assigning $R_x$ to $s_0$ will result in the two propagation arrows previously pointing to $s_0$ shifting to the right to fill up the slots left by the missing propagation arrows of $x_1$ and
$x_2$. These gaps are filled from the left hand side, which results in the two crossing propagation arrows shown in~\Cref{fig:afterCrossing}. 

Assume that Algorithm~\ref{alg:minallcrossings}, given case \ref{fig:beforeCrossing}(a) or (b) and a request $R_x$ with $v_2 < x_1 < x_2$, assigns $R_x$ to a free slot to the left hand side of $s_1$. 
\Cref{fig:wrongPlacement} shows that assigning $R_x$ more to the right results in fewer crossings, which is a contradiction to the 
procedure of the algorithm itself. 
If the vertices $x_1$ and $x_2$ do not coincide with $v_3$ and $v_4$, they are even more to the right hand side.
If this is the case, we get even more crossings if $R_x$ is assigned to the left hand side of $s_1$. 

\begin{figure}
\begin{center}
  \subfigure{
	  \begin{tikzpicture}[node distance=1cm,
	  slot/.style={draw,rectangle},
	  vertex/.style={draw,circle},
	  scale=0.8,every node/.style={scale=0.8}
	  ]
	  \foreach \t/\label in {0/v_2,1/v_3,2/v_4}{
		  \node[vertex, label=below:$\label$] (w_\t) at (\t,0) {};
	  }
	  \node[slot, label=above:$s_0$] (s_0) at (0,2) {};
	  \node[slot,fill=black, label=above:$s_1$] (s_1) at (1,2) {};
	  \node[slot, label=above:$s_2$] (s_2) at (2,2) {};

	  \draw[blue,thick,->] (w_0) edge [bend left=20](s_2);
	  \draw[blue,thick,->] (w_0) edge [bend right=20](s_2);
	  
	  \draw[thick] (s_1) edge (w_1);
	  \draw[thick] (s_1) edge (w_2);
	  
	  \draw[red, thick] (s_0) edge (w_1);
	  \draw[red, thick] (s_0) edge (w_2);
	  \end{tikzpicture}
  }
  \hspace*{1.25cm}
  \subfigure{
	  \begin{tikzpicture}[node distance=1cm,
	  slot/.style={draw,rectangle},
	  vertex/.style={draw,circle},
	  scale=0.8,every node/.style={scale=0.8}
	  ]
	  \foreach \t/\label in {0/v_2,1/v_3,2/v_4}{
		  \node[vertex, label=below:$\label$] (w_\t) at (\t,0) {};
	  }
	  \node[slot, label=above:$s_0$] (s_0) at (0,2) {};
	  \node[slot,fill=black, label=above:$s_1$] (s_1) at (1,2) {};
	  \node[slot, label=above:$s_2$] (s_2) at (2,2) {};

	  \draw[blue,thick,->] (w_0) edge [bend left=20](s_0);
	  \draw[blue,thick,->] (w_0) edge [bend right=20](s_0);
	  
	  \draw[thick] (s_1) edge (w_1);
	  \draw[thick] (s_1) edge (w_2);
	  
	  \draw[red, thick] (s_2) edge (w_1);
	  \draw[red, thick] (s_2) edge (w_2);
	  \end{tikzpicture}
  }
	\hspace*{1.25cm}
  \subfigure{
	  \begin{tikzpicture}[node distance=1cm,
	  slot/.style={draw,rectangle},
	  vertex/.style={draw,circle},
	  scale=0.8,every node/.style={scale=0.8}
	  ]
	  \foreach \t/\label in {-1/v_1,0/v_2,1/v_3,2/v_4}{
		  \node[vertex, label=below:$\label$] (w_\t) at (\t,0) {};
	  }
	  \node[slot, label=above:$s_0$] (s_0) at (0,2) {};
	  \node[slot,fill=black, label=above:$s_1$] (s_1) at (1,2) {};
	  \node[slot, label=above:$s_2$] (s_2) at (2,2) {};
	  
	  \draw[blue,thick,->] (w_0) edge (s_2);
	  \draw[blue,thick,->] (w_-1) edge (s_2);
	  
	  \draw[thick] (s_1) edge (w_1);
	  \draw[thick] (s_1) edge (w_2);
	  
	  \draw[red, thick] (s_0) edge (w_1);
	  \draw[red, thick] (s_0) edge (w_2);
	  \end{tikzpicture}
  }
  \hspace*{1.25cm}
  \subfigure{
	  \begin{tikzpicture}[node distance=1cm,
	  slot/.style={draw,rectangle},
	  vertex/.style={draw,circle},
	  scale=0.8,every node/.style={scale=0.8}
	  ]
	  \foreach \t/\label in {-1/v_1,0/v_2,1/v_3,2/v_4}{
		  \node[vertex, label=below:$\label$] (w_\t) at (\t,0) {};
	  }
	  \node[slot, label=above:$s_0$] (s_0) at (0,2) {};
	  \node[slot,fill=black, label=above:$s_1$] (s_1) at (1,2) {};
	  \node[slot, label=above:$s_2$] (s_2) at (2,2) {};
	  
	  \draw[blue,thick,->] (w_0) edge (s_0);
	  \draw[blue,thick,->] (w_-1) edge (s_0);
	  
	  \draw[thick] (s_1) edge (w_1);
	  \draw[thick] (s_1) edge (w_2);
	  
	  \draw[red, thick] (s_2) edge (w_1);
	  \draw[red, thick] (s_2) edge (w_2);
	  \end{tikzpicture}
  }
  \caption{Comparing the crossings of assigning $R_x$ to the left or right hand side of $s_1$ for the cases (a) and (b) from \Cref{fig:beforeCrossing}. The propagation arrows are blue, already present edges are black and the newly introduced edges, adjacent to the recently fulfilled request $R_x$, are red. }
  \label{fig:wrongPlacement}
\end{center}
\end{figure}
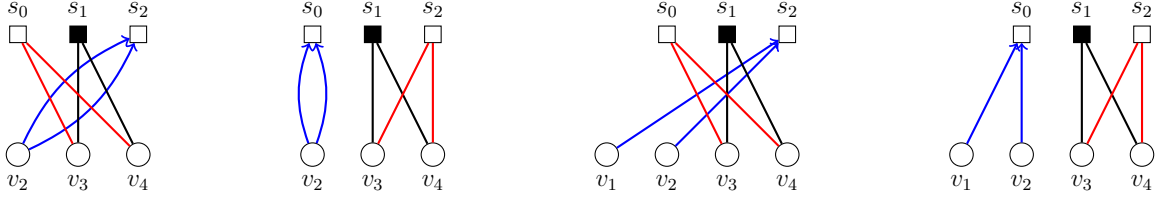

For the cases (c) and (d) from \Cref{fig:beforeCrossing} only one propagation arrow needs to be pushed to the right hand side. 
Thus, w.l.o.g. only $x_2$ has to be to the right hand side of $v_2$ and
the position of $x_1$ is arbitrary. 
Either $x_1$ is a vertex to the left hand side of $v_2$ (it is even possible that $x_1=v_1$) or it is to the right. 
The latter case is equivalent to the cases (a) and (b) in the sense that two more propagation arrows will cross over $s_1$ and a placement
to the right of $s_1$ will result in less crossings as we saw in \Cref{fig:wrongPlacement}. 

In the first case, on the other hand, $x_1$ is to the left of $v_2$, and we push only one more propagation arrow to the right
hand side of $s_1$. \Cref{fig:wrongPlacement2} shows that choosing the position $s_2$ to the right of $s_1$ results in fewer crossings. 
Just as with cases (a) and (b) we can assume that $x_2$ is the leftmost possible vertex to the right of $v_2$ and otherwise the number of avoided crossings with the placement to the left of $s_1$ only grows. 

Thus, \Cref{alg:minallcrossings} will not place a request such that two edges of a slot are crossed by two propagation arrows.
\end{proof}
	
\begin{figure}[h]
\begin{center}
  \subfigure{
	  \begin{tikzpicture}[node distance=1cm,
	  slot/.style={draw,rectangle},
	  vertex/.style={draw,circle},
	  scale=0.8,every node/.style={scale=0.8}
	  ]
	  \foreach \t/\label in {0/v_2,1/v_3,2/v_4}{
		  \node[vertex, label=below:$\label$] (w_\t) at (\t,0) {};
	  }
	  \node[slot, label=above:$s_0$] (s_0) at (0,2) {};
	  \node[slot,fill=black, label=above:$s_1$] (s_1) at (1,2) {};
	  \node[slot, label=above:$s_2$] (s_2) at (2,2) {};

	  \draw[blue,thick,->] (w_0) edge [bend left=20](s_2);
	  \draw[blue,thick,->] (w_0) edge [bend right=20](s_2);
	  
	  \draw[thick] (s_1) edge (w_1);
	  \draw[thick] (s_1) edge (w_2);
	  
	  \draw[red, thick] (s_0) edge (w_1);
	  \draw[red, thick] (s_0) edge (-1,0);
	  \end{tikzpicture}
  }
  \hspace*{1cm}
  \subfigure{
	  \begin{tikzpicture}[node distance=1cm,
	  slot/.style={draw,rectangle},
	  vertex/.style={draw,circle},
	  scale=0.8,every node/.style={scale=0.8}
	  ]
	  \foreach \t/\label in {0/v_2,1/v_3,2/v_4}{
		  \node[vertex, label=below:$\label$] (w_\t) at (\t,0) {};
	  }
	  \node[slot, label=above:$s_0$] (s_0) at (0,2) {};
	  \node[slot,fill=black, label=above:$s_1$] (s_1) at (1,2) {};
	  \node[slot, label=above:$s_2$] (s_2) at (2,2) {};

	  \draw[blue,thick,->] (w_0) edge [bend left=20](s_0);
	  \draw[blue,thick,->] (w_0) edge [bend right=20](s_0);
	  
	  \draw[thick] (s_1) edge (w_1);
	  \draw[thick] (s_1) edge (w_2);
	  
	  \draw[red, thick] (s_2) edge (w_1);
	  \draw[red, thick] (s_2) edge (-1,0);
	  \end{tikzpicture}
  }
\hspace*{1cm}
  \subfigure{
	  \begin{tikzpicture}[node distance=1cm,
	  slot/.style={draw,rectangle},
	  vertex/.style={draw,circle},
	  scale=0.8,every node/.style={scale=0.8}
	  ]
	  \foreach \t/\label in {-1/v_1,0/v_2,1/v_3,2/v_4}{
		  \node[vertex, label=below:$\label$] (w_\t) at (\t,0) {};
	  }
	  \node[slot, label=above:$s_0$] (s_0) at (0,2) {};
	  \node[slot,fill=black, label=above:$s_1$] (s_1) at (1,2) {};
	  \node[slot, label=above:$s_2$] (s_2) at (2,2) {};
	  
	  \draw[blue,thick,->] (w_-1) edge (s_2);
	  \draw[blue,thick,->] (w_0) edge (s_2);
	  
	  \draw[thick] (s_1) edge (w_1);
	  \draw[thick] (s_1) edge (w_2);
	  
	  \draw[red, thick] (s_0) edge (w_1);
	  \draw[red, thick] (s_0) edge (w_-1);
	  \end{tikzpicture}
  }
  \hspace*{1cm}
  \subfigure{
	  \begin{tikzpicture}[node distance=1cm,
	  slot/.style={draw,rectangle},
	  vertex/.style={draw,circle},
	  scale=0.8,every node/.style={scale=0.8}
	  ]
	  \foreach \t/\label in {-1/v_1,0/v_2,1/v_3,2/v_4}{
		  \node[vertex, label=below:$\label$] (w_\t) at (\t,0) {};
	  }
	  \node[slot, label=above:$s_0$] (s_0) at (0,2) {};
	  \node[slot,fill=black, label=above:$s_1$] (s_1) at (1,2) {};
	  \node[slot, label=above:$s_2$] (s_2) at (2,2) {};
	  
	  \draw[blue,thick,->] (w_-1) edge (s_0);
	  \draw[blue,thick,->] (w_0) edge (s_0);
	  
	  \draw[thick] (s_1) edge (w_1);
	  \draw[thick] (s_1) edge (w_2);
	  
	  \draw[red, thick] (s_2) edge (w_1);
	  \draw[red, thick] (s_2) edge (w_-1);
	  \end{tikzpicture}
  }
  \caption{Comparing the crossings for assigning $R_x$ to the left or right hand side of $s_1$, for the cases (c) and (d) from \Cref{fig:beforeCrossing}. The propagation arrows are blue, already present edges are black and the newly introduced edges, adjacent to $R_x$, are red. }
  \label{fig:wrongPlacement2}
\end{center}
\end{figure}
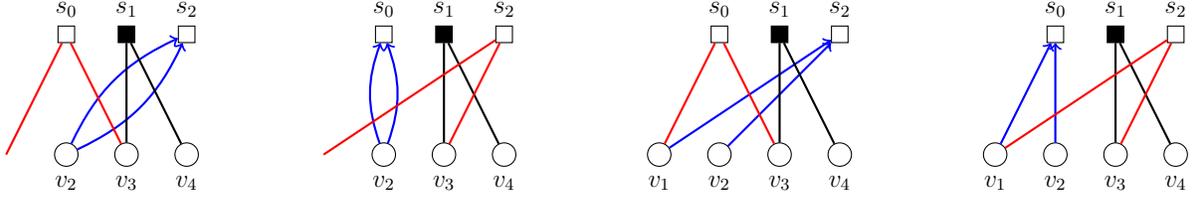

Lemma~\ref{lem:PropagationCrossEdge} forbids specific configurations of the propagation arrows during the course of
applying Algorithm~\ref{alg:minallcrossings} to a request sequence.
The following lemma uses a counting argument to guarantee that a specific request between two (far apart)
vertices must eventually appear in a specific setting. Such requests from vertices that are far apart, always guarantee
the appearance of unavoidable crossings as depicted in Figure~\ref{fig:node_exchange_uncritical} (f). The appearance of
such requests guarantees, in later proofs, the existence of such unavoidable crossings, which can be counted in a way that
bounds the competitive ratio.
	
\begin{lemma}\label{lem:UnavoidableHouses} 
Let there be two request $\{x_1, x_2\}$ and $\{y_1, y_2\}$ that are assigned to slots $s_x$ and $s_y$, 
with $x_1 < x_2 < y_1 < y_2$ and no free slot between $s_x$ and $s_y$. 
If there are two neighboring vertices $u,v$, with $x_2 \leq u < v \leq y_1$ and propagation arrows pointing to two different slots $s_l, s_r$, with $s_l < s_x < s_y < s_r$, and the request $\{u,v\}$ appears, then there must be a future request $\{a,b\}$, with $a \leq x_2$ and $y_1 \leq b$, which unavoidably crosses all edges of $u$ and $v$. 
\end{lemma}
Figure~\ref{fig:houses} depicts the situation described in the statement  of Lemma~\ref{lem:UnavoidableHouses}.
\begin{figure}[h]
\begin{center}
 \begin{tikzpicture}[node distance=1cm,
 slot/.style={draw,rectangle},
 vertex/.style={draw,circle},
 scale=0.8,every node/.style={scale=0.8}
 ]
  \node[slot, label=above:$s_{\ell}$] (sl) at (0,2) {};
  \node[slot, label=above:$s_{x}$] (sx) at (1,2) {};
  \node[slot, label=above:$s_{y}$] (sy) at (2,2) {};
  \node[slot, label=above:$s_{r}$] (sr) at (3,2) {};
  \node[vertex, label=below:$x_1$] (x1) at (-1,0) {};
  \node[vertex, label=below:$x_2$] (x2) at (0,0) {};
  \node[vertex, label=below:$u$] (u) at (1,0) {};
  \node[vertex, label=below:$v$] (v) at (2,0) {};
  \node[vertex, label=below:$y_1$] (y1) at (3,0) {};
  \node[vertex, label=below:$y_2$] (y2) at (4,0) {};
  
  \draw[thick] (x1) -- (sx)
               (x2) -- (sx)
               (y1) -- (sy)
               (y2) -- (sy);
  \draw[blue,->] (u) -- (sl);
  \draw[blue,->] (v) -- (sr);
 \end{tikzpicture}

 \caption{Sketch of the situation described in the statement of Lemma~\ref{lem:UnavoidableHouses}.}
 \label{fig:houses}
\end{center}
\end{figure}
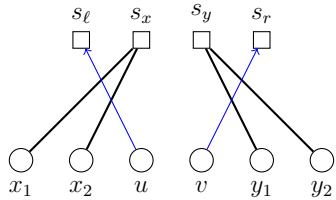
\begin{proof}
Our proof is a simple counting argument. 
The request $\{u,v\}$ removes two propagation arrows. 
One points to the left of the filled block between $s_x$ and $s_y$ and the other one points to the right of it. 
The request, depending on its placement, pushes one propagation arrow from one side of the fulfilled block between $s_x$ and $s_y$ to the other one. 

W.l.o.g. we assume that $\{u,v\}$ is placed on $s_l$. 
The second propagation arrow pointing to $s_l$ comes from $x_2$ (if $u \neq x_2$) or a vertex even more to the left. 
It is not possible that is comes from a vertex between $x_2$ and $v$ due to \Cref{lem:PropagationCrossEdge}. 
When the request $\{u,v\}$ is placed, it pushes this second propagation arrow to the slot $s_r$. 
This propagation arrow represents a mismatch between open slots and "open/remaining" edges. 
The number of "open/remaining" edges to the left of $u$ and to the right of $v$ is odd, but the slots always consume two of these "open/remaining" edges. 
This has to be compensated by some request $\{a,b\}$ that is placed right of $s_y$, where $a$ is to the left hand side of $u$ and $b$ is to the right hand side of $v$. 
This request crosses all edges of $u$ and $v$. 
\end{proof}

Where \Cref{lem:UnavoidableHouses,lem:PropagationCrossEdge} are applicable for specific configurations, the following lemma provides a tool that gives a set of edges or propagation arrows that are necessary to make a local configuration (e.g., a crossing of two requests) feasible in the context of the remaining graph. 

\begin{lemma}\label{lem:Equator}
	For every edge or propagation arrow, starting at a vertex $v_i$ of $V$ and pointing to a slot $s_j$ with $i < j$ (analogously $j<i$), there is one edge or propagation arrow pointing from a vertex $v_k$ to a slot $s_l$ with $i < k$ and $l \leq i$ (analogously $k<i$ and $i \leq l$). 
\end{lemma}
\begin{proof}
	We use a simple handshake argument and count the already present edges and the propagation arrows in the graph to prove the statement.
	
	At first, we separate the vertices into four sets, as depicted in \Cref{fig:Equator}. 
	The set $A$ contains the vertex $s_l$ and all vertices from $S$ that are to the left hand side of $s_l$. 
	The set $B$ contains the vertices from $S$ that are to the right hand side of $s_l$. 
	The set $C$ contains the vertex $v_i$ and all vertices from $V$ that are to the left hand side of $v_i$. 
	The last set, called $D$, contains the vertices from $V$ that are to the right hand side of $v_l$. 
	
	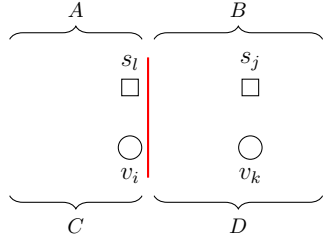
\begin{figure}[t]
		\begin{center}
			\begin{tikzpicture}[node distance=1cm,
			slot/.style={draw,rectangle},
			vertex/.style={draw,circle},
			scale=0.8,every node/.style={scale=0.8}
			]
			\node[vertex, label=below:$v_i$] (v1) at (0,0) {};
			\node[vertex, label=below:$v_k$] (v2) at (2,0) {};
			\node[slot, label=above:$s_l$] (s1) at (0,1) {};
			\node[slot, label=above:$s_j$] (s2) at (2,1) {};
			
			\draw[red, thick] (0.3,-0.5) edge (0.3,1.5);
			
			\draw[decorate,decoration={brace,amplitude=5pt,mirror}] (-2.0,-0.8) -- (0.2,-0.8) node [black,midway,yshift=-0.5cm] {\footnotesize $C$};
			\draw[decorate,decoration={brace,amplitude=5pt,mirror}] (0.4,-0.8) -- (3.2,-0.8) node [black,midway,yshift=-0.5cm] {\footnotesize $D$};
			
			\draw[decorate,decoration={brace,amplitude=5pt,mirror}] (0.2,1.8) -- (-2,1.8) node [black,midway,yshift=0.5cm] {\footnotesize $A$};
			\draw[decorate,decoration={brace,amplitude=5pt,mirror}] (3.2,1.8) -- (0.4,1.8) node [black,midway,yshift=0.5cm] {\footnotesize $B$};
			\end{tikzpicture}
			\caption{Depending on the vertex $v_i$, the vertices are split into four sets, $A,B,C$ and $D$.}
			\label{fig:Equator}
		\end{center}
	\end{figure}
	
	The vertices in the set $A$ have two incident edges or two incident propagation arrows. 
	These edges or propagation arrows start either at a vertex in $C$ or $D$. 
	We denote the set of edges that connect a vertex from $A$ with a vertex from $C$ as $E_{AC}$. 
	Analogously, we define the edge set $E_{AD}$. 
	We also split the propagation arrows, starting at the vertices from $V$ and ending at a vertex in $A$, into two sets, $P_{AC}$ and $P_{AD}$.  
	We can observe that 
	\begin{align}
	2|A| = E_{AC} + E_{AD} + P_{AC} + P_{AD} \label{eq:Equator1}
	\end{align} 
	must always be true. 
	
	Additionally, the sum of the edges and propagation arrows starting at a vertex in $C$ must be $2i$. 
	Or more formal, 
	\begin{align}
	2i= E_{AC} + E_{BC} + P_{AC} + P_{BC} \text{ .} \label{eq:Equator2}
	\end{align} 
	The number of vertices in the set $A$ must be $i$, because we choose the vertex $v_i$ at position $i$ as a reference point to define the set $A$. 
	Thus, we can combine \Cref{eq:Equator1} and \Cref{eq:Equator2} to obtain 
	\[
	0 = E_{AD} + P_{AD} - E_{BC} - P_{BC} \text{ .}
	\]

	Note, because propagation arrows never cross each other, either $P_{AD}$ or $P_{BC}$ is empty (it is also possible that both are empty). 
	Thus, every edge or propagation arrow crossing from one side to the other is compensated by an edge, crossing into the other direction. 
\end{proof}

With our structural properties and observations regarding the propagation arrows we can now start to analyze the critical crossings depicted in \Cref{fig:node_exchange_uncritical} (e) and (c). 
These crossings are critical in the sense that they have only avoidable crossings and no unavoidable ones. 
So, they decrease the performance of our algorithm and do not guarantee a constant competitive ratio like the other crossings depicted in \Cref{fig:node_exchange_uncritical}. 
In the following sections, we overcome this problem by showing that for each of these critical crossings there must exist some other request that unavoidably crosses one of the requests, involved in the critical crossing. 

\subsection{The 4-0 Crossings}
 
 Recall that, by Lemma~\ref{lem:node_exchange_uncritical}, 
 the optimal solution for a 2-regular instance of the online OSCM-2
 consists on minimizing crossings between every pair of requests. 
 Thus, we can look at a pair of requests and exhaustively classify them as depicted in 
 Figure~\ref{fig:node_exchange_uncritical}, and analyze the competitive ratio of an algorithm
 depending on how many of these types of crossings appear.
 In particular, if no 3-0 crossings (Figure~\ref{fig:node_exchange_uncritical}(c)) or
 4-0 crossings (Figure~\ref{fig:node_exchange_uncritical}(e)) were produced by an algorithm, 
 the algorithm would be 3-competitive at worst, as any sub-optimal placement would be
 trivially compensated by at least one unavoidable crossing.
 Thus, in order to analyze the competitive ratio of Algorithm~\ref{alg:minallcrossings},
 we only have to look at 3-0 and 4-0 crossings.

 Using Lemma~\ref{lem:PropagationCrossEdge} we can now prove that Algorithm~\ref{alg:minallcrossings} will not make too many mistakes when producing 4-0 crossings.
 First we prove that Algorithm~\ref{alg:minallcrossings} will never produce 4-0 crossings with gaps, i.e.,
 unfulfilled slots between the 2 slots generating the 4-0 crossing as depicted for instance in Figure~\ref{fig:40gaps}.
 
 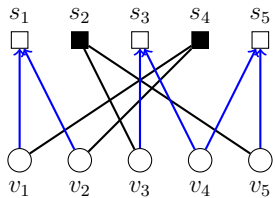
\begin{figure}[ht]
 \begin{center}
   \begin{tikzpicture}[node distance=1cm,
   slot/.style={draw,rectangle},
   vertex/.style={draw,circle},
   scale=0.8,every node/.style={scale=0.8}
   ]
      \foreach \t in {1,...,5}{
         \node[vertex, label=below:$v_\t$] (w_\t) at (\t,0) {};
      }
      \node[slot, label=above:$s_1$] (s_1) at (1,2) {};
      \node[slot,fill=black, label=above:$s_2$] (s_2) at (2,2) {};
      \node[slot, label=above:$s_3$] (s_3) at (3,2) {};
      \node[slot,fill=black, label=above:$s_4$] (s_4) at (4,2) {};
      \node[slot, label=above:$s_5$] (s_5) at (5,2) {};
      \draw[thick] (s_4) edge (w_1)
		   (s_4) edge (w_2)
		   (s_2) edge (w_5)
		   (s_2) edge (w_3);
		   
      \draw[blue,thick,->] (w_1) edge (s_1);
      \draw[blue,thick,->] (w_2) edge (s_1);
      \draw[blue,thick,->] (w_3) edge (s_3);
      \draw[blue,thick,->] (w_4) edge (s_3);
      \draw[blue,thick,->] (w_4) edge (s_5);
      \draw[blue,thick,->] (w_5) edge (s_5);
  \end{tikzpicture}
   \caption{A 4-0 crossing with a slot in between. These types of crossings are forbidden by Lemma~\ref{lem:40gaps}.}
   \label{fig:40gaps}
  \end{center}
 \end{figure}

 \begin{lemma}\label{lem:40gaps}
  Algorithm~\ref{alg:minallcrossings} never generates 4-0 crossings with gaps in between. More precisely,
  for each pair $s_i$, $s_j$ with $i< j$ assigned by Algorithm~\ref{alg:minallcrossings} that generate a 4-0 crossing,
  every $s_k$ with $i< k< j$ is already full.
 \end{lemma}
 
 \begin{proof}
  Let us assume that there are no 4-0 crossings with gaps in the graph yet.
  We prove this lemma by means of a contradiction.
  
  Let $\{v_1,v_2\}$ be the request assigned to slot $s_i$ by
  Algorithm~\ref{alg:minallcrossings}, and a new request $R = \{v_3,v_4\}$ is made
  where $v_1 < v_2 < v_3 < v_4$ without loss of generality.
  
  Let $s_j$ be a slot to the left of $s_i$ with the smallest crossing values for $R$ and
  let $s_k$ be the leftmost empty slot between $s_j$ and $s_i$.
  
  The only crossings that would make a placement in $s_k$ more unfavorable than a placement in $s_j$ are edges coming from 
  the right of $v_4$ to a slot between $s_j$ and $s_k$ as depicted in the left of Figure~\ref{fig:40gapscrossings}. 
  There cannot be any propagation arrows of this kind as we assume that all the slots between $s_j$ and $s_k$ are full. 
  
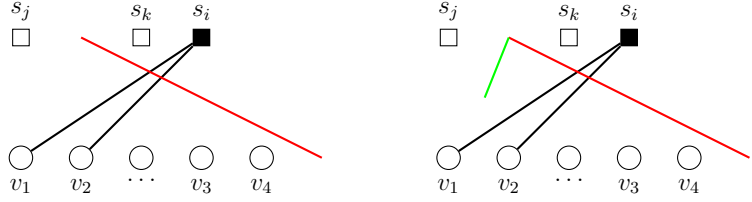
\begin{figure}[t]
    \begin{center}
    \subfigure{
        \begin{tikzpicture}[node distance=1cm,
        slot/.style={draw,rectangle},
        vertex/.style={draw,circle},
        scale=0.8,every node/.style={scale=0.8}
        ]
            \node[vertex, label=below:$v_1$] (v1) at (1,0) {};
            \node[vertex, label=below:$v_2$] (v2) at (2,0) {};
            \node[vertex, label=below:$\hdots$] (van) at (3,0) {};
            \node[vertex, label=below:$v_3$] (v3) at (4,0) {};
            \node[vertex, label=below:$v_4$] (v4) at (5,0) {};
            \node[slot, label=above:$s_j$] (si) at (1,2) {};
            \node[slot,fill=black, label=above:$s_i$] (sj) at (4,2) {};
            \node[slot, label=above:$s_k$] (sk) at (3,2) {};
            \draw[thick] (sj) edge (v1)
                        (sj) edge (v2);
            \draw[thick,red] (6,0) edge (2,2);
        \end{tikzpicture}
    }\hspace*{1cm}
    \subfigure{
        \begin{tikzpicture}[node distance=1cm,
        slot/.style={draw,rectangle},
        vertex/.style={draw,circle},
        scale=0.8,every node/.style={scale=0.8}
        ]
            \node[vertex, label=below:$v_1$] (v1) at (1,0) {};
            \node[vertex, label=below:$v_2$] (v2) at (2,0) {};
            \node[vertex, label=below:$\hdots$] (van) at (3,0) {};
            \node[vertex, label=below:$v_3$] (v3) at (4,0) {};
            \node[vertex, label=below:$v_4$] (v4) at (5,0) {};
            \node[slot, label=above:$s_j$] (si) at (1,2) {};
            \node[slot,fill=black, label=above:$s_i$] (sj) at (4,2) {};
            \node[slot, label=above:$s_k$] (sk) at (3,2) {};
            \draw[thick] (sj) edge (v1)
                        (sj) edge (v2);
            \draw[thick,red] (6,0) edge (2,2);
            \draw[thick,green] (1.6,1) edge (2,2);
        \end{tikzpicture}
    }
    \caption{Counting crossings in a 4-0 crossing with a gap in Lemma~\ref{lem:40gaps}.
    Only edges like the one depicted in red on top will  make a placement in $s_k$ more favorable than a placement in $s_j$.
    But for every red edge a green edge must exist or the graph is not free of 4-0 crossings.
    }
    \label{fig:40gapscrossings}
    \end{center}
\end{figure}

  For any edge coming from a vertex $v_{t_1}$ to the right of $v_4$ into slot $s_t$ with $j< t< k$
  there must be another edge coming from a vertex $v_{t_2}$ to the left of (or directly from) the vertex $v_2$. Otherwise
  we would have a 4-0 crossing with an empty slot, namely $v_1< v_2< v_{t_2}< v_{t_1}$ and the slots $s_t< s_k< s_i$,
  which would be a contradiction to the assumption that this is the first occurrence, as we can see in
  \Cref{fig:40gapscrossings}.
  Thus, this means that $v_{t_2}\leq v_2$.
  However, then the edge $v_{t_2}$ generates crossings only for the assignment of $R$ to $s_j$ and 
  not for the assignment to $s_k$, which means that for every crossing counting for $s_k$ there is at least one crossing
  counting for $s_j$.
  
  Finally we are only left to count the crossings for the propagation arrows going to $s_j$ with the placing in $s_k$
  and vice-versa as depicted in the three drawings of \Cref{fig:40gapscrossings2}.

\begin{figure}[h]
	\begin{center}
		\subfigure{
			\begin{tikzpicture}[node distance=1cm,
			slot/.style={draw,rectangle},
			vertex/.style={draw,circle},
			scale=0.7,every node/.style={scale=0.7}
			]
			\node[vertex, label=below:$v_1$] (v1) at (1,0) {};
			\node[vertex, label=below:$v_2$] (v2) at (2,0) {};
			\node[vertex, label=below:$\hdots$] (van) at (3,0) {};
			\node[vertex, label=below:$v_3$] (v3) at (4,0) {};
			\node[vertex, label=below:$v_4$] (v4) at (5,0) {};
			\node[slot, label=above:$s_j$] (si) at (1,2) {};
			\node[slot,fill=black, label=above:$s_i$] (sj) at (4,2) {};
			\node[slot, label=above:$s_k$] (sk) at (3,2) {};
			\draw[thick] (sj) edge (v1)
			(sj) edge (v2);
			\draw[thick,blue,->] (van) edge (sk);
			\draw[thick,blue,->](1.5,1) -- (sk);
			\draw[thick,blue,->](1,1) -- (si);
			\draw[thick,blue,->](0.5,1) -- (si);
			
			\end{tikzpicture}
		}\hspace*{0.35cm}
		\subfigure{
			\begin{tikzpicture}[node distance=1cm,
			slot/.style={draw,rectangle},
			vertex/.style={draw,circle},
			scale=0.7,every node/.style={scale=0.7}
			]
			\node[vertex, label=below:$v_1$] (v1) at (1,0) {};
			\node[vertex, label=below:$v_2$] (v2) at (2,0) {};
			\node[vertex, label=below:$\hdots$] (van) at (3,0) {};
			\node[vertex, label=below:$v_3$] (v3) at (4,0) {};
			\node[vertex, label=below:$v_4$] (v4) at (5,0) {};
			\node[slot, label=above:$s_j$] (si) at (1,2) {};
			\node[slot,fill=black, label=above:$s_i$] (sj) at (4,2) {};
			\node[slot, label=above:$s_k$] (sk) at (3,2) {};
			\draw[thick] (sj) edge (v1)
			(sj) edge (v2);
			\draw[thick,blue,->] (1.3,1.2) -- (sk);
			\draw[thick,blue,->](1.5,1) -- (sk);
			\draw[thick,blue,->](van) -- (5,2);
			\draw[thick,blue,->](1.6,0.8) -- (4.7,2);
			\draw[thick, red] (v3) edge (si)
			(v4) edge (si);
			
			\end{tikzpicture}
		}\hspace*{0.35cm}
		\subfigure{
			\begin{tikzpicture}[node distance=1cm,
			slot/.style={draw,rectangle},
			vertex/.style={draw,circle},
			scale=0.7,every node/.style={scale=0.7}
			]
			\node[vertex, label=below:$v_1$] (v1) at (1,0) {};
			\node[vertex, label=below:$v_2$] (v2) at (2,0) {};
			\node[vertex, label=below:$\hdots$] (van) at (3,0) {};
			\node[vertex, label=below:$v_3$] (v3) at (4,0) {};
			\node[vertex, label=below:$v_4$] (v4) at (5,0) {};
			\node[slot, label=above:$s_j$] (si) at (1,2) {};
			\node[slot,fill=black, label=above:$s_i$] (sj) at (4,2) {};
			\node[slot, label=above:$s_k$] (sk) at (3,2) {};
			\draw[thick] (sj) edge (v1)
			(sj) edge (v2);
			\draw[thick,blue,->] (van) -- (5,2);
			\draw[thick,blue,->](1.5,1) -- (4.7,2);
			\draw[thick,blue,->](1,1) -- (si);
			\draw[thick,blue,->](0.5,1) -- (si);
			\draw[thick, red] (v3) edge (sk)
			(v4) edge (sk);
			
			\end{tikzpicture}
		} 
		\caption{Lemma~\ref{lem:40gaps}: At most one propagation arrow
			crosses from the left of $v_2$ to $s_k$ by
			Lemma~\ref{lem:PropagationCrossEdge}.  }
		\label{fig:40gapscrossings2}
	\end{center}
\end{figure}
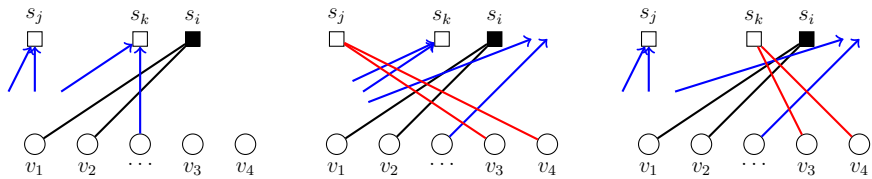
  
  Before we assign the request $\{v_3,v_4\}$, we know that by \Cref{lem:PropagationCrossEdge} only
  one propagation arrow can cross from the left of (or directly from) $v_2$ to the slot $s_k$. Thus, when assigning the request to the slot $s_k$ there are
  no extra crossings for the propagation arrows going to $s_j$. However, if we assign the request to slot $s_j$, the propagation arrows
  assigned to $s_j$ will now be transferred to $s_k$ as we saw in \Cref{obs:shiftingarrows}, creating four new crossings between these
  propagation arrows and the new edges. This results in a contradiction, as we have just seen that the placement in $s_j$ generates more
  crossings than the placement in $s_k$ which contradicts our assumption that $s_j$ has the smallest crossing values.
 \end{proof}
 
 We prove now that \Cref{alg:minallcrossings} only generates 4-0
 crossings when they are forced or in a very specific configuration.  
 We will prove this in two different lemmas.
 
 If we have a request for a pair of vertices, such that every available slot generates at least one 4-0 crossing,
 we call it a forced 4-0 crossing. 
 Observe, that it is possible that more than one 4-0 crossing is forced by the same request (See Figure~\ref{fig:forced40-stacking}).
 
   \begin{figure}[h]
 \begin{center}
   \begin{tikzpicture}[node distance=1cm,
   slot/.style={draw,rectangle},
   vertex/.style={draw,circle},
   scale=0.8,every node/.style={scale=0.8}
   ]
      \node[vertex] (va) at (1,0) {};
      \node[vertex] (vb) at (2,0) {};
      \node[vertex] (vc) at (5,0) {};
      \node[vertex] (vd) at (6,0) {};
      \node[vertex, label=below:$v_1$] (v1) at (3,0) {};
      \node[vertex, label=below:$v_2$] (v2) at (4,0) {};
      \node[slot, fill=black] (sab1) at (2,2) {};
      \node[slot, fill=black] (sab2) at (1,2) {};
      \node[slot, fill=black] (scd1) at (6,2) {};
      \node[slot, fill=black] (scd2) at (5,2) {};
      \node at (3.5,2) {$\hdots$};
      \node[slot,label=above:$s_{\ell}$] (s1) at (0,2) {};
      \node[slot,label=above:$s_{r}$] (s2) at (7,2) {};
      \draw[thick] (sab1) edge (va)
                   (sab1) edge (vb)
                   (scd1) edge (vd)
		   (scd1) edge (vc)
		   (sab2) edge (va)
                   (sab2) edge (vb)
                   (scd2) edge (vd)
		   (scd2) edge (vc);
      \draw[thick,blue,->] (v1)--(s1);
      \draw[thick,blue,->] (v2)--(s2);
  \end{tikzpicture}
   \caption{More than one 4-0 crossing might be forced by the same request}
   \label{fig:forced40-stacking}
  \end{center}
 \end{figure}
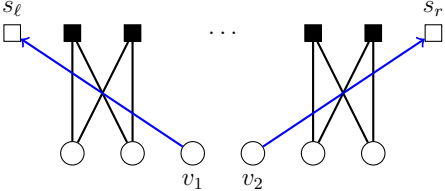

  \begin{lemma}\label{lem:forced40-stacking}
  If Algorithm~\ref{alg:minallcrossings} is used,
  for every forced 4-0 crossing there is at least
  one uniquely identifiable and unavoidable crossing.
 \end{lemma}
 
 \begin{proof}
 We will prove this using \Cref{lem:UnavoidableHouses}.
 If a request $\{v_1, v_2\}$ arrives in time step $t$ and every possible placement generates a 4-0 crossing, the propagation arrows of $v_1$ and $v_2$ have to point to two different slots before the request is served due to \Cref{lem:PropagationCrossEdge}. 
 We assume $v_1 < v_2$ and call the slot on the left hand side $s_\ell$ and the other one $s_r$,
 as sketched in Figure~\ref{fig:forced40-stacking}. 
 
 We denote the set of edges that are crossed by the request $\{v_1, v_2\}$ when it is placed in $s_\ell$ with $L$ and analogously we define the set of crossed edges $R$ for the slot $s_r$. 
 To be precise, the set $L$ contains the edges $(s_j,v_i)$ with $s_\ell < s_j < s_r$ and $v_i < v_1$ 
 (See Figure~\ref{fig:forced40-LR}). 
 Our algorithm will always choose the slot which results in the least amount of crossings. 
 Therefore, if our algorithm chooses (w.l.o.g.) the slot $s_\ell$, we know that positioning the request in slot $s_r$ results in at least the same number of crossings. 
 Thus, we know that $|L| \leq |R|$ (or $|L|\leq |R|+1$ if the edge other connected to $v_2$ is placed between
 $s_{\ell}$ and $s_r$ but the edge connected to $v_1$ is not). 
 
 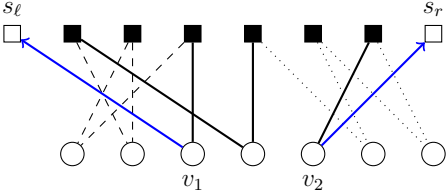
\begin{figure}[h]
 	\begin{center}
 		\begin{tikzpicture}[node distance=1cm,
 		slot/.style={draw,rectangle},
 		vertex/.style={draw,circle},
 		scale=0.8,every node/.style={scale=0.8}
 		]
 		\node[vertex] (va) at (1,0) {};
 		\node[vertex] (vb) at (2,0) {};
 		\node[vertex] (vc) at (6,0) {};
 		\node[vertex] (vd) at (7,0) {};
 		\node[vertex] (vmid) at (4,0) {};
 		\node[vertex, label=below:$v_1$] (v1) at (3,0) {};
 		\node[vertex, label=below:$v_2$] (v2) at (5,0) {};
 		\node[slot, fill=black] (sab1) at (3,2) {};
 		\node[slot, fill=black] (sab2) at (2,2) {};
 		\node[slot, fill=black] (scd1) at (6,2) {};
 		\node[slot, fill=black] (scd2) at (5,2) {};
 		\node[slot, fill=black] (smid) at (4,2) {};
 		\node[slot, fill=black] (sl1) at (1,2) {};
 		\node[slot, label=above:$s_{\ell}$] (s1) at (0,2) {};
 		\node[slot, label=above:$s_r$] (s2) at (7,2) {};
 		\draw[dashed] (sab1) edge (va)                    
 		(sab2) edge (va)
 		(sab2) edge (vb)
 		(sl1) edge (vb);
 		\draw[dotted] (smid) edge (vc)
 		(scd2) edge (vd)
 		(scd2) edge (vc)
 		(scd1) edge (vd);
 		\draw[thick] (smid) edge (vmid)
 		(sab1) edge (v1)
 		(scd1) edge (v2)
 		(sl1) edge (vmid);
 		
 		\draw[thick,blue,->] (v1)--(s1);
 		\draw[thick,blue,->] (v2)--(s2);
 		\end{tikzpicture}
 		\caption{The set of edges $L$ is dashed, and the set of edges $R$ is dotted.
 			The request $\{v_1,v_2\}$  will generate two crossings per dashed
 			edge if positioned in $s_{\ell}$ and two crossings per dotted edge respectively
 			if positioned in $s_r$. There is one unavoidable crossing for the edges going to 
 			vertices between $v_1$ and $v_2$ no matter the positioning.}
 		\label{fig:forced40-LR}
 	\end{center}
 \end{figure} 
 
 Because we are in the situation of a forced 4-0 crossing there are at least two edges in each set $L$ and $R$ that belong to the same request. 
 We look at the pair of edges $(s_i,v_{i_1})$ and $(s_i,v_{i_2})$ in $L$, with $v_{i_1}<v_{i_2}$ 
 with the smallest possible $v_{i_2}$, respectively
 the pair of edges $(s_j,v_{j_1})$ and $(s_j,v_{j_2})$ in $R$, with $v_{j_1}<v_{j_2}$ and 
 the largest possible $v_{j_1}$. 
 Applying now \Cref{lem:UnavoidableHouses}, we know that there will be a future request between at least
 $v_{i_2}$ and $v_{j_1}$, meaning that a future housing request will cross at least one of the edges
 - in $L$ and $R$ respectively - of every pair generating a 4-0 crossing except at most one.
 Observe also, that by \Cref{lem:PropagationCrossEdge} there cannot be any available edge slot between
 $v_{i_2}$ and $v_{j_1}$ other than $v_1$ and $v_2$, this will mather further in the proof.
 
 Note that every request, from which only one edge is in $L$ (or $R$), unavoidably crosses the request $\{v_1, v_2\}$ anyway. 
 Thus, at least $\frac{|L|}{2}+\frac{|R|}{2}(+1)$ edges are unavoidably crossed after the request
 in \Cref{lem:UnavoidableHouses}. Here, the unavoidable crossing between the 
 ``housing request'' and $v_1$ and $v_2$ compensate for the potentially missing crossings with the
 edges from $v_{i_2}$ and $v_{j_1}$. 
 Or in other words, for every slot between $s_\ell$ and $s_r$ there is at least one edge which is crossed unavoidably, except for the aforementioned exceptions. 
 
 The request $\{v_1, v_2\}$ crosses all edges in $L$ twice, if it is placed in $s_\ell$. 
 Thus, the number of avoidable crossings is $2|L|$ which is at most twice as large as the number of unavoidable crossings $\frac{|L|}{2}+\frac{|R|}{2} (+1)\geq |L|$. 
 
 The argument above works if there is only one forced 4-0 crossing for a set of requests before the 
 housing request from \Cref{lem:UnavoidableHouses} appears.
 In the following, we discuss why we can assign, for each set of potentially overlapping forced 4-0
 crossings, a uniquely identifiable unavoidable crossing. 
 Overlapping 4-0 crossings appear, 
 when both involved requests in a 4-0 crossing are again completely crossed by another third request 
 (See Figure~\ref{fig:forced40-overlapping}).
 
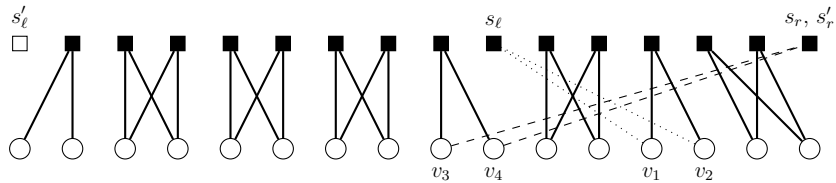
\begin{figure}[b]
 \begin{center}
   \begin{tikzpicture}[node distance=1cm,
   slot/.style={draw,rectangle},
   vertex/.style={draw,circle},
   scale=0.7,every node/.style={scale=0.7}
   ]
      \node[slot, label=above:$s'_{\ell}$] (s0) at (0,0) {};
      \node[slot, fill=black] (s1) at (1,0) {};
      \node[slot, fill=black] (s2) at (2,0) {};
      \node[slot, fill=black] (s3) at (3,0) {};
      \node[slot, fill=black] (s4) at (4,0) {};
      \node[slot, fill=black] (s5) at (5,0) {};
      \node[slot, fill=black] (s6) at (6,0) {};
      \node[slot, fill=black] (s7) at (7,0) {};
      \node[slot, fill=black] (s8) at (8,0) {};
      \node[slot, fill=black, label=above:$s_{\ell}$] (s9) at (9,0) {};
      \node[slot, fill=black] (s10) at (10,0) {};
      \node[slot, fill=black] (s11) at (11,0) {};
      \node[slot, fill=black] (s12) at (12,0) {};
      \node[slot, fill=black] (s13) at (13,0) {};
      \node[slot, fill=black] (s14) at (14,0) {};
      \node[slot, fill=black, label=above:$s_r\text{, } s'_r$] (s15) at (15,0) {};
      
      \node[vertex] (v0) at (0,-2) {};
      \node[vertex] (v1) at (1,-2) {};
      \node[vertex] (v2) at (2,-2) {};
      \node[vertex] (v3) at (3,-2) {};
      \node[vertex] (v4) at (4,-2) {};
      \node[vertex] (v5) at (5,-2) {};
      \node[vertex] (v6) at (6,-2) {};
      \node[vertex] (v7) at (7,-2) {};
      \node[vertex, label=below:$v_3$] (v8) at (8,-2) {};
      \node[vertex, label=below:$v_4$] (v9) at (9,-2) {};
      \node[vertex] (v10) at (10,-2) {};
      \node[vertex] (v11) at (11,-2) {};
      \node[vertex, label=below:$v_1$] (v12) at (12,-2) {};
      \node[vertex, label=below:$v_2$] (v13) at (13,-2) {};
      \node[vertex] (v14) at (14,-2) {};
      \node[vertex] (v15) at (15,-2) {};
      
      \draw[thick] (v0) edge (s1)
                   (v1) edge (s1)
                   (v2) edge (s2)
                   (v3) edge (s2)
                   (v2) edge (s3)
                   (v3) edge (s3)
                   (v4) edge (s4)
                   (v5) edge (s4)
                   (v4) edge (s5)
                   (v5) edge (s5)
                   (v6) edge (s6)
                   (v7) edge (s6)
                   (v6) edge (s7)
                   (v7) edge (s7)
                   (v8) edge (s8)
                   (v9) edge (s8)
                   (v10) edge (s10)
                   (v11) edge (s10)
                   (v10) edge (s11)
                   (v11) edge (s11)
                   (v12) edge (s12)
                   (v13) edge (s12)
                   (v14) edge (s14)
                   (v15) edge (s14)
                   (v14) edge (s13)
                   (v15) edge (s13);
              
      \draw[dashed] (v8) edge (s15)                    
                    (v9) edge (s15);
      \draw[dotted] (v12) edge (s9)
                    (v13) edge (s9);
  \end{tikzpicture}
   \caption{Two sets of 4-0 crossings overlap each other before a housing request appears.
   The first one is drawn with a dotted pair of edges and the second one is dashed.
   This particular scenario is not completely realistic for Algorithm~\ref{alg:minallcrossings} but could happen if the  overall
   graph is larger.}
   \label{fig:forced40-overlapping}
  \end{center}
 \end{figure}  
 
   \begin{figure}[h]
 	\begin{center}
 		\subfigure{
 			\begin{tikzpicture}[node distance=1cm,
 			slot/.style={draw,rectangle},
 			vertex/.style={draw,circle},
 			scale=0.7,every node/.style={scale=0.7}
 			]
 			\node[vertex] (va) at (1,0) {};
 			\node[vertex] (vb) at (2,0) {};
 			\node[vertex] (vc) at (5,0) {};
 			\node[vertex] (vd) at (6,0) {};
 			\node[vertex, label=below:$v_1$] (v1) at (3,0) {};
 			\node[vertex, label=below:$v_2$] (v2) at (4,0) {};
 			\node[slot, fill=black] (sab) at (2,2) {};
 			\node[slot, fill=black] (scd) at (5,2) {};
 			\node at (3.5,2) {$\hdots$};
 			\node[slot] (s1) at (1,2) {};
 			\node[slot] (s2) at (6,2) {};
 			\draw[thick] (sab) edge (va)
 			(sab) edge (vb)
 			(scd) edge (vd)
 			(scd) edge (vc);
 			\draw[thick,blue,->] (v1)--(s1);
 			\draw[thick,blue,->] (v2)--(s2);
 			\end{tikzpicture}}\hspace*{1cm}
 		\subfigure{
 			\begin{tikzpicture}[node distance=1cm,
 			slot/.style={draw,rectangle},
 			vertex/.style={draw,circle},
 			scale=0.7,every node/.style={scale=0.7}
 			]
 			\node[vertex] (va) at (1,0) {};
 			\node[vertex] (vb) at (2,0) {};
 			\node[vertex] (vc) at (5,0) {};
 			\node[vertex] (vd) at (6,0) {};
 			\node[vertex, label=below:$v_1$] (v1) at (3,0) {};
 			\node[vertex, label=below:$v_2$] (v2) at (4,0) {};
 			\node[vertex, label=below:$u_1$] (u1) at (0,0) {};
 			\node[vertex, label=below:$u_2$] (u2) at (7,0) {};
 			\node[slot, fill=black] (sab) at (2,2) {};
 			\node[slot, fill=black] (scd) at (5,2) {};
 			\node at (3.5,2) {$\hdots$};
 			\node[slot, fill=black] (s1) at (1,2) {};
 			\node[slot] (s2) at (6,2) {};
 			\draw[thick] (sab) edge (va)
 			(sab) edge (vb)
 			(scd) edge (vd)
 			(scd) edge (vc);
 			\draw[thick] (v1) -- (s1);
 			\draw[thick] (v2) -- (s1);
 			\draw[thick,blue,->] (u1)--(s2);
 			\draw[thick,blue,->] (u2)--(s2);
 			\end{tikzpicture}}
 		\caption{If there is no available slot without a 4-0 crossings, the propagation arrows
 			point to different sides, and a request $u_1,u_2$ must eventually exist.}
 		\label{fig:forced40}
 	\end{center}
 \end{figure}
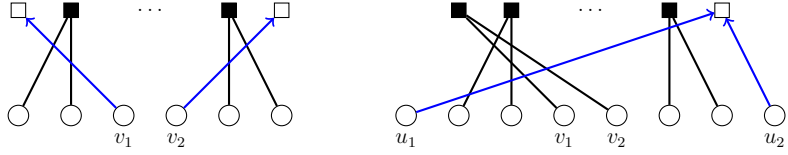  
 
 We call the request that generates the first forced 4-0 crossing after its placement $\{v_1, v_2\}$. 
 The algorithm had the decision to place it in the left slot $s_\ell$, crossing $|L|$ edges or in the right slot $s_r$, crossing $|R|$ edges. 
 
 Without loss of generality assume that the request $\{v_1, v_2\}$ was placed in
 $s_{\ell}$. In order to have an overlapping 4-0 request, we assume that the request of the second forced 4-0 crossing $\{v_3, v_4\}$ is to the left of the request $\{v_1, v_2\}$. 
 Because this request is also a forced 4-0 crossing, it can be placed in a slot to the left 
 $s_\ell^\prime < s_\ell$ or a slot to the right $s'_r \ge s_r$. 
 The slots must be more to the left (respectively, more to the right) because when the request $\{v_1, v_2\}$ arrives, all of the other vertices between $v_{i_2}$ and $v_{j_1}$ must be filled, as we already
 argued, thus $v_3$ and $v_4$ are to the left of
 $v_{i_2}$. 
 Together with the assumption that the second 4-0 crossing is forced, we get the restricted position for $v_3$ and $v_4$.
 
 Analogous to the previous case, let  $L^\prime$ be the set of edges $(v,s)$ with $v<v_3$
 and $s'_{\ell}<s<s'_r$ and let $R'$ be the set of edges $(v,s)$
 with  $v>v_4$ and $s'_{\ell}<s<s'_r$
 Note that, the edges of the first 4-0 request are now part of $R^\prime$. 
 Like already explained above, if the algorithm decides to place the request in $s_r^\prime$, 
 this implies that $|R^\prime| \leq |L^\prime|(+1)$ holds. 
 Moreover, $|R^\prime| \geq \frac{|L|}{2} + |L|$ holds,
 because $v_4<v_{i_2}$, which by definition means that half of the edges
 of $R$ are to the right of $v_4$ and, thus, part of $R'$. 
 Thus, applying again \Cref{lem:UnavoidableHouses}, 
 we know that there will be a future request that unavoidably crosses at least 
 \[\frac{|L^\prime|}{2}(+1)+\frac{|R^\prime|}{2} \geq \frac{|R^\prime|}{2}+\frac{|R^\prime|}{2} 
\geq \frac{|R^\prime|}{2}+ \frac{|R|}{4} + \frac{|L|}{2} \]
 edges. 
 The number of avoidable crossings is $2|L| + 2|R^\prime|\le |L|+|R|+2|R'|$,
 which if divided by $4$, for each possible $4-0$ crossing, means that
 \[\frac{|L|}{4}+\frac{|R|}{4}+\frac{|R'|}{2}\le\frac{|R^\prime|}{2}+ \frac{|R|}{4} + \frac{|L|}{2} \;.\] 
 Thus, we have for each of the avoidable 4-0 crossings at least one uniquely identifiable 
 unavoidable crossing. Observe that we can iterate this argument for every possible overlapping 4-0
 crossing. Moreover, if in the second case, the request $\{v_3,v_4\}$ was placed in $s'_{\ell}$
 the analogous counting argument still holds.
 \end{proof}
 
  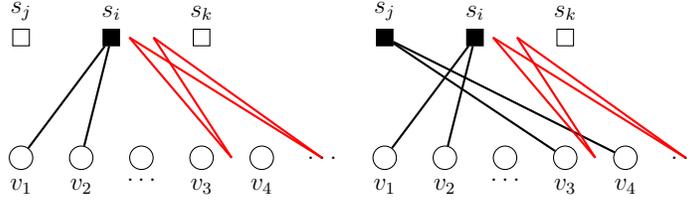
\begin{figure}
 \begin{center}
 \subfigure{
   \begin{tikzpicture}[node distance=1cm,
   slot/.style={draw,rectangle},
   vertex/.style={draw,circle},
   scale=0.8,every node/.style={scale=0.8}
   ]
      \node[vertex, label=below:$v_1$] (v1) at (1,0) {};
      \node[vertex, label=below:$v_2$] (v2) at (2,0) {};
      \node[vertex, label=below:$\hdots$] (van) at (3,0) {};
      \node[vertex, label=below:$v_3$] (v3) at (4,0) {};
      \node[vertex, label=below:$v_4$] (v4) at (5,0) {};
      \node[slot, label=above:$s_j$] (si) at (1,2) {};
      \node[slot,fill=black, label=above:$s_i$] (sj) at (2.5,2) {};
      \node[slot, label=above:$s_k$] (sk) at (4,2) {};
      \node at (6,0) {$\hdots$};
      \draw[thick] (sj) edge (v1)
		   (sj) edge (v2);
      \draw[red,thick] (6,0) -- (2.8,2);
      \draw[red,thick] (4.5,0) -- (2.8,2);
      \draw[red,thick] (6,0) -- (3.2,2);
      \draw[red,thick] (4.5,0) -- (3.2,2);
  \end{tikzpicture}}
   \subfigure{
   \begin{tikzpicture}[node distance=1cm,
   slot/.style={draw,rectangle},
   vertex/.style={draw,circle},
   scale=0.8,every node/.style={scale=0.8}
   ]
      \node[vertex, label=below:$v_1$] (v1) at (1,0) {};
      \node[vertex, label=below:$v_2$] (v2) at (2,0) {};
      \node[vertex, label=below:$\hdots$] (van) at (3,0) {};
      \node[vertex, label=below:$v_3$] (v3) at (4,0) {};
      \node[vertex, label=below:$v_4$] (v4) at (5,0) {};
      \node[slot, fill=black, label=above:$s_j$] (si) at (1,2) {};
      \node[slot,fill=black, label=above:$s_i$] (sj) at (2.5,2) {};
      \node[slot, label=above:$s_k$] (sk) at (4,2) {};
      \node at (6,0) {$\hdots$};
      \draw[thick] (sj) edge (v1)
		   (sj) edge (v2);
		   \draw[thick] (si) edge (v3)
		   (si) edge (v4);
      \draw[red,thick] (6,0) -- (2.8,2);
      \draw[red,thick] (4.5,0) -- (2.8,2);
      \draw[red,thick] (6,0) -- (3.2,2);
      \draw[red,thick] (4.5,0) -- (3.2,2);
  \end{tikzpicture}}
   \caption{If there is more than one slot positioned like the red ones (between  the slots $s_i$ and $s_k$
   with one vertex between $v_3$ and $v_4$, and one to the right of $v_4$ each), Algorithm \ref{alg:minallcrossings} may choose 
   slot $s_j$ generating a 4-0 crossing.}
   \label{fig:40exception}
  \end{center}
 \end{figure}  
 
 We just proved that forced 4-0 when using Algorithm~\ref{alg:minallcrossings},
 incur in one additional unavoidable crossing, this means that we can consider 4-0 crossings
 as if they were, in a sense 5-1 crossings instead, with a competitive ratio of 5 instead of
 being unbounded.
 However, this is not enough, there can be 4-0 crossings produced by Algorithm~\ref{alg:minallcrossings}
 that are not forced. In the following lemma we prove that non-forced 4-0 crossings are only produced
 by Algorithm~\ref{alg:minallcrossings} in a very specific configuration. Then we will proceed to look at 
 the number of uniquely identifiable unavoidable crossings of that configuration.
 
 \begin{lemma}\label{lem:preferno40}
  Given a request for a pair of vertices in a graph, whose 4-0
  crossings have either been forced (Figure~\ref{fig:forced40})
  or were served because any alternative placement would result in two 3-1
  crossings as sketched in Figure~\ref{fig:40exception}.
  If a slot is available which will not generate any 4-0 crossings 
  this slot will be selected by Algorithm~\ref{alg:minallcrossings}
  over any slot which will generate a 4-0 crossing, 
  unless there are two additional requests 
  resulting in two 3-1 crossings for the alternative placement 
  as depicted in Figure~\ref{fig:40exception}.
 \end{lemma}
 
 \begin{proof}
  Let us assume that we have a graph with the only 4-0 crossings appearing in the configurations of 
  \Cref{fig:forced40,fig:40exception}.
  Let $\{v_1,v_2\}$ be a request assigned to slot $s_i$.
  Let $\{v_3,v_4\}$ be a new request with $v_2<v_3$ without loss of generality.
  The new request can be assigned to a slot
  $s_k$ right of $s_i$ without generating new 4-0 crossings or to a slot $s_j$ to the left of $s_i$ as
  depicted in the first drawing of Figure~\ref{fig:preferno40}. We can assume by 
  Lemma~\ref{lem:40gaps} that $s_j$ is the rightmost available slot after $s_i$.
  
  As we did in Lemma~\ref{lem:40gaps}, we first count edge crossings and then count
  the propagation arrow crossings.
  
  In order to do this, we divide the relevant slots into two subsets.
  The subset $X$ contains the slots between $s_j$ and $s_i$. Recall that all the slots
  in this area are filled.
  The subset $Y$ contains the slots between $s_i$ and $s_k$, all of them are filled too.
  We also divide the vertices into three subsets.
  Any vertex to the left of $v_3$ belongs to subset $A$. Vertices between $v_3$ and $v_4$
  belong to subset $B$ and vertices to the right of $v_4$ belong to subset $C$.
  This division is depicted in Figure~\ref{fig:preferno40}.
  
  Only edges to slots in $X$ or $Y$ will generate crossings that count only for 
  one of the two placements. In particular any edge from $C$ to $X$ or $Y$ will generate two additional crossings
  for the placement in $s_k$ with respect to the placement in $s_j$, those 
  edges are depicted in red in the second drawing of Figure~\ref{fig:preferno40}.
  On the other hand any edge from $A$ to a slot in $X$ or $Y$ generates two additional crossings
  for for the placement in $s_j$ with respect to the placement in $s_k$. Those edges are depicted
  in green in the second drawing of Figure~\ref{fig:preferno40}. Finally, the edges from
  $B$ to $X$ or $Y$ are neutral with respect to both placements. This means that we only need to analyze
  previously placed requests in $X$ or $Y$ with one endpoint in $C$, as these are the only ones that will
  make a placement in $s_j$ more likely with respect to a placement in $s_j$.
  
  We now analyze all possible requests in $X$ with at least one endpoint in $C$.
  Recall that we assume that there is always a slot that does not force a 4-0
  crossing. This means that there cannot be a pair of edges from $C$ connected to the same slot in $X$ or $Y$.
  On the other hand if a pair of edges from $C$ and $B$ respectively go to a slot $X$
  (we call this request $CXB$), we have a previous 4-0 crossing
  in the graph. This means that
  either the request placed in $s_i$ generated a 4-0 crossing, or the request $CXB$ did,
  we distinguish these two cases.
  
  If the placement in $s_i$ generated the 4-0 crossing we argue that 
  there was a situation like in Figure~\ref{fig:preferno40}.
  If the 4-0 crossing was forced when $s_i$ was placed, this means that there was a request
  to the right of $v_1$, and it was fulfilled by a slot in $X$, but between this request and 
  $CXB$ there were at least 3 propagation arrows, in particular, from $v_1$, $v_2$ and $v_3$,
  so this situation is forbidden by Lemma~\ref{lem:PropagationCrossEdge}.
  If there was a situation Figure~\ref{fig:preferno40} involving the request $CXB$ and 
  $v_1,v_2$, then there must be two requests between $v_1$ and $v_2$ are in region $A$ and are fulfilled
  in the region of $X$. These two requests will completely counteract the crossing contributions of 
  the request $CXB$.
  
  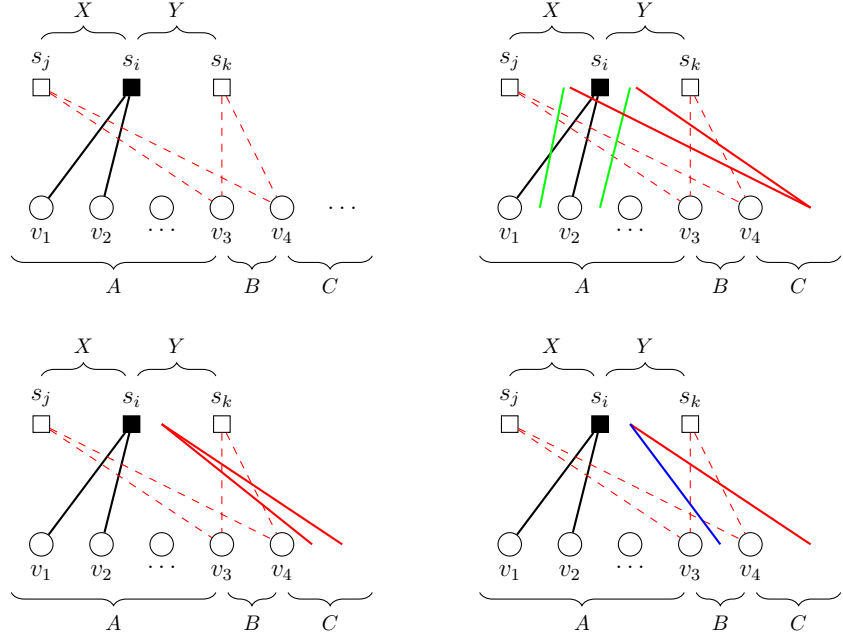
\begin{figure}[t]
  	\begin{center}
  		\subfigure{
  			\begin{tikzpicture}[node distance=1cm,
  			slot/.style={draw,rectangle},
  			vertex/.style={draw,circle},
  			scale=0.8,every node/.style={scale=0.8}
  			]
  			\node[vertex, label=below:$v_1$] (v1) at (1,0) {};
  			\node[vertex, label=below:$v_2$] (v2) at (2,0) {};
  			\node[vertex, label=below:$\hdots$] (van) at (3,0) {};
  			\node[vertex, label=below:$v_3$] (v3) at (4,0) {};
  			\node[vertex, label=below:$v_4$] (v4) at (5,0) {};
  			\node[slot, label=above:$s_j$] (si) at (1,2) {};
  			\node[slot,fill=black, label=above:$s_i$] (sj) at (2.5,2) {};
  			\node[slot, label=above:$s_k$] (sk) at (4,2) {};
  			\node at (6,0) {$\hdots$};
  			\draw[thick] (sj) edge (v1)
  			(sj) edge (v2);
  			\draw[dashed, red] (si) edge (v3)
  			(si) edge (v4)
  			(sk) edge (v3)
  			(sk) edge (v4);
  			\draw[decorate,decoration={brace,amplitude=5pt,mirror}] (0.5,-0.8) -- (3.9,-0.8) node [black,midway,yshift=-0.5cm] {\footnotesize $A$};
  			\draw[decorate,decoration={brace,amplitude=5pt,mirror}] (4.1,-0.8) -- (4.9,-0.8) node [black,midway,yshift=-0.5cm] {\footnotesize $B$};
  			\draw[decorate,decoration={brace,amplitude=5pt,mirror}] (5.1,-0.8) -- (6.5,-0.8) node [black,midway,yshift=-0.5cm] {\footnotesize $C$};
  			\draw[decorate,decoration={brace,amplitude=5pt}] (1,2.8) -- (2.4,2.8) node [black,midway,yshift=0.5cm] {\footnotesize $X$};
  			\draw[decorate,decoration={brace,amplitude=5pt}] (2.6,2.8) -- (3.9,2.8) node [black,midway,yshift=0.5cm] {\footnotesize $Y$};
  			\end{tikzpicture}
  		}\hspace*{1cm}
  		\subfigure{
  			\begin{tikzpicture}[node distance=1cm,
  			slot/.style={draw,rectangle},
  			vertex/.style={draw,circle},
  			scale=0.8,every node/.style={scale=0.8}
  			]
  			\node[vertex, label=below:$v_1$] (v1) at (1,0) {};
  			\node[vertex, label=below:$v_2$] (v2) at (2,0) {};
  			\node[vertex, label=below:$\hdots$] (van) at (3,0) {};
  			\node[vertex, label=below:$v_3$] (v3) at (4,0) {};
  			\node[vertex, label=below:$v_4$] (v4) at (5,0) {};
  			\node[slot, label=above:$s_j$] (si) at (1,2) {};
  			\node[slot,fill=black, label=above:$s_i$] (sj) at (2.5,2) {};
  			\node[slot, label=above:$s_k$] (sk) at (4,2) {};
  			\draw[thick] (sj) edge (v1)
  			(sj) edge (v2);
  			\draw[dashed, red] (si) edge (v3)
  			(si) edge (v4)
  			(sk) edge (v3)
  			(sk) edge (v4);
  			\draw[red,thick] (6,0) -- (3.1,2);
  			\draw[green,thick] (2.5,0) -- (3,2);
  			\draw[red,thick] (6,0) -- (2,2);
  			\draw[green,thick] (1.5,0) -- (1.9,2);
  			\draw[decorate,decoration={brace,amplitude=5pt,mirror}] (0.5,-0.8) -- (3.9,-0.8) node [black,midway,yshift=-0.5cm] {\footnotesize $A$};
  			\draw[decorate,decoration={brace,amplitude=5pt,mirror}] (4.1,-0.8) -- (4.9,-0.8) node [black,midway,yshift=-0.5cm] {\footnotesize $B$};
  			\draw[decorate,decoration={brace,amplitude=5pt,mirror}] (5.1,-0.8) -- (6.5,-0.8) node [black,midway,yshift=-0.5cm] {\footnotesize $C$};
  			\draw[decorate,decoration={brace,amplitude=5pt}] (1,2.8) -- (2.4,2.8) node [black,midway,yshift=0.5cm] {\footnotesize $X$};
  			\draw[decorate,decoration={brace,amplitude=5pt}] (2.6,2.8) -- (3.9,2.8) node [black,midway,yshift=0.5cm] {\footnotesize $Y$};
  			\end{tikzpicture}
  		}\\
  		\subfigure{
  			\begin{tikzpicture}[node distance=1cm,
  			slot/.style={draw,rectangle},
  			vertex/.style={draw,circle},
  			scale=0.8,every node/.style={scale=0.8}
  			]
  			\node[vertex, label=below:$v_1$] (v1) at (1,0) {};
  			\node[vertex, label=below:$v_2$] (v2) at (2,0) {};
  			\node[vertex, label=below:$\hdots$] (van) at (3,0) {};
  			\node[vertex, label=below:$v_3$] (v3) at (4,0) {};
  			\node[vertex, label=below:$v_4$] (v4) at (5,0) {};
  			\node[slot, label=above:$s_j$] (si) at (1,2) {};
  			\node[slot,fill=black, label=above:$s_i$] (sj) at (2.5,2) {};
  			\node[slot, label=above:$s_k$] (sk) at (4,2) {};
  			\draw[thick] (sj) edge (v1)
  			(sj) edge (v2);
  			\draw[dashed, red] (si) edge (v3)
  			(si) edge (v4)
  			(sk) edge (v3)
  			(sk) edge (v4);
  			\draw[red,thick] (6,0) -- (3,2);
  			\draw[red,thick] (5.5,0) -- (3,2);
  			\draw[decorate,decoration={brace,amplitude=5pt,mirror}] (0.5,-0.8) -- (3.9,-0.8) node [black,midway,yshift=-0.5cm] {\footnotesize $A$};
  			\draw[decorate,decoration={brace,amplitude=5pt,mirror}] (4.1,-0.8) -- (4.9,-0.8) node [black,midway,yshift=-0.5cm] {\footnotesize $B$};
  			\draw[decorate,decoration={brace,amplitude=5pt,mirror}] (5.1,-0.8) -- (6.5,-0.8) node [black,midway,yshift=-0.5cm] {\footnotesize $C$};
  			\draw[decorate,decoration={brace,amplitude=5pt}] (1,2.8) -- (2.4,2.8) node [black,midway,yshift=0.5cm] {\footnotesize $X$};
  			\draw[decorate,decoration={brace,amplitude=5pt}] (2.6,2.8) -- (3.9,2.8) node [black,midway,yshift=0.5cm] {\footnotesize $Y$};
  			\end{tikzpicture}
  		}\hspace*{1cm}
  		\subfigure{
  			\begin{tikzpicture}[node distance=1cm,
  			slot/.style={draw,rectangle},
  			vertex/.style={draw,circle},
  			scale=0.8,every node/.style={scale=0.8}
  			]
  			\node[vertex, label=below:$v_1$] (v1) at (1,0) {};
  			\node[vertex, label=below:$v_2$] (v2) at (2,0) {};
  			\node[vertex, label=below:$\hdots$] (van) at (3,0) {};
  			\node[vertex, label=below:$v_3$] (v3) at (4,0) {};
  			\node[vertex, label=below:$v_4$] (v4) at (5,0) {};
  			\node[slot, label=above:$s_j$] (si) at (1,2) {};
  			\node[slot,fill=black, label=above:$s_i$] (sj) at (2.5,2) {};
  			\node[slot, label=above:$s_k$] (sk) at (4,2) {};
  			\draw[thick] (sj) edge (v1)
  			(sj) edge (v2);
  			\draw[dashed, red] (si) edge (v3)
  			(si) edge (v4)
  			(sk) edge (v3)
  			(sk) edge (v4);
  			\draw[red,thick] (6,0) -- (3,2);
  			\draw[blue,thick] (4.5,0) -- (3,2);
  			\draw[decorate,decoration={brace,amplitude=5pt,mirror}] (0.5,-0.8) -- (3.9,-0.8) node [black,midway,yshift=-0.5cm] {\footnotesize $A$};
  			\draw[decorate,decoration={brace,amplitude=5pt,mirror}] (4.1,-0.8) -- (4.9,-0.8) node [black,midway,yshift=-0.5cm] {\footnotesize $B$};
  			\draw[decorate,decoration={brace,amplitude=5pt,mirror}] (5.1,-0.8) -- (6.5,-0.8) node [black,midway,yshift=-0.5cm] {\footnotesize $C$};
  			\draw[decorate,decoration={brace,amplitude=5pt}] (1,2.8) -- (2.4,2.8) node [black,midway,yshift=0.5cm] {\footnotesize $X$};
  			\draw[decorate,decoration={brace,amplitude=5pt}] (2.6,2.8) -- (3.9,2.8) node [black,midway,yshift=0.5cm] {\footnotesize $Y$};
  			\end{tikzpicture}
  		}
  		\caption{We have two possible placements for the request $v_3$, $v_4$ red edges contribute extra crossings to the placement in $s_k$
  			and green edges contribute extra crossings to the placement in $s_j$.  For the slots in $Y$ if a slot has both endpoints in $C$ it is a forced 4-0 crossing, but it can happen that one endpoint
  			is in $B$.
  		}
  		\label{fig:preferno40}
  	\end{center}
  \end{figure}

  If the request $CXB$ generated the 4-0 crossing,
  it also could not have been forced, as the propagation arrow from $v_4$ is between the two propagation
  arrows of the request and it also generates a situation forbidden by Lemma~\ref{lem:PropagationCrossEdge}.
  This means that also in this case there must have been a situation like in \Cref{fig:40exception}
  involving the request $CXB$. In this case, depicted in Figure~\ref{fig:double40}, 
  the two requests contributing to the situation in Figure~\ref{fig:40exception} will already be present.
  
  Finally, if a request has one endpoint in $C$ and one in $A$, this means that their crossings for the placements in 
  $s_j$ and $s_k$ compensate.


 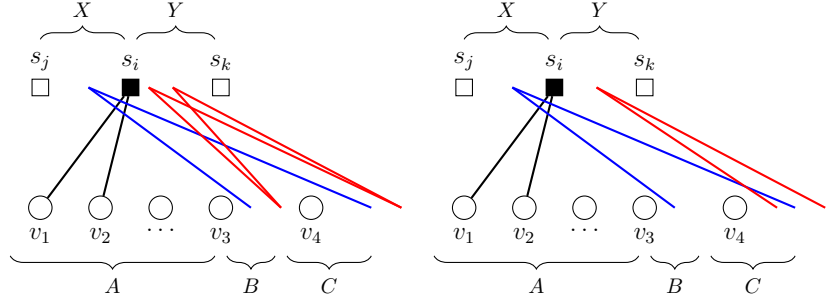
\begin{figure}[t]
 \begin{center}
  \subfigure{
  \begin{tikzpicture}[node distance=1cm,
      slot/.style={draw,rectangle},
      vertex/.style={draw,circle},
      scale=0.8,every node/.style={scale=0.8}]
      \node[vertex, label=below:$v_1$] (v1) at (1,0) {};
      \node[vertex, label=below:$v_2$] (v2) at (2,0) {};
      \node[vertex, label=below:$\hdots$] (van) at (3,0) {};
      \node[vertex, label=below:$v_3$] (v3) at (4,0) {};
      \node[vertex, label=below:$v_4$] (v4) at (5.5,0) {};
      \node[slot, label=above:$s_j$] (si) at (1,2) {};
      \node[slot,fill=black, label=above:$s_i$] (sj) at (2.5,2) {};
      \node[slot, label=above:$s_k$] (sk) at (4,2) {};
      \draw[thick] (sj) edge (v1)
		   (sj) edge (v2);
      \draw[blue,thick] (6.5,0) -- (1.8,2);
      \draw[blue,thick] (4.5,0) -- (1.8,2);
      \draw[red,thick] (7,0) -- (2.8,2);
      \draw[red,thick] (5,0) -- (2.8,2);
      \draw[red,thick] (7,0) -- (3.2,2);
      \draw[red,thick] (5,0) -- (3.2,2);
      \draw[decorate,decoration={brace,amplitude=5pt,mirror}] (0.5,-0.8) -- (3.9,-0.8) node [black,midway,yshift=-0.5cm] {\footnotesize $A$};
      \draw[decorate,decoration={brace,amplitude=5pt,mirror}] (4.1,-0.8) -- (4.9,-0.8) node [black,midway,yshift=-0.5cm] {\footnotesize $B$};
      \draw[decorate,decoration={brace,amplitude=5pt,mirror}] (5.1,-0.8) -- (6.5,-0.8) node [black,midway,yshift=-0.5cm] {\footnotesize $C$};
      \draw[decorate,decoration={brace,amplitude=5pt}] (1,2.8) -- (2.4,2.8) node [black,midway,yshift=0.5cm] {\footnotesize $X$};
      \draw[decorate,decoration={brace,amplitude=5pt}] (2.6,2.8) -- (3.9,2.8) node [black,midway,yshift=0.5cm] {\footnotesize $Y$};
  \end{tikzpicture}
  }
  \subfigure{
  \begin{tikzpicture}[node distance=1cm,
      slot/.style={draw,rectangle},
      vertex/.style={draw,circle},scale=0.8,every node/.style={scale=0.8}
      ]
      \node[vertex, label=below:$v_1$] (v1) at (1,0) {};
      \node[vertex, label=below:$v_2$] (v2) at (2,0) {};
      \node[vertex, label=below:$\hdots$] (van) at (3,0) {};
      \node[vertex, label=below:$v_3$] (v3) at (4,0) {};
      \node[vertex, label=below:$v_4$] (v4) at (5.5,0) {};
      \node[slot, label=above:$s_j$] (si) at (1,2) {};
      \node[slot,fill=black, label=above:$s_i$] (sj) at (2.5,2) {};
      \node[slot, label=above:$s_k$] (sk) at (4,2) {};
      \draw[thick] (sj) edge (v1)
		   (sj) edge (v2);
      \draw[blue,thick] (6.5,0) -- (1.8,2);
      \draw[blue,thick] (4.5,0) -- (1.8,2);
      \draw[red,thick] (6.2,0) -- (3.2,2);
      \draw[red,thick] (7,0) -- (3.2,2);
      \draw[decorate,decoration={brace,amplitude=5pt,mirror}] (0.5,-0.8) -- (3.9,-0.8) node [black,midway,yshift=-0.5cm] {\footnotesize $A$};
      \draw[decorate,decoration={brace,amplitude=5pt,mirror}] (4.1,-0.8) -- (4.9,-0.8) node [black,midway,yshift=-0.5cm] {\footnotesize $B$};
      \draw[decorate,decoration={brace,amplitude=5pt,mirror}] (5.1,-0.8) -- (6.5,-0.8) node [black,midway,yshift=-0.5cm] {\footnotesize $C$};
      \draw[decorate,decoration={brace,amplitude=5pt}] (1,2.8) -- (2.4,2.8) node [black,midway,yshift=0.5cm] {\footnotesize $X$};
      \draw[decorate,decoration={brace,amplitude=5pt}] (2.6,2.8) -- (3.9,2.8) node [black,midway,yshift=0.5cm] {\footnotesize $Y$};
  \end{tikzpicture}
  }
   \caption{If we have a request with one endpoint in $C$ and one in $B$ placed in $X$, a 4-0 request was already present in the graph.
   This means that we had a situation like Figure~\ref{fig:40exception} already with respect to that placement, and we either have
   a situation  like Figure~\ref{fig:40exception} with respect to $v_3$ and $v_4$ too (left picture) or we have a forced 4-0 crossing 
   (right picture), contradicting the assumption of Lemma~\ref{lem:preferno40}.
   }
   \label{fig:double40}
  \end{center}
 \end{figure}

  We thus only care for slots in $Y$ with at least one endpoint in $C$. 
  If a slot in $Y$ with one endpoint in $C$ has the other endpoint in $A$, the number of edge crossings will be higher for the placement in $s_j$ already.
  Moreover, there can not be a slot in $Y$ with two edges directed to $C$, as in the case for slots in $X$, this would contradict the assumption that we are not in the case of a forced 4-0 crossing,
  as depicted with two red edges in the third drawing of Figure~\ref{fig:preferno40}.
  We are only left with one case, if there is a fulfilled request in $Y$ with a vertex in $B$ and a vertex in $C$,
  as depicted in the fourth drawing of Figure~\ref{fig:preferno40}.
  This type of request generates two extra crossings for the placement in $s_k$ with respect to the placement in $s_j$.
  This is still not a problem if there is only one such request, as these crossings would still be offset by the 
  4 extra crossings of the
  placement in $s_j$. Moreover, if there is more than one such request we are in the case of Figure~\ref{fig:40exception},
  where a 4-0 placement is allowed.
  
  Finally, we are left to count propagation arrow crossings.
   \begin{figure}[h]
 \begin{center}
 \subfigure{
   \begin{tikzpicture}[node distance=1cm,
   slot/.style={draw,rectangle},
   vertex/.style={draw,circle},
   scale=0.8,every node/.style={scale=0.8}
   ]
      \node[vertex, label=below:$v_1$] (v1) at (1,0) {};
      \node[vertex, label=below:$v_2$] (v2) at (2,0) {};
      \node[vertex, label=below:$\hdots$] (van) at (3,0) {};
      \node[vertex, label=below:$v_3$] (v3) at (4,0) {};
      \node[vertex, label=below:$v_4$] (v4) at (5,0) {};
      \node[slot, label=above:$s_j$] (si) at (1,2) {};
      \node[slot,fill=black, label=above:$s_i$] (sj) at (2.5,2) {};
      \node[slot, label=above:$s_k$] (sk) at (4,2) {};
      \node at (6,0) {$\hdots$};
      \draw[thick] (sj) edge (v1)
		   (sj) edge (v2);
      \draw[thick,blue, ->] (0.5,0) -- (si);
       \draw[thick,blue, ->]                   (v3) edge (si)
                         (v4) edge (sk)
                         (5.5,0) -- (sk);
      \draw[decorate,decoration={brace,amplitude=5pt,mirror}] (0.5,-0.8) -- (3.9,-0.8) node [black,midway,yshift=-0.5cm] {\footnotesize $A$};
      \draw[decorate,decoration={brace,amplitude=5pt,mirror}] (4.1,-0.8) -- (4.9,-0.8) node [black,midway,yshift=-0.5cm] {\footnotesize $B$};
      \draw[decorate,decoration={brace,amplitude=5pt,mirror}] (5.1,-0.8) -- (6.5,-0.8) node [black,midway,yshift=-0.5cm] {\footnotesize $C$};
      \draw[decorate,decoration={brace,amplitude=5pt}] (1,2.8) -- (2.4,2.8) node [black,midway,yshift=0.5cm] {\footnotesize $X$};
      \draw[decorate,decoration={brace,amplitude=5pt}] (2.6,2.8) -- (3.9,2.8) node [black,midway,yshift=0.5cm] {\footnotesize $Y$};
  \end{tikzpicture}
  }\hspace*{0cm}
  \subfigure{
  \begin{tikzpicture}[node distance=1cm,
  slot/.style={draw,rectangle},
  vertex/.style={draw,circle},
  scale=0.8,every node/.style={scale=0.8}
  ]
      \node[vertex, label=below:$v_1$] (v1) at (1,0) {};
      \node[vertex, label=below:$v_2$] (v2) at (2,0) {};
      \node[vertex, label=below:$\hdots$] (van) at (3,0) {};
      \node[vertex, label=below:$v_3$] (v3) at (4,0) {};
      \node[vertex, label=below:$v_4$] (v4) at (5,0) {};
      \node[slot, label=above:$s_j$] (si) at (1,2) {};
      \node[slot,fill=black, label=above:$s_i$] (sj) at (2.5,2) {};
      \node[slot, label=above:$s_k$] (sk) at (4,2) {};
      \node at (6,0) {$\hdots$};
      \draw[thick] (sj) edge (v1)
		   (sj) edge (v2);
      \draw[thick,red] (si) edge (v3)
		   (si) edge (v4);   
      \draw[thick,blue, ->] (0.5,0) -- (sk);
       \draw[thick,blue, ->] (5.5,0) -- (sk);
      \draw[decorate,decoration={brace,amplitude=5pt,mirror}] (0.5,-0.8) -- (3.9,-0.8) node [black,midway,yshift=-0.5cm] {\footnotesize $A$};
      \draw[decorate,decoration={brace,amplitude=5pt,mirror}] (4.1,-0.8) -- (4.9,-0.8) node [black,midway,yshift=-0.5cm] {\footnotesize $B$};
      \draw[decorate,decoration={brace,amplitude=5pt,mirror}] (5.1,-0.8) -- (6.5,-0.8) node [black,midway,yshift=-0.5cm] {\footnotesize $C$};
      \draw[decorate,decoration={brace,amplitude=5pt}] (1,2.8) -- (2.4,2.8) node [black,midway,yshift=0.5cm] {\footnotesize $X$};
      \draw[decorate,decoration={brace,amplitude=5pt}] (2.6,2.8) -- (3.9,2.8) node [black,midway,yshift=0.5cm] {\footnotesize $Y$};
  \end{tikzpicture}
  }
  \hspace*{0cm}
  \subfigure{
  \begin{tikzpicture}[node distance=1cm,
  slot/.style={draw,rectangle},
  vertex/.style={draw,circle},
  scale=0.8,every node/.style={scale=0.8}
  ]
      \node[vertex, label=below:$v_1$] (v1) at (1,0) {};
      \node[vertex, label=below:$v_2$] (v2) at (2,0) {};
      \node[vertex, label=below:$\hdots$] (van) at (3,0) {};
      \node[vertex, label=below:$v_3$] (v3) at (4,0) {};
      \node[vertex, label=below:$v_4$] (v4) at (5,0) {};
      \node[slot, label=above:$s_j$] (si) at (1,2) {};
      \node[slot,fill=black, label=above:$s_i$] (sj) at (2.5,2) {};
      \node[slot, label=above:$s_k$] (sk) at (4,2) {};
      \node at (6,0) {$\hdots$};
      \draw[thick] (sj) edge (v1)
		   (sj) edge (v2);
      \draw[thick,red] (sk) edge (v3)
		   (sk) edge (v4);
      \draw[thick,blue, ->] (0.5,0) -- (si);
       \draw[thick,blue, ->] (5.5,0) -- (si);
      \draw[decorate,decoration={brace,amplitude=5pt,mirror}] (0.5,-0.8) -- (3.9,-0.8) node [black,midway,yshift=-0.5cm] {\footnotesize $A$};
      \draw[decorate,decoration={brace,amplitude=5pt,mirror}] (4.1,-0.8) -- (4.9,-0.8) node [black,midway,yshift=-0.5cm] {\footnotesize $B$};
      \draw[decorate,decoration={brace,amplitude=5pt,mirror}] (5.1,-0.8) -- (6.5,-0.8) node [black,midway,yshift=-0.5cm] {\footnotesize $C$};
      \draw[decorate,decoration={brace,amplitude=5pt}] (1,2.8) -- (2.4,2.8) node [black,midway,yshift=0.5cm] {\footnotesize $X$};
      \draw[decorate,decoration={brace,amplitude=5pt}] (2.6,2.8) -- (3.9,2.8) node [black,midway,yshift=0.5cm] {\footnotesize $Y$};
  \end{tikzpicture}
  }
   \caption{Only one propagation arrow might cross $s_i$ due to Lemma~\ref{lem:PropagationCrossEdge}, and at the rightmost it comes from $v_3$.
   }
   \label{fig:preferno40arrows}
  \end{center}
 \end{figure}
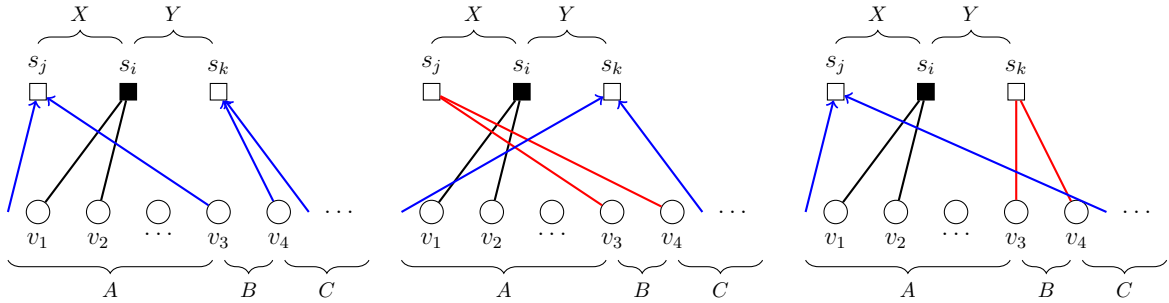 
  As depicted in Figure~\ref{fig:preferno40arrows}, the rightmost placement of the propagation arrows has the arrow from $v_3$ pointing to $s_i$
  and only the leftmost arrow from $C$ pointing to $s_k$. In the second and third pictures we see what happens to these arrows after a possible $s_j$ and
  $s_k$ placement. The number of crossings due to the propagation arrows stays the same. If the propagation arrows would be more to the left, the number of
  crossings in the 4-0 placement would only possibly increase, and the number of crossings for the $s_k$ placement would only possibly decrease. This means that if there is at most one slot in $Y$ with an endpoint in $B$ and an endpoint in $C$, 
  a placement in $s_k$ is prefered.
 \end{proof}

We now prove that the 4-0 crossings described in Lemma~\ref{lem:preferno40},
also have uniquely identifiable unavoidable crossings, just as we did in 
Lemma~\ref{lem:forced40-stacking} for the forced 4-0 crossings.
 
 \begin{lemma}\label{lem:noforced40UnavoidableCrossing}
  Any 4-0 crossing incurred by Algorithm~\ref{alg:minallcrossings},
  because any alternative placement would result in two 3-1
  crossings as sketched in Figure~\ref{fig:40exception}, 
  has two uniquely identifiable unavoidable crossings. 
 \end{lemma}

 \begin{proof}
  Observe, that if we only consider the crossings generated by the placement 
  of the request generating the 4-0 crossings we do not risk double counting
  unavoidable crossings in this case.
  In a configuration like depicted in Figure~\ref{fig:40exception}, where Algorithm~\ref{alg:minallcrossings}
  generates a 4-0 crossing, an
  optimal algorithm can place the same requests as depicted in the right side of Figure~\ref{fig:40exceptionCR}. 
  The placement of the new request with a 4-0 crossing by Algorithm~\ref{alg:minallcrossings}
  has 6 crossings with previously placed requests (Figure~\ref{fig:40exceptionCR} left) while the optimal placement for this request has only 2 crossings with previously placed requests 
  (Figure~\ref{fig:40exceptionCR} right). 
  These crossings are unavoidable. 
 \end{proof}

\begin{figure}[h]
	\begin{center}
		\subfigure{
			\begin{tikzpicture}[node distance=1cm,
			slot/.style={draw,rectangle},
			vertex/.style={draw,circle},
			scale=0.8,every node/.style={scale=0.8}
			]
			\node[vertex, label=below:$v_1$] (v1) at (1,0) {};
			\node[vertex, label=below:$v_2$] (v2) at (2,0) {};
			\node[vertex, label=below:$\hdots$] (van) at (3,0) {};
			\node[vertex, label=below:$v_3$] (v3) at (4,0) {};
			\node[vertex, label=below:$v_4$] (v4) at (5,0) {};
			\node[slot, fill=black] (si) at (1,2) {};
			\node[slot,fill=black] (sj) at (2.5,2) {};
			\node[slot] (sk) at (4,2) {};
			\node at (6,0) {$\hdots$};
			\draw[thick] (sj) edge (v1)
			(sj) edge (v2);
			\draw[thick] (si) edge (v3)
			(si) edge (v4);
			\draw[red,thick] (6,0) -- (2.8,2);
			\draw[red,thick] (4.5,0) -- (2.8,2);
			\draw[red,thick] (6,0) -- (3.2,2);
			\draw[red,thick] (4.5,0) -- (3.2,2);
			\end{tikzpicture}}
		\subfigure{
			\begin{tikzpicture}[node distance=1cm,
			slot/.style={draw,rectangle},
			vertex/.style={draw,circle},
			scale=0.8,every node/.style={scale=0.8}
			]
			\node[vertex, label=below:$v_1$] (v1) at (1,0) {};
			\node[vertex, label=below:$v_2$] (v2) at (2,0) {};
			\node[vertex, label=below:$\hdots$] (van) at (3,0) {};
			\node[vertex, label=below:$v_3$] (v3) at (4,0) {};
			\node[vertex, label=below:$v_4$] (v4) at (5,0) {};
			\node[slot] (si) at (1,2) {};
			\node[slot,fill=black] (sj) at (2.5,2) {};
			\node[slot, fill=black] (sred1) at (1.8,2) {};
			\node[slot,fill=black] (sred2) at (3.2,2) {};
			\node[slot, fill=black] (sk) at (4,2) {};
			\node at (6,0) {$\hdots$};
			\draw[thick] (sred1) edge (v1)
			(sred1) edge (v2);
			\draw[thick] (sj) edge (v3)
			(sj) edge (v4);
			\draw[red,thick] (6,0) -- (sred2);
			\draw[red,thick] (4.5,0) -- (sred2);
			\draw[red,thick] (6,0) -- (sk);
			\draw[red,thick] (4.5,0) -- (sk);
			\end{tikzpicture}}
		\caption{Algorithm \ref{alg:minallcrossings} has generated a 4-0 crossing. We also depict the optimal configuration for such a situation}
		\label{fig:40exceptionCR}
	\end{center}
\end{figure}
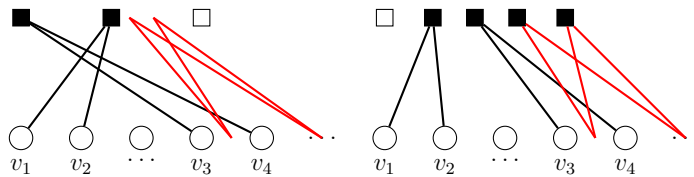 

 We can finally conclude, using \Cref{lem:noforced40UnavoidableCrossing,lem:forced40-stacking},
 that any 4-0 crossings incurred by Algorithm~\ref{alg:minallcrossings} have at least one unavoidable crossing.
 
\begin{theorem}\label{the:4-0Crossings}
Forced and non-forced 4-0 crossings incurred by Algorithm~\ref{alg:minallcrossings} 
have at least one unavoidable crossing.
\end{theorem}

\subsection{The 3-0 Crossings}

It remains to prove that Algorithm~\ref{alg:minallcrossings} only generates a 3-0 crossing -- depicted in 
\Cref{fig:node_exchange_uncritical} (c) -- if there is at least one unavoidable 
crossing for one of the two requests that are responsible for the 3-0 crossing. 
In general the proofs use case distinction in a similar way 
to the proofs from the previous section, handling the 4-0 crossings. 

Similarly as in the 4-0 case, we start by proving that Algorithm~\ref{alg:minallcrossings} 
never produces a 3-0 crossing with a gap.

 \begin{lemma}\label{lem:30gaps}
  Algorithm~\ref{alg:minallcrossings} never generates 3-0 crossings with gaps in between. More precisely,
  for each pair $s_j$, $s_i$ assigned by Algorithm~\ref{alg:minallcrossings} with $j< i$ 
  that generate a 3-0 crossing,
  every slot $s_k$ with $j< k< i$ is already full.
 \end{lemma}

 \begin{proof}
  Let us assume that $v_1$ and $v_2$ is the pair of vertices adjacent to a filled slot $s_i$. 
  Let $v_2$ and $v_3$ be a pair of vertices from a new request, with $v_1 < v_2 < v_3$. 
  If Algorithm~\ref{alg:minallcrossings} creates a 3-0 crossing between the requests $\{v_1, v_2\}$ and $\{v_2,v_3\}$, it places the second request in a slot $s_j$ with $s_j < s_i$. 

  We can assume that $s_j$ is the rightmost available slot to the left of $s_i$, 
  due to the following observations. 
  If there exists a slot $s_i^\prime$ between $s_j$ and $s_i$, 
  we observe that the propagation arrow of $v_2$ cannot point to $s_j$, 
  when the request $\{v_2,v_3\}$ arrives. 
  Because then, the two propagation arrows pointing to $s_i^\prime$ 
  have to start at vertices to the right of $v_2$ and cross the edges of the request 
  $\{v_1, v_2\}$ which violates \Cref{lem:PropagationCrossEdge}. 
  This means that the propagation arrow of $v_2$ must point to $s_i'$ or to a slot to the right
  of $s_i$.
  But in this case, when the request $\{v_2,v_3\}$ is placed in $s_j$ it pushes the propagation arrows 
  that pointed to $s_j$, which must come from vertices to the left of $v_2$, to a slot to the right of $s_j$,
  in this case $s_i'$, crossing the edges of the newly placed request and
  violating, again, \Cref{lem:PropagationCrossEdge}. 
  Thus, $s_j$ must be the rightmost available slot to the left of $s_i$. 
 \end{proof}

 In the following lemma we explore the situation that the 3-0 crossing happens at the edge of the graph,
 that is, a placement on any remaining slot causes a 3-0 crossing.
 
 \begin{lemma}\label{lem:30edge}
  Given two requests $\{v_1, v_2\}$ and $\{v_2,v_3\}$ with $v_1<v_2<v_3$.
  Assume without loss of generality that Algorithm~\ref{alg:minallcrossings}
  creates a 3-0 crossing between
  these requests, with
  the first request for vertices $\{v_1,v_2\}$ being placed in slot $s_i$ and
  during the placement of the second request there is no available slot $s_k>s_i$.
  Then there is at least one uniquely identifiable unavoidable crossing with 
  the request $\{v_2,v_3\}$.
 \end{lemma}

\begin{proof}
 Let us assume that $v_1$ and $v_2$ is the pair of vertices adjacent to a filled slot $s_i$. 
  Let $v_2$ and $v_3$ be a pair of vertices from a new request, with $v_1 < v_2 < v_3$. 
  If Algorithm~\ref{alg:minallcrossings} creates a 3-0 crossing between the requests $\{v_1, v_2\}$ and $\{v_2,v_3\}$, it places the second request in a slot $s_j$ with $s_j < s_i$. 
  Recall that by Lemma~\ref{lem:30gaps}, $s_j$ must be the rightmost available slot to the left of $s_i$. 
  
  If there is no free slot $s_k$ with $s_k > s_i$, then all slots to the right of $s_j$ are filled. 
  Moreover the propagation arrows from $v_2$ and $v_3$ are 
  the two right most propagation arrows and both point to $s_j$, when the request $\{v_2,v_3\}$ arrives. 
  Let the number of fulfilled slots to the right of $s_j$ be $t$ and the number of vertices that are to the right of $v_2$ be $b$. 
  If the number of fulfilled slots is larger than the number of vertices on the bottom line ($t > b$), 
  there are at least two edges from vertices that are to the left of $v_2$ pointing to fulfilled slots 
  that are to the right of $s_j$. 
  It is not possible that a fulfilled slot right of $s_j$ is adjacent to two vertices that are to the left of $v_2$ due to \Cref{lem:PropagationCrossEdge}, because the propagation arrows of $v_2$ and $v_3$ would completely cross it. 
  Thus, the second adjacent vertex must be to the right of $v_2$, resulting in at least one unavoidable crossing for the request $\{v_2,v_3\}$. 
  If $t \leq b$, the slot $s_j$ is above the vertex $v_2$ or to the right of it. 
  In this case, the edge $\{v_1,s_i\}$ must be compensated by an edge that starts to the right of $v_2$ and points to a slot to the left of $s_j$. 
  It is not possible that the second edge of this slot comes from a vertex to the right of $v_2$, too (see \Cref{lem:Equator} or \Cref{lem:40gaps}). 
  Thus, when no slot $s_k$ exists, there must exist an edge that unavoidably crosses the request $\{v_2,v_3\}$. 
\end{proof}

What remains is an exhaustive case distinction analogous to the analysis done for the 4-0 crossings.
 
\begin{theorem}\label{the:3-0Crossings}
If Algorithm~\ref{alg:minallcrossings} creates a 3-0 crossing between
two requests $\{v_1, v_2\}$ and $\{v_2,v_3\}$, there is at least one
uniquely identifiable unavoidable crossing 
for at least one of the two requests. 
\end{theorem}

\begin{proof}
  Let us assume that $v_1$ and $v_2$ is the pair of vertices adjacent to a filled slot $s_i$. 
  Let $v_2$ and $v_3$ be a pair of vertices from a new request, with $v_1 < v_2 < v_3$. 
  If Algorithm~\ref{alg:minallcrossings} creates a 3-0 crossing between the requests $\{v_1, v_2\}$ and $\{v_2,v_3\}$, it places the second request in a slot $s_j$ with $s_j < s_i$. 
  Recall that, by Lemma~\ref{lem:30gaps}, $s_j$ must be the rightmost available slot to the left of $s_i$. 
  Moreover, the case where there is no available slot $s_i < s_k$ is already covered by Lemma~\ref{lem:30edge}.
  
  Thus, in the following we assume our algorithm can choose between the slot $s_j$, resulting in a 3-0 crossing, and the slot $s_k$ which does not generate a 3-0 crossing ($s_j < s_i < s_k$). 
  This situation is depicted in \Cref{fig:preferno30} (a). 
  We look at the cases where \Cref{alg:minallcrossings} prefers a placement on $s_j$. 
  In the following we investigate which edges must exits to make the slot $s_j$ more preferable. 
  
  To count the crossings we divide the relevant slots into two subsets. 
  The subset $X$ contains the slots between $s_j$ and $s_i$ and the subset $Y$ contains the slots between $s_i$ and $s_k$. 
  We do not need to consider the area to the left of $s_j$ or to the right of $s_k$, because the number of crossings with edges incident to a slot in one of these areas is independent of the position of the request $\{v_2, v_3\}$. 
  We also divide the vertices on the bottom line into three different subsets. 
  The vertices to the left of $v_2$ form the set $A$. 
  The vertices between $v_2$ and $v_3$ are in the set $B$ and the vertex $v_3$ 
  and all vertices to its right form the set $C$.  
  
  \Cref{fig:preferno30} (b) shows which edges or propagation arrows cross 
  the new request only if placed in the the slot $s_j$ (in green) 
  and the ones that cross the new request only if placed in slot $s_k$ (in red). 
  An edge that is incident to a vertex between $v_2$ and $v_3$ (area $B$) does not favor a particular slot, because it crosses the request $\{v_2,v_3\}$ once, independent of its placement. 

  There must be edges or propagation arrows in our graph such that, 
  avoiding the 3-0 crossing results in at least three crossings in order to make a placement in 
  $s_j$ favorable. 
  Thus, requests like in \Cref{fig:preferno30} (c) or (d) 
  must be present in our graph in order to make the 3-0 crossing a feasible choice. 
  
  But, if a vertex between $v_2$ and $v_3$ exists (area $B$ is not empty), 
  its edges unavoidably cross the request $\{v_2, v_3\}$ two times, 
  fulfilling the statement of our lemma. 
  So, in the following, we can assume that there is no vertex between $v_2$ and $v_3$. 
  Therefore, the edges that make the 3-0 crossing more favorable must be incident to vertices from $C$. 
  The corresponding slot for these edges can be either in $X$ or $Y$.
  
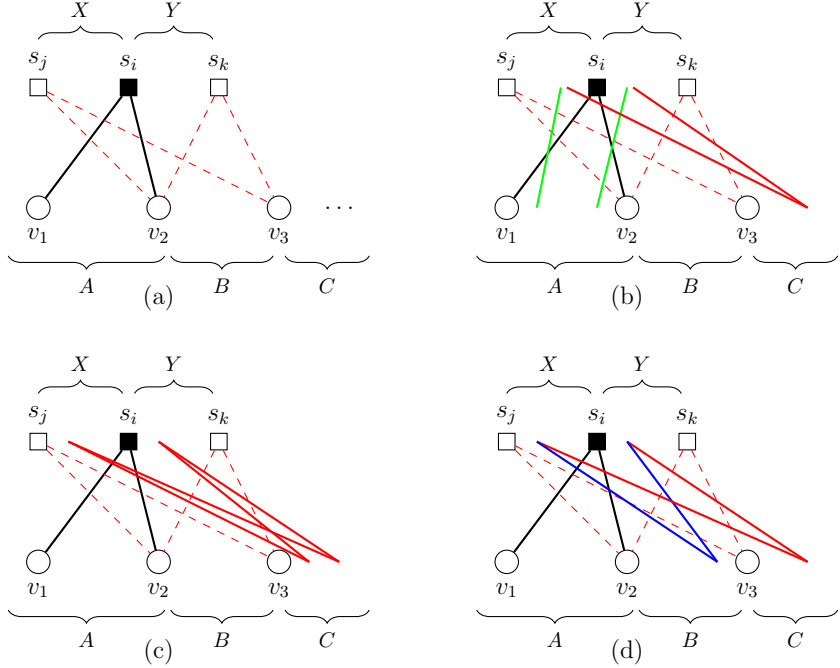
\begin{figure}[t]
	\centering
		\subfigure{
			\begin{tikzpicture}[node distance=1cm,
			slot/.style={draw,rectangle},
			vertex/.style={draw,circle},
			scale=0.8,every node/.style={scale=0.8}
			]
			\node[vertex, label=below:$v_1$] (v1) at (1,0) {};
			\node[vertex, label=below:$v_2$] (v2) at (3,0) {};
			\node[vertex, label=below:$v_3$] (v3) at (5,0) {};
			\node[slot, label=above:$s_j$] (si) at (1,2) {};
			\node[slot,fill=black, label=above:$s_i$] (sj) at (2.5,2) {};
			\node[slot, label=above:$s_k$] (sk) at (4,2) {};
			\node at (6,0) {$\hdots$};
			
			\draw[thick] (sj) edge (v1)
			(sj) edge (v2);
			
			\draw[dashed, red] (si) edge (v3)
			(si) edge (v2)
			(sk) edge (v3)
			(sk) edge (v2);
			
			\draw[decorate,decoration={brace,amplitude=5pt,mirror}] (0.5,-0.8) -- (3.1,-0.8) node [black,midway,yshift=-0.5cm] {\footnotesize $A$};
			\draw[decorate,decoration={brace,amplitude=5pt,mirror}] (3.2,-0.8) -- (4.9,-0.8) node [black,midway,yshift=-0.5cm] {\footnotesize $B$};
			\draw[decorate,decoration={brace,amplitude=5pt,mirror}] (5.1,-0.8) -- (6.5,-0.8) node [black,midway,yshift=-0.5cm] {\footnotesize $C$};
			
			\draw[decorate,decoration={brace,amplitude=5pt}] (1,2.8) -- (2.4,2.8) node [black,midway,yshift=0.5cm] {\footnotesize $X$};
			\draw[decorate,decoration={brace,amplitude=5pt}] (2.6,2.8) -- (3.9,2.8) node [black,midway,yshift=0.5cm] {\footnotesize $Y$};

			\node[] (a) at (3,-1.5) {(a)};
			\end{tikzpicture}
		}\hspace*{1cm}
		\subfigure{
			\begin{tikzpicture}[node distance=1cm,
			slot/.style={draw,rectangle},
			vertex/.style={draw,circle},
			scale=0.8,every node/.style={scale=0.8}
			]
			\node[vertex, label=below:$v_1$] (v1) at (1,0) {};
			\node[vertex, label=below:$v_2$] (v2) at (3,0) {};
			\node[vertex, label=below:$v_3$] (v3) at (5,0) {};
			\node[slot, label=above:$s_j$] (si) at (1,2) {};
			\node[slot,fill=black, label=above:$s_i$] (sj) at (2.5,2) {};
			\node[slot, label=above:$s_k$] (sk) at (4,2) {};
			
			\draw[thick] (sj) edge (v1)
			(sj) edge (v2);
			
			\draw[dashed, red] (si) edge (v3)
			(si) edge (v2)
			(sk) edge (v3)
			(sk) edge (v2);
			
			\draw[red,thick] (6,0) -- (3.1,2);
			\draw[green,thick] (2.5,0) -- (3,2);
			\draw[red,thick] (6,0) -- (2,2);
			\draw[green,thick] (1.5,0) -- (1.9,2);
			
			\draw[decorate,decoration={brace,amplitude=5pt,mirror}] (0.5,-0.8) -- (3.1,-0.8) node [black,midway,yshift=-0.5cm] {\footnotesize $A$};
			\draw[decorate,decoration={brace,amplitude=5pt,mirror}] (3.2,-0.8) -- (4.9,-0.8) node [black,midway,yshift=-0.5cm] {\footnotesize $B$};
			\draw[decorate,decoration={brace,amplitude=5pt,mirror}] (5.1,-0.8) -- (6.5,-0.8) node [black,midway,yshift=-0.5cm] {\footnotesize $C$};
			
			\draw[decorate,decoration={brace,amplitude=5pt}] (1,2.8) -- (2.4,2.8) node [black,midway,yshift=0.5cm] {\footnotesize $X$};
			\draw[decorate,decoration={brace,amplitude=5pt}] (2.6,2.8) -- (3.9,2.8) node [black,midway,yshift=0.5cm] {\footnotesize $Y$};
			
			\node[] (b) at (3,-1.5) {(b)};
			\end{tikzpicture}
		}\\
		\subfigure{
			\begin{tikzpicture}[node distance=1cm,
			slot/.style={draw,rectangle},
			vertex/.style={draw,circle},
			scale=0.8,every node/.style={scale=0.8}
			]
			\node[vertex, label=below:$v_1$] (v1) at (1,0) {};
			\node[vertex, label=below:$v_2$] (v2) at (3,0) {};
			\node[vertex, label=below:$v_3$] (v3) at (5,0) {};
			
			\node[slot, label=above:$s_j$] (si) at (1,2) {};
			\node[slot,fill=black, label=above:$s_i$] (sj) at (2.5,2) {};
			\node[slot, label=above:$s_k$] (sk) at (4,2) {};
			
			\draw[thick] (sj) edge (v1)
			(sj) edge (v2);
			\draw[dashed, red] (si) edge (v3)
			(si) edge (v2)
			(sk) edge (v3)
			(sk) edge (v2);
			
			\draw[red,thick] (6,0) -- (3,2);
			\draw[red,thick] (5.5,0) -- (3,2);
			
			\draw[red,thick] (6,0) -- (1.5,2);
			\draw[red,thick] (5.5,0) -- (1.5,2);
			
			\draw[decorate,decoration={brace,amplitude=5pt,mirror}] (0.5,-0.8) -- (3.1,-0.8) node [black,midway,yshift=-0.5cm] {\footnotesize $A$};
			\draw[decorate,decoration={brace,amplitude=5pt,mirror}] (3.2,-0.8) -- (4.9,-0.8) node [black,midway,yshift=-0.5cm] {\footnotesize $B$};
			\draw[decorate,decoration={brace,amplitude=5pt,mirror}] (5.1,-0.8) -- (6.5,-0.8) node [black,midway,yshift=-0.5cm] {\footnotesize $C$};
			
			\draw[decorate,decoration={brace,amplitude=5pt}] (1,2.8) -- (2.4,2.8) node [black,midway,yshift=0.5cm] {\footnotesize $X$};
			\draw[decorate,decoration={brace,amplitude=5pt}] (2.6,2.8) -- (3.9,2.8) node [black,midway,yshift=0.5cm] {\footnotesize $Y$};

			\node[] (c) at (3,-1.5) {(c)};
			\end{tikzpicture}
		}\hspace*{1cm}
		\subfigure{
			\begin{tikzpicture}[node distance=1cm,
			slot/.style={draw,rectangle},
			vertex/.style={draw,circle},
			scale=0.8,every node/.style={scale=0.8}
			]
			\node[vertex, label=below:$v_1$] (v1) at (1,0) {};
			\node[vertex, label=below:$v_2$] (v2) at (3,0) {};
			\node[vertex, label=below:$v_3$] (v3) at (5,0) {};
			
			\node[slot, label=above:$s_j$] (si) at (1,2) {};
			\node[slot,fill=black, label=above:$s_i$] (sj) at (2.5,2) {};
			\node[slot, label=above:$s_k$] (sk) at (4,2) {};
			
			\draw[thick] (sj) edge (v1)
			(sj) edge (v2);
			\draw[dashed, red] (si) edge (v3)
			(si) edge (v2)
			(sk) edge (v3)
			(sk) edge (v2);
			
			\draw[red,thick] (6,0) -- (3,2);
			\draw[blue,thick] (4.5,0) -- (3,2);
			
			\draw[red,thick] (6,0) -- (1.5,2);
			\draw[blue,thick] (4.5,0) -- (1.5,2);
			
			\draw[decorate,decoration={brace,amplitude=5pt,mirror}] (0.5,-0.8) -- (3.1,-0.8) node [black,midway,yshift=-0.5cm] {\footnotesize $A$};
			\draw[decorate,decoration={brace,amplitude=5pt,mirror}] (3.2,-0.8) -- (4.9,-0.8) node [black,midway,yshift=-0.5cm] {\footnotesize $B$};
			\draw[decorate,decoration={brace,amplitude=5pt,mirror}] (5.1,-0.8) -- (6.5,-0.8) node [black,midway,yshift=-0.5cm] {\footnotesize $C$};
			
			\draw[decorate,decoration={brace,amplitude=5pt}] (1,2.8) -- (2.4,2.8) node [black,midway,yshift=0.5cm] {\footnotesize $X$};
			\draw[decorate,decoration={brace,amplitude=5pt}] (2.6,2.8) -- (3.9,2.8) node [black,midway,yshift=0.5cm] {\footnotesize $Y$};

			\node[] (d) at (3,-1.5) {(d)};
			\end{tikzpicture}
		}
		\caption{			
			We have two possible placements for the request $v_2$, $v_3$. 
			Red edges contribute extra crossings to the placement in $s_k$ and green edges contribute extra crossings to the placement in $s_j$. 
			The blue edge cross with one edge of the request independent of its placement. 
		}
		\label{fig:preferno30}
\end{figure}
We start with the case that there is a filled slot $s_x \in X$ 
with both edges incident to the set $C$ and analyze it in more detail. 
  Assume that the request $\{v_2, v_3\}$ arrives in the time step $t$. 
  Thus, at the end of time step $t-1$ the slots $s_i$ and $s_x$ are filled. 
  If the propagation arrows of $v_2$ and $v_3$ point to two different slots ($s_j$ and $s_k$), we can apply \Cref{lem:UnavoidableHouses} and know that eventually there will be two unavoidable crossings
  between the request $\{v_2, v_3\}$ and the overarching request. 
  To see this more clearly, we point out explicitly how we can apply \Cref{lem:UnavoidableHouses}: 
  The request $\{v_2,v_3\}$ is the request $\{u,v\}$, the free slots $s_j$ and $s_k$ are the slots $s_l$ and $s_r$, the filled slots $s_i$ and $s_x \in X$ correspond to the slots $s_x$ and $s_y$. 
  
  If the propagation arrows of $v_2$ and $v_3$ point to the same slot \Cref{lem:PropagationCrossEdge} prohibits that the propagation arrows point towards $s_k$.
  Therefore, the only possible configuration at the end of time step $t-1$ is that both propagation arrows already point towards $s_j$. 
  This indicates that the placement of a previous request from a time step $t^\prime < t$ lead
  to this situation. 

  In the following we look at the time step $t^\prime$, in which the request arrived that pushed the propagation arrows - starting at $v_2$ and $v_3$ - to the slot $s_j$ for the last time. 
  Note that the request $\{v_1, v_2\}$, placed in $s_i$, cannot push the propagation arrows to $s_j$, because positioning it onto $s_j$ instead of $s_i$ creates fewer crossings. 
  To be precise, the three crossings between the propagation arrows of $v_2$ and $v_3$ and the edges of the request $\{v_1, v_2\}$ can be avoided by placing it in $s_j$ instead of $s_i$. This means, that when
  the request $\{v_1, v_2\}$ arrived, the propagation arrow configuration must have been different.
  Thus, the slot $s_i$ must be filled and there must be a different request that is responsible for the last push of the propagation arrows in time step $t^\prime$. 
  
  So, at the beginning of time step $t^\prime$, at least one of the propagation arrows of $v_2$ and $v_3$
  points to a free slot left of $s_j$ or to a free slot right of $s_j$. 
  But, if the propagation arrow of $v_2$ points to a free slot left of $s_j$, 
  there must be another propagation arrow, starting to the right of $v_2$, that points to $s_j$. 
  Thus, there are two propagation arrows, starting at a vertex to the right of $v_2$ 
  that cross the edges incident to $s_i$, violating \Cref{lem:PropagationCrossEdge}. 
  Thus, the request in time step $t^\prime$ cannot push the propagation arrows from left to right;
  it must push the propagation arrows from a free slot to the right of $s_j$ to the left, onto $s_j$. 
  Note, that this implies that all of the slots $s_x\in X$ are also already filled in the time step $t^\prime$, 
  because otherwise propagation arrows, coming from the right of $v_2$, would point to $s_x\in X$ 
  and cross both edges incident to $s_i$, violating again \Cref{lem:PropagationCrossEdge}. 
  
  To push the propagation arrows to the left, it is necessary that at least one vertex of the request is to the left of $v_3$. 
  And because there is no vertex between $v_2$ and $v_3$, it must also be to the left of $v_2$. 
  We differentiate between two different cases. 
  The other vertex of this request can be to the left of $v_2$ or to the right of $v_3$. 
  
  If the other vertex of the request is also to the left of $v_2$, 
  the propagation arrows of $v_2$ and $v_3$ cross both edges of this request at the end of $t^\prime$, 
  violating \Cref{lem:PropagationCrossEdge}. 
  If the second vertex is to the right of $v_3$ 
  we have a request that unavoidably crosses all edges of $v_2$ and $v_3$. 
  Thus, for every feasible configuration for the time step $t^\prime$, we have unavoidable crossings. 
%
%
%
%
%
	
  Now, we consider the case that there is a slot $s_y \in Y$ with two incident edges to vertices from $C$ and analyze it in more detail.  
  Assume that the request $\{v_2,v_3\}$ arrives in the time step $t$. 
  Thus, at the end of time step $t-1$ the slot $s_i$ and $s_y$ are filled. 
  We differentiate between two cases. 
  Either the propagation arrows of $v_2$ and $v_3$
  point to different slots, $s_j$ and $s_k$ respectively, or they point to the same slot. 

  If they point to different slots, we can apply \Cref{lem:UnavoidableHouses} just as in the previous
  case and know that there will be two unavoidable crossings between the edges from the request
  $\{v_2,v_3\}$ and the future overarching request. 
    
  In the following, we assume that the propagation arrows from $v_2$ and $v_3$
  point to the same slot, $s_j$ or $s_k$. 
  Note, if both point to $s_k$, the vertex $v_3$ must be adjacent to $s_y$,
  otherwise \Cref{lem:PropagationCrossEdge} is violated. 
  But in this case, the configuration in which the propagation arrows point to the same slot becomes symmetric. 
  Thus, in the following, we look w.l.o.g. at the case that both propagation arrows point to $s_j$. 
  
  Just as in the previous case, we look at the last time step 
  in which the propagation arrows are pushed to the slot $s_j$ and call it $t^\prime$. 
  At the start of time step $t^\prime$, the request $\{v_1, v_2\}$ must already be placed in the slot $s_i$,
  by the same argument as before. 
  It is, again, also not possible that the last time the propagation arrows are pushed is from left to right: 
  If the propagation arrow of $v_2$ points to a slot left of $s_j$, the propagation arrows pointing to 
  $s_j$ start at a vertex to the right of $v_2$ and violate \Cref{lem:PropagationCrossEdge}. 
  Thus, at least one propagation arrow from $v_3$ is pointing to a slot right of $s_j$, 
  at the start of $t^\prime$. 
  To push the propagation arrows to the left, 
  it is necessary that at least one vertex of the request is to the left of $v_3$.  
  The other vertex can be to the left of $v_2$ too or to the right of $v_3$ again. 
    
  If the other vertex of the request is also to the left of $v_2$, by the same argument as before,
  the propagation arrows of $v_2$ and $v_3$ cross both edges at the end of $t^\prime$, 
  violating \Cref{lem:PropagationCrossEdge}. 
  If the second vertex is to the right of $v_3$ it unavoidably crosses two times with the request $\{v_2, v_3\}$.
  
  Thus, we have proven, finally that for every feasible configuration leading to a 3-0 crossing there must be,
  at least by the end of the request sequence, two unavoidable crossings with the edges of the request 
  $\{v_2,v_3\}$.
\end{proof}

This theorem shows that 3-0 crossings incurred by Algorithm~\ref{alg:minallcrossings} only happen in conjunction
with two extra unavoidable crossings with the request generating the 3-0 crossing, this means, that any 3-0 
crossing is in effect a 5-2 crossing, which would be better than $5$-competitive.

\subsection{The Upper Bound}
 We can finally put all results together to conclude with an upper bound for the competitive ratio of
 Algorithm~\ref{alg:minallcrossings} to solve online slotted OSCM-$2$ on $2$-regular graphs.
 
\begin{theorem}\label{thm:FinalResult}
\Cref{alg:minallcrossings} solves the online slotted OSCM-$2$ on $2$-regular graphs with a competitive ratio of at most $5$.
\end{theorem}
\begin{proof}
In order to calculate the competitive ratio of Algorithm~\ref{alg:minallcrossings}
we simply compare for every pair of requests, what the optimal placement compared
to the placement chosen by Algorithm~\ref{alg:minallcrossings} would be.

We exhaustively look at possible placements of pairs of requests, as depicted in Figure~\ref{fig:node_exchange_uncritical}.
Observe, that except for the 3-0 crossings
and 4-0 crossings, the rest of possible request pairs are no worse than 3-competitive regardless of
the algorithm used. Moreover, Theorem~\ref{the:4-0Crossings} ensures that for every 4-0 crossing
incurred by Algorithm~\ref{alg:minallcrossings}, there is at least one uniquely identifiable unavoidable
crossing, meaning that the number of crossings incurred by Algorithm~\ref{alg:minallcrossings} is 5, but 
optimally there must be at least 1 unavoidable crossing. Finally, Theorem~\ref{the:3-0Crossings} guarantees that there
are also two uniquely identifiable unavoidable crossings for every occurrence of a 3-0 crossing. Thus,
Algorithm~\ref{alg:minallcrossings} is at most 5-competitive.
\end{proof}
  
    \section{Conclusion}
    In this work we have shown that the general slotted OSCM-$k$ is
    not competitive for any $k\geq 2$, which led us to analyze the
    case of the slotted OSCM-$k$ on $2$-regular graphs. On this graph
    class, we have given a  construction which proved a lower bound on
    the competitive ratio of $4/3$.
    Algorithm~\ref{alg:minallcrossings}, which utilizes the
    information of the remaining space and unavoidable crossings in
    the graph in the form of our so-called \emph{propagation arrows},
    was proven to be at most $5$-competitive. This was done by limiting the
    number of total crossings generated by pairs of requests that do not
    cross one another in an optimal solution.

    There are several open questions which we were not able to answer
    in the scope of this work. First, there is still a considerable
    gap between the lower and upper bound of the competitive ratio
    that we have given. We assume that
    Algorithm~\ref{alg:minallcrossings} performs better than analyzed
    and that the upper bound can be made tighter.

    While Theorem~\ref{thm:oscmunrestricted} proves
    non-competitiveness on general graphs for any $k \geq 2$, the case
    of regular graphs with degree $3$ or higher is still open. 
    We suggest to analyze this graph class further.

\bibliography{bibliography}

\begin{thebibliography}{10}

\bibitem{DBLP:journals/comgeo/BattistaETT94}
Giuseppe~Di Battista, Peter Eades, Roberto Tamassia, and Ioannis~G. Tollis.
\newblock Algorithms for drawing graphs: an annotated bibliography.
\newblock {\em Comput. Geom.}, 4:235--282, 1994.

\bibitem{DBLP:books/daglib/0097013}
Allan Borodin and Ran El{-}Yaniv.
\newblock {\em Online computation and competitive analysis}.
\newblock Cambridge University Press, 1998.

\bibitem{DBLP:journals/algorithmica/DujmovicW04}
Vida Dujmovic and Sue Whitesides.
\newblock An efficient fixed parameter tractable algorithm for 1-sided crossing
  minimization.
\newblock {\em Algorithmica}, 40(1):15--31, 2004.

\bibitem{DBLP:journals/algorithmica/EadesW94}
Peter Eades and Nicholas~C. Wormald.
\newblock Edge crossings in drawings of bipartite graphs.
\newblock {\em Algorithmica}, 11(4):379--403, 1994.

\bibitem{DBLP:journals/tvcg/FrishmanT08}
Yaniv Frishman and Ayellet Tal.
\newblock Online dynamic graph drawing.
\newblock {\em {IEEE} Trans. Vis. Comput. Graph.}, 14(4):727--740, 2008.

\bibitem{garey1983crossing}
Michael~R Garey and David~S Johnson.
\newblock Crossing number is np-complete.
\newblock {\em SIAM Journal on Algebraic Discrete Methods}, 4(3):312--316,
  1983.

\bibitem{DBLP:journals/algorithmica/KobayashiT15}
Yasuaki Kobayashi and Hisao Tamaki.
\newblock A fast and simple subexponential fixed parameter algorithm for
  one-sided crossing minimization.
\newblock {\em Algorithmica}, 72(3):778--790, 2015.

\bibitem{DBLP:series/txtcs/Komm16}
Dennis Komm.
\newblock {\em An Introduction to Online Computation - Determinism,
  Randomization, Advice}.
\newblock Texts in Theoretical Computer Science. An {EATCS} Series. Springer,
  2016.

\bibitem{DBLP:journals/ipl/LiS02}
Xiao~Yu Li and Matthias F.~M. Stallmann.
\newblock New bounds on the barycenter heuristic for bipartite graph drawing.
\newblock {\em Inf. Process. Lett.}, 82(6):293--298, 2002.

\bibitem{walter}
Xavier Mu{\~{n}}oz, Walter Unger, and Imrich Vrto.
\newblock One sided crossing minimization is np-hard for sparse graphs.
\newblock In {\em Graph Drawing, 9th International Symposium, {GD} 2001 Vienna,
  Austria, September 23-26, 2001, Revised Papers}, pages 115--123, 2001.

\bibitem{DBLP:journals/dcg/Nagamochi05}
Hiroshi Nagamochi.
\newblock An improved bound on the one-sided minimum crossing number in
  two-layered drawings.
\newblock {\em Discret. Comput. Geom.}, 33(4):569--591, 2005.

\bibitem{DBLP:journals/tcs/Nagamochi05}
Hiroshi Nagamochi.
\newblock On the one-sided crossing minimization in a bipartite graph with
  large degrees.
\newblock {\em Theor. Comput. Sci.}, 332(1-3):417--446, 2005.

\bibitem{DBLP:conf/gd/NorthW01}
Stephen~C. North and Gordon Woodhull.
\newblock Online hierarchical graph drawing.
\newblock In Petra Mutzel, Michael J{\"{u}}nger, and Sebastian Leipert,
  editors, {\em Graph Drawing, 9th International Symposium, {GD} 2001 Vienna,
  Austria, September 23-26, 2001, Revised Papers}, volume 2265 of {\em Lecture
  Notes in Computer Science}, pages 232--246. Springer, 2001.

\bibitem{schaefer2012graph}
Marcus Schaefer.
\newblock The graph crossing number and its variants: A survey.
\newblock {\em The electronic journal of combinatorics}, 2012.

\bibitem{shannon2007considerations}
Ross Shannon and Aaron~J Quigley.
\newblock Considerations in dynamic graph drawing: A survey.
\newblock {\em Comput. Sci. Informatics}, (June), 2007.

\bibitem{DBLP:journals/cacm/SleatorT85}
Daniel~Dominic Sleator and Robert~Endre Tarjan.
\newblock Amortized efficiency of list update and paging rules.
\newblock {\em Commun. {ACM}}, 28(2):202--208, 1985.

\end{thebibliography}
\bibliographystyle{plain}
\end{document}